\documentclass[12pt, draftclsnofoot,onecolumn]{IEEEtran}
\ifCLASSINFOpdf

\else

\fi
\usepackage{subeqnarray}
\usepackage{mathrsfs}
\usepackage{amsfonts}
\usepackage{amsmath}
\usepackage{amsthm}
\usepackage{amssymb}
\usepackage{amstext}
\usepackage{graphicx}
\usepackage{indentfirst}
\usepackage{array}
\usepackage{cite}
\usepackage[numbers,sort&compress]{natbib}
\usepackage{enumerate}
\usepackage{stmaryrd}
\usepackage[linesnumbered,ruled,vlined]{algorithm2e}
\usepackage{multirow}
\usepackage{multicol}
\usepackage{verbatim}
\usepackage{stfloats}
\usepackage{color}
\usepackage{bm}
\usepackage{multirow}
\usepackage{multicol}

\newtheorem{theorem}{\bf{Theorem}}
\newtheorem{lemma}{\bf{Lemma}}

\newtheorem{corollary}{\bf{Corollary}}

\newtheorem{definition}{\bf{Definition}}
\newtheorem{remark}{\bf{Remark}}
\newtheorem{example}{\bf{Example}}

\usepackage{hyperref}

\makeatletter
\def\endIEEEproof{\hspace*{\fill}~\hollowsquare\par}
\makeatother

\newcommand{\hollowsquare}{\text{$\square$}}

\begin{document}
%
\title{A Mathematical Theory of Semantic Communication}

\author{Kai Niu, \IEEEmembership{Member,~IEEE}, Ping Zhang,  \IEEEmembership{Fellow,~IEEE}
\thanks
{
This work is supported by the National Natural Science Foundation of China (No. 62293481, 62071058).\protect\\
\indent K. Niu is with the Key Laboratory of Universal Wireless Communications, Ministry of Education,
and P. Zhang is with the State Key Laboratory of Networking and Switching Technology, Beijing University of Posts and Telecommunications,  Beijing, 100876, China (e-mail: \{niukai, pzhang\}@bupt.edu.cn). \protect\\
}}

\maketitle

\begin{abstract}
The year 1948 witnessed the historic moment of the birth of classic information theory (CIT). Guided by CIT, modern communication techniques have approached the theoretic limitations, such as, entropy function $H(U)$, channel capacity $C=\max_{p(x)}I(X;Y)$ and rate-distortion function $R(D)=\min_{p(\hat{x}|x):\mathbb{E}d(x,\hat{x})\leq D} I(X;\hat{X})$. Semantic communication paves a new direction for future communication techniques whereas the guided theory is missed. In this paper, we try to establish a systematic framework of semantic information theory (SIT). We investigate the behavior of semantic communication and find that synonym is the basic feature so we define the synonymous mapping between semantic information and syntactic information. Stemming from this core concept, synonymous mapping $f$, we introduce the measures of semantic information, such as semantic entropy $H_s(\tilde{U})$, up/down semantic mutual information $I^s(\tilde{X};\tilde{Y})$ $(I_s(\tilde{X};\tilde{Y}))$, semantic capacity $C_s=\max_{f_{xy}}\max_{p(x)}I^s(\tilde{X};\tilde{Y})$, and semantic rate-distortion function $R_s(D)=\min_{\{f_x,f_{\hat{x}}\}}\min_{p(\hat{x}|x):\mathbb{E}d_s(\tilde{x},\hat{\tilde{x}})\leq D}I_s(\tilde{X};\hat{\tilde{X}})$. Furthermore, we prove three coding theorems of SIT by using random coding and (jointly) typical decoding/encoding, that is, the semantic source coding theorem, semantic channel coding theorem, and semantic rate-distortion coding theorem. We find that the limits of SIT are extended by using synonymous mapping, that is, $H_s(\tilde{U})\leq H(U)$, $C_s\geq C$ and $R_s(D)\leq R(D)$. All these works composite the basis of semantic information theory. In addition, we discuss the semantic information measures in the continuous case. Especially, for the band-limited Gaussian channel, we obtain a new channel capacity formula, $C_s=B\log\left[S^4\left(1+\frac{P}{N_0B}\right)\right]$, where the average synonymous length $S$ indicates the identification ability of information. In summary, the theoretic framework of SIT proposed in this paper is a natural extension of CIT and may reveal great performance potential for future communication.

\end{abstract}


\IEEEpeerreviewmaketitle

\section{Introduction}
\label{section_I}

\IEEEPARstart{C}{l}assic information theory (CIT), established by C. E. Shannon \cite{Classicpaper_Shannon} in 1948, was a great achievement in the modern information and communication field. As shown in Fig. \ref{model_classic_comm}, the classic communication system includes source, encoder, channel with noise, decoder and destination. This theory is concerned with the uncertainty of information and introduces four critical measures, such as entropy, mutual information, channel capacity, and rate-distortion function to evaluate the performance of information processing and transmission. Especially, three famous coding theorems, such as, lossless/lossy source coding theorem and channel coding theorem, reveal the fundamental limitation of data compression and information transmission. Over the past 70 years or so, people developed many advanced techniques to approach these theoretical limits. For the lossless source coding, Huffman coding and arithmetic coding are the representative optimal coding methods can achieve the source entropy. Similarly, for the channel coding, polar code, as a great breakthrough \cite{Polarcode_Arikan}, is the first constructive capacity-achieving coding scheme. Correspondingly, for the lossy source coding, some modern coding schemes, such as BPG (Better Portable Graphics) standard and H. 265/266 standard, can approach the rate-distortion lower bounds of image and video sources. It follows that information and communication technologies guided by CIT have approached the theoretical limitation and the performance improvement of modern communication systems encounters a lot of bottlenecks.

\begin{figure*}[htbp]
\setlength{\abovecaptionskip}{0.cm}
\setlength{\belowcaptionskip}{-0.cm}
  \centering{\includegraphics[scale=1]{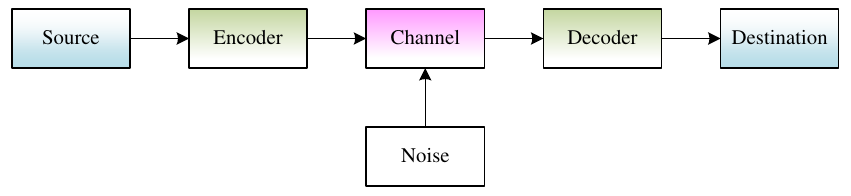}}
  \caption{The block diagram of classic communication system.}\label{model_classic_comm}
\end{figure*}

Essentially, Weaver \cite{Semantic_Weaver}, just one year after Shannon published the seminal paper on information theory, pointed out that communication involves problems at three levels as follows:
\begin{enumerate}[]
  \item  ``\textbf{LEVEL A}. How accurately can the symbols of communication be transmitted? (The technical problem.)
  \item  \textbf{LEVEL B}. How precisely do the transmitted symbols convey the desired meaning? (The semantic problem.)
  \item  \textbf{LEVEL C}. How effectively does the received meaning affect conduct in the desired way? (The effectiveness problem.)"
\end{enumerate}

Shannon \cite{Classicpaper_Shannon} wrote that ``semantic aspects of communication are irrelevant to the engineering problem". Thus, classic information theory only handles LEVEL A (technical) problem of the information. On the contrary, date back to the time of classic information theory birth, many works also focused on LEVEL B problem and the semantic communication theory. Carnap and Bar-Hillel \cite{Semantic_Carnap} and Floridi \cite{Semantic_Floridi} considered using propositional logic sentences to express semantic information. They introduced the semantic information entropy, which is calculated based on logical probability \cite{Logic_Nilsson} rather than statistical probability in CIT. Then Bao \emph{et al.} \cite{Semantic_Bao} extended this theoretical framework and derived the semantic source coding and semantic channel coding theorem based on propositional logic probabilities. On the other hand, De Luca \emph{et al.} \cite{Entropy_Luca} \cite{Fuzzy_Luca} regarded semantic information as fuzzy variable and defined fuzzy entropy to measure the uncertainty of semantic information. Then Wu \cite{Wuweiling} extended this work and introduced general entropy, general conditional entropy, and general mutual information based on fuzzy variable. However, the propositional logic or fuzzy variables are only suitable for the simple processing of text or speech source and can not sufficiently describe semantic information of the complex data, such as image or video source. Furthermore, although some recent works \cite{RateD_Liu, SideInfo_Guo, Theory_Shao, Theory_Tang} investigated the theoretic property of semantic information, e.g. rate-distortion function,  designing semantic communication system still lacks of systematic guiding theory.

Recently, semantic communication systems based on deep learning demonstrate excellent performance than the traditional counterparts. Many works \cite{Semantic_ZhangPing, Survey_Shi, Survey_Qin, Survey_Xie} investigated the design principles and technical challenges of semantic communication. As surveyed in \cite{Survey_Gunduz, Semantic_Niu}, semantic communication techniques become a hot topic in the communication community and provide a promising methodology to break through the Shannon limits. However, semantic communication research faces a dilemma. We can neither precisely answer what is the semantic information or the meaning of information nor provide the fundamental limits to guide the design of the semantic communication system. Thus, there is an urgent necessity to establish a mathematical theory of semantic communication in order to solve these basic problems.

If we want to establish a semantic information theory, we should firstly consider the essence of semantic information, that is, what is the meaning of information. Let's investigate the semantic information from the source side and the destination side respectively. Figure \ref{source_synonym_example} shows some examples of semantic information of text, speech, image and video source. In Fig. \ref{source_synonym_example}(a), the words ``happy, joyful, content, joyous" have the same or similar meaning, that is, all these text data means ``happy". From the viewpoint of linguistics, these words compose the synonym of ``happy". From the understanding of human beings, we can regard these words as having the same semantic information. So we can use synonym alternation to generate the same meaning presentation, as shown in the following sentences
\begin{enumerate}[]
  \item  ``She appeared \textbf{happy} and content after receiving the good news".
  \item  ``She appeared \textbf{joyful} and content after receiving the good news".
\end{enumerate}
Although these two sentences have different presentations, the processing of the language center of the brain, they have the same semantic information. Generally, such synonym phenomenon ubiquitously exists in various language texts. Many different presentations of phrases, words, and sentences have the same or similar meaning and compose the synonymous mapping to indicate the same semantic information. Therefore, we conclude that synonymous mapping is a critical feature of semantic information of text data.

\begin{figure*}[htbp]
\setlength{\abovecaptionskip}{0.cm}
\setlength{\belowcaptionskip}{-0.cm}
  \centering{\includegraphics[scale=0.7]{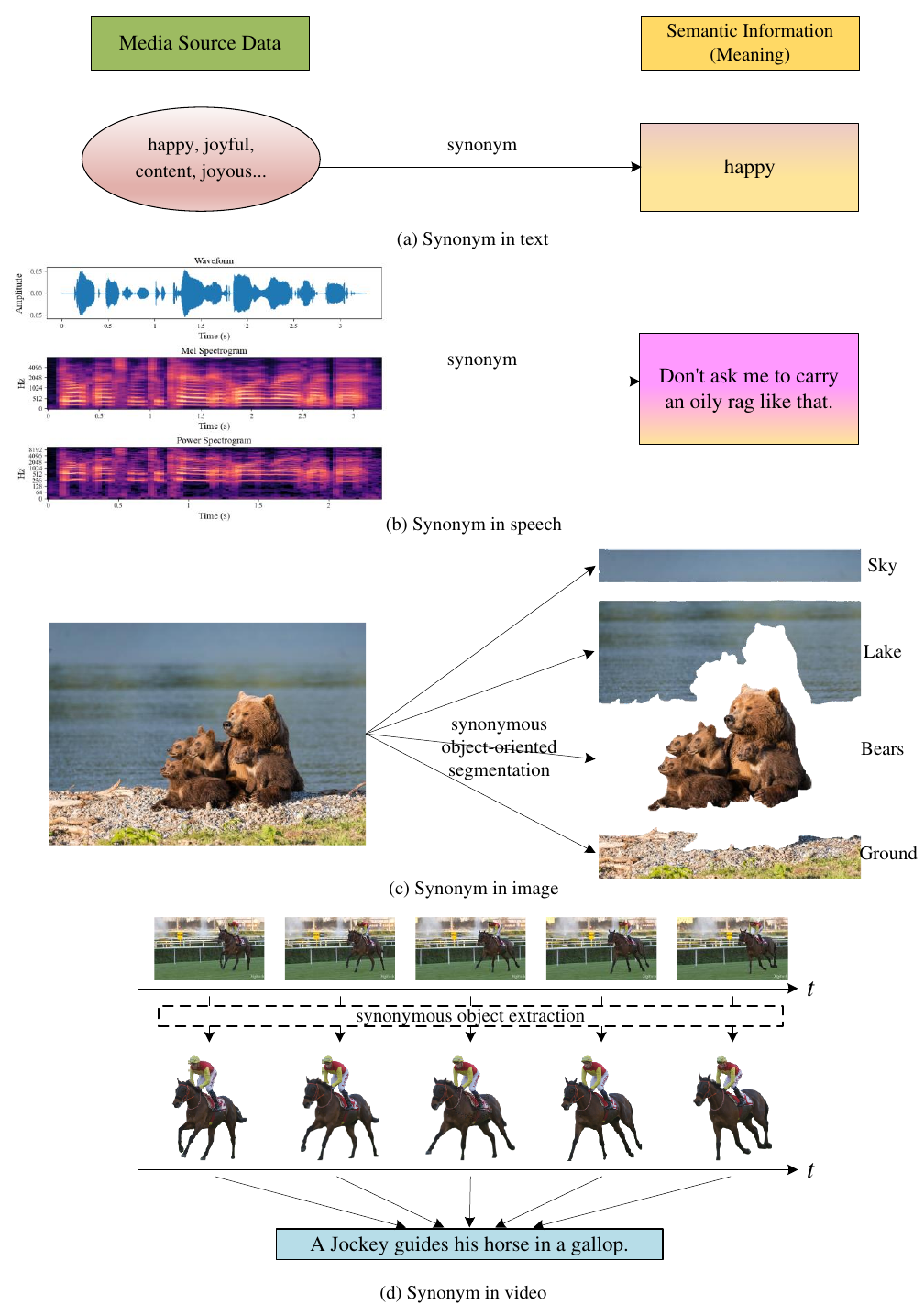}}
  \caption{Semantic examples of source data.}\label{source_synonym_example}
\end{figure*}

Similarly, synonym phenomenon can also be observed in speech source. As shown in Fig. \ref{source_synonym_example}(b), a piece of speech has three presentations, that is, waveform, Mel spectrogram, and power spectrogram. However, all these presentations indicate the same meaning ``Don't ask me to carry an oily rag like that". So we conclude that all these speech waveform or spectrograms compose the synonym presentation and have the same semantic information.

We can also find the synonym objects in image source, as depicted in Fig. \ref{source_synonym_example}(c). After the object-oriented segmentation, this figure is decomposed into four objects, that is, sky, lake, bears, and ground. Intuitively, single pixel in this figure have no meaning whereas a set of many pixels indicates some meaningful object. Furthermore, for the frame sequence shown in Fig. \ref{source_synonym_example}(d), by using the object extraction method, we can obtain the meaning of this video ``A jockey guides his horse in a gallop". From the two examples, it follows that synonymous mapping for the semantic understanding is a popular phenomenon in image and video source.

On the other hand, in the destination side, we also observe the synonym phenomenon for various downstream tasks. Figure \ref{downstreamtask_synonym_example} depicts some representative examples. For the task of character recognition in Fig. \ref{downstreamtask_synonym_example}(a), different shapes and fonts of the images in each row present the same letter. So we easily derive the meaning from these images, that is, ``SEMANTIC". Similarly, for the task of image classification, shown in Fig. \ref{downstreamtask_synonym_example}(b), various images stands for the same entity or object, such as forest, staircase, ocean etc. Furthermore, Fig. \ref{downstreamtask_synonym_example}(c) depicts the task of obstacle detection and the marker boxes presents the pedestrians. Hence, we can conclude that many downstream tasks involve the synonymous mapping and semantic reasoning.

\begin{figure*}[htbp]
\setlength{\abovecaptionskip}{0.cm}
\setlength{\belowcaptionskip}{-0.cm}
  \centering{\includegraphics[scale=0.7]{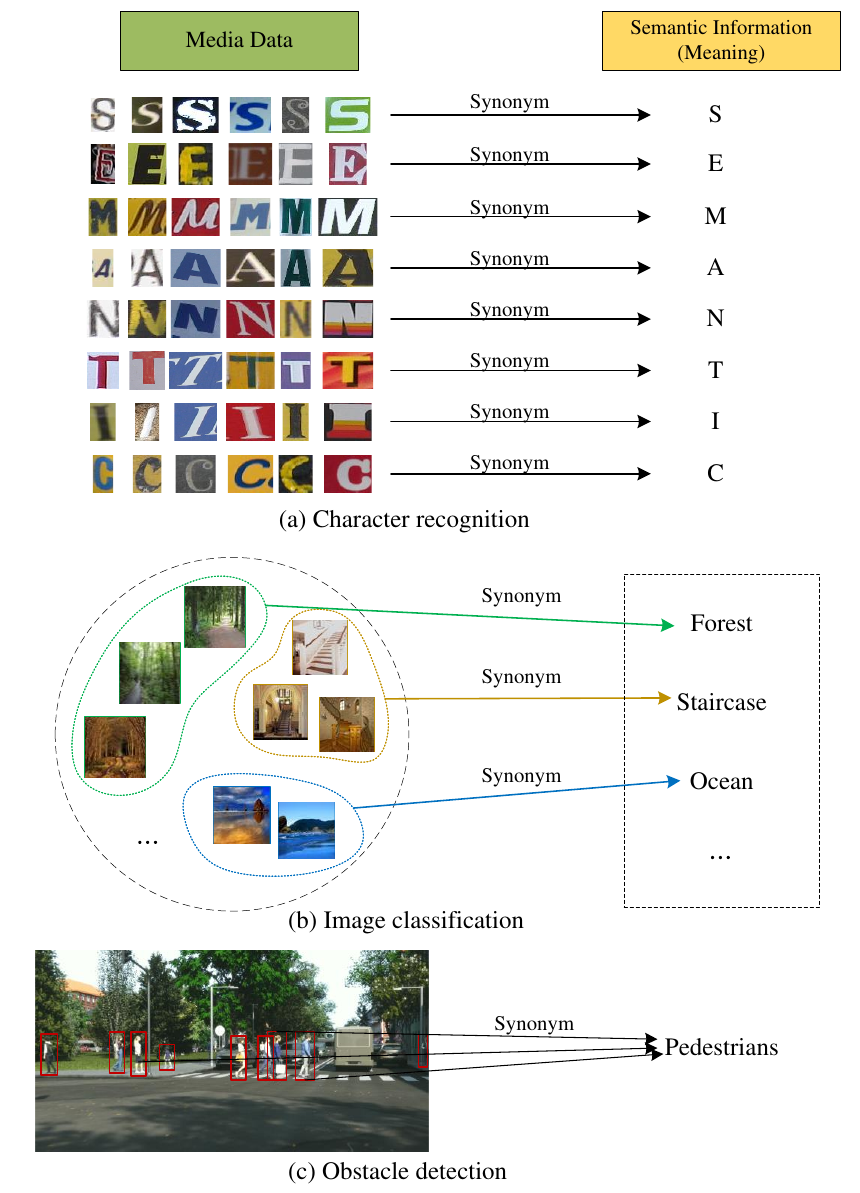}}
  \caption{Semantic examples of downstream task.}\label{downstreamtask_synonym_example}
\end{figure*}

\begin{remark}
In a word, after inspecting various examples of source and destination, we can summarize these two rules for the semantic information processing.

\textbf{(1)} All the perceptible messages, such as text, speech, image, video, and so on, are syntactic information. However, based on these messages, we can derive or reason some semantic information. There are common and stable mappings between the syntactic information and the semantic information. These mappings can be built from the knowledge of human beings or the configuration conditions of the application scene.

\textbf{(2)} Generally, the relationships between syntactic information and semantic information are very complex. However, in most instances, the semantic information has a single meaning whereas the presentations of source data and downstream tasks are myriads. So synonymous mapping is a major relationship and popularly exists in various tasks of semantic reasoning. Certainly, there may exist the ambiguity of semantic information. Nonetheless, such ambiguity is secondary and can be removed by using multiple interactions. So in this paper, we mainly handle the synonym of semantic information and the ambiguity will be left to future works.
\end{remark}

In this paper, we aim to establish a systematic framework of semantic information theory as the natural extension of classic information theory. The outline of the theoretic framework is shown in Fig. \ref{SIT_framework}.

\begin{figure*}[htbp]
\setlength{\abovecaptionskip}{0.cm}
\setlength{\belowcaptionskip}{-0.cm}
  \centering{\includegraphics[scale=0.9]{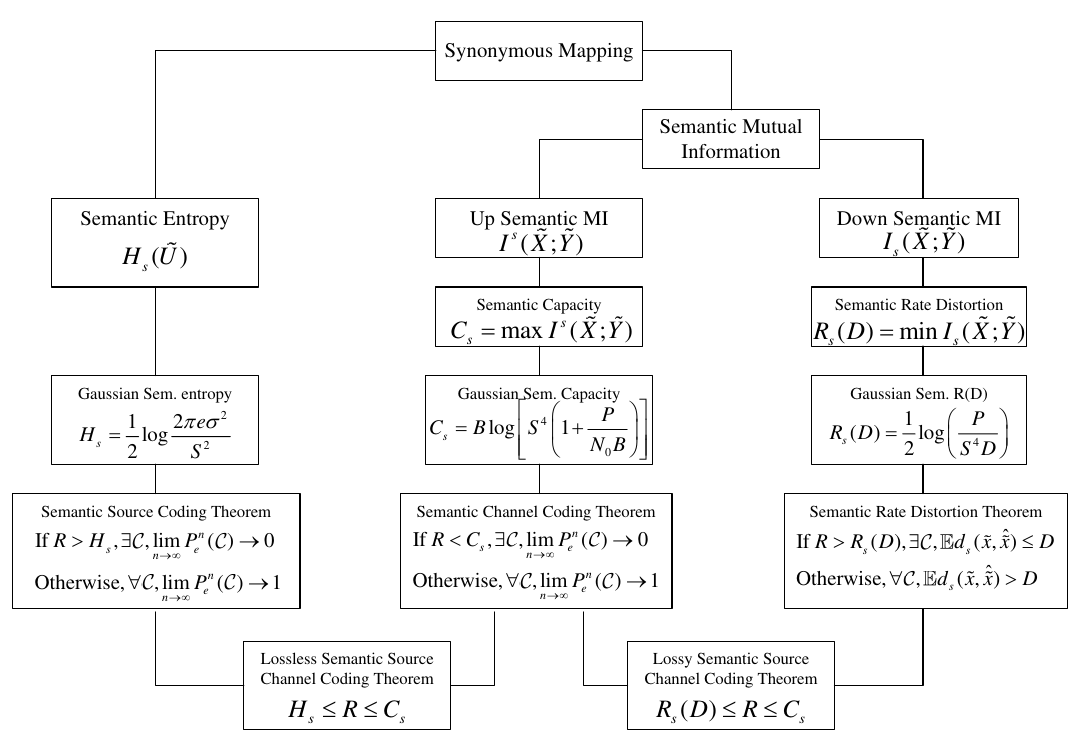}}
  \caption{The theoretic framework of semantic information.}\label{SIT_framework}
\end{figure*}

Therefore, the contributions of this paper can be summarized as follows.
\begin{enumerate}[(1)]
  \item We develop a systematic model for semantic communication with specific design criteria. In this model, semantic information remains invisible but perceptible, characterized by its prevalent synonymous features. Hence, we recognize synonymy as the fundamental aspect of semantic information. To illustrate the connection between semantic and syntactic information, we introduce synonymous mapping which is an one-to-many mapping from the semantic alphabet to syntactic alphabet. Essentially, synonymous mapping constructs an equivalent class relationship between the semantic space and the syntactic space.

  \item Stemming for a synonymous mapping $f_u$, we introduce the semantic variable $\tilde{U}$ associated with a random variable $U$. Essentially, the semantic variable is also a random variable. Hence, we define the semantic entropy $H_s(\tilde{U})$ to measure the uncertainty of semantic information and name the unit as semantic bit (\textbf{sebit}). Analogous to the classic information theory, this measure can also be extended to semantic conditional/joint entropy. In addition, we prove that all the semantic entropies are no more than their classic counterparts.

  \item We introduce the concepts of semantic relative entropy and semantic mutual information. Unlike the classic mutual information, it is noted that we use two measures to indicate semantic mutual information, such as the up semantic mutual information $I^s(\tilde{X};\tilde{Y})=H(X)+H(Y)-H_s(\tilde{X},\tilde{Y})$ and the down semantic mutual information $I_s(\tilde{X};\tilde{Y})=H_s(\tilde{X})+H_s(\tilde{Y})-H(X,Y)$. Furthermore, given a channel with the transition probability $p(y|x)$, we can use the semantic capacity $C_s=\max_{f_{xy}}\max_{p(x)}I^s(\tilde{X};\tilde{Y})$ to indicate the maximum transmission rate of semantic information on the channel. Correspondingly, given a source with the distribution $p(x)$ and the average distortion measure constraint $\mathbb{E}\left[d_s(\tilde{x},\hat{\tilde{x}})\right]\leq D$, the semantic rate distortion $R_s(D)=\min_{\{f_x,f_{\tilde{x}}\}}\min_{p(\hat{x}|x): \mathbb{E}d_s(\tilde{x},\hat{\tilde{x}})\leq D}I_s(\tilde{X};\hat{\tilde{X}})$ reveals the minimum compression rate of semantic information in the case of lossy source coding.

  \item We investigate the asymptotic equipartition property (AEP) in the semantic sense and introduce the semantically typical set $\tilde{A}_{\epsilon}^{(n)}$ which presents all the information of semantic variable $\tilde{U}$. Especially, we prove that, under the synonymous mapping, the syntactically typical set $A_{\epsilon}^{(n)}$ can be equipartitioned into many synonymous typical sets $B_{\epsilon}^{(n)}$. By using random coding and synonymous mapping between a semantic typical sequence to the synonymous typical set $B_{\epsilon}^{(n)}$, we prove the lossless source coding theorem, that is, given a syntactic source $U\sim p(u)$ and the associated semantic variable $\tilde{U}$, if the semantic code rate $R>H_s(\tilde{U})$, there exists a semantic source code satisfies $P_e^{(n)}\to 0$ with a sufficiently large code length $n$, otherwise, any codes cannot satisfy the lossless reconstruction of semantic information. Furthermore, we extend the Kraft inequality to the semantic version and provide an example of semantic Huffman coding. Thanks to the synonymous mapping and coding, the compression rate of semantic lossless source coding can be further lowered and the semantic compressive efficiency can outperform the classic source coding.

  \item Consider the problem of channel transmission, we investigate the jointly asymptotic equipartition property (JAEP) of semantic version and define the corresponding jointly typical set. Particularly, we find that, under the jointly synonymous mapping, the syntactically jointly typical set can be evenly divided into a series of jointly synonymous typical sets. By using random coding on jointly synonymous typical sets and jointly typical decoding, we prove the semantic channel coding theorem. This theorem reveals that given a channel with the transition probability $p(y|x)$, under the jointly synonymous mapping $f_{xy}$, if the code rate is lower than the semantic capacity, i.e., $R<C_s=\max_{f_{xy}}\max_{p(x)} I^s(\tilde{X},\tilde{Y})$, there exists a semantic channel code can satisfy the requirement of asymptotic error-free transmission, that is $P_e^{(n)}\to 0$ with sufficiently large $n$, on the contrary, if $R>C_s$, the error probability of any code can not tend to zero. Since the semantic capacity $C_s$ is no less than the classic capacity $C$, i.e., $C_s\geq C$, we conclude that semantic channel coding can improve the capability of information transmission. Furthermore, inspired by the jointly synonymous mapping, we consider the decoding rule of semantic channel code and propose a maximum likelihood group (MLG) decoding algorithm. Unlike traditional ML decoding, the MLG algorithm calculates all the group likelihood probabilities and selects one synonymous codeword of the group with the maximum likelihood probability as the final result. Hence, we define the group Hamming distance as the construction metric and derive the group-wise error probability to evaluate the performance of the semantic channel code.

  \item Consider the problem of lossy source coding, we define the semantic distortion and the jointly typical set of source and reconstruction sequence. By using jointly typical encoding based on synonymous typical set, we prove the semantic rate distortion coding theorem. This theorem states that given a source $X\sim p(x)$ with the associated semantic source $\tilde{X}$, the synonymous mappings $f_x,f_{\hat{x}}$, and the bounded semantic distortion function $d_s(\tilde{x},\hat{\tilde{x}})$, if the code rate $R>R_s(D)$, there exist a sequences of semantic source codes, the semantic distortion satisfies $\mathbb{E}d_s(\tilde{X},\hat{\tilde{X}})<D$ with sufficiently large $n$. Conversely, if $R<R_s(D)$, then the semantic distortion of any code meets $\mathbb{E}d_s(\tilde{X},\hat{\tilde{X}})>D$. Since the semantic rate distortion $R_s(D)$ is no more than the classic counterpart $R(D)$, that is, $R_s(D)\leq R(D)$, it follows that semantic source coding can further compress the source data and achieve a lower rate than the classic source coding.

  \item We also investigate the measure of semantic information in the continuous case. Given a continuous random variable $U\sim p(u)$ and a synonymous mapping $f$, we define the entropy of the associated semantic variable $\tilde{U}$ as $H_s(\tilde{U})=-\int p(u)\log p(u)du-\log S$, where $S$ is named as the average synonymous length. Furthermore, we extend the semantic conditional/joint entropy and semantic mutual information to the continuous case. Especially, we derive the semantic capacity of Gaussian channel $C_s=\frac{1}{2}\log \left[S^4\left(1+\frac{P}{\sigma^2}\right)\right]$ and a lower bound $\underline{C}_s=\frac{1}{2}\log\left(1+S^4\frac{P}{\sigma^2}\right)$ where $P$ is the signal power and $\sigma^2$ is the variance of Gaussian noise. In addition, we obtain the channel capacity formula of time-limited, band-limited and power-limited Gaussian channel, that is, $C_s=B\log \left[S^4\left(1+\frac{P}{N_0B}\right)\right]$ where $B$ is the bandwidth and $N_0$ is the single-sided power spectral density of white Gaussian noise. In addition, we also obtain the rate distortion function of Gaussian source, that is, $R_s(D)=\log \frac{P}{S^4D}$ where $P$ is the power of signal sample and $D$ is the constraint of semantic distortion.

  \item Finally, we inspect the source and channel coding problem in the semantic sense. Both for lossless and lossy cases, we prove the semantic source channel coding theorem. We find that the code rate of semantic communication system satisfies $R_s(D) (H_s(\tilde{U})) \leq R\leq C_s$. Compared with the classic communication system, the code rate range of sematic communication can be further extended. This point reveals the significant potential of semantic coding.

\end{enumerate}

The remainder of the paper is organized as follows. Section \ref{section_II} presents the systematic model of semantic communication and introduces the concept of synonymous mapping. In Section \ref{section_III}, we define the semantic entropy, the semantic joint and conditional entropy and discuss the basic properties of these measures. In Section \ref{section_IV} we define the semantic relative entropies and the up/down semantic mutual information and discuss the corresponding properties. Then in Section \ref{section_V}, the semantic channel capacity and the semantic rate distortion function are introduced and investigated. Furthermore, we explore the semantically asymptotic equipartition property and prove the semantic source coding theorem in Section \ref{section_VI}. We extend the jointly asymptotic equipartition property in the semantic sense. By using random coding and jointly typical decoding, we prove the semantic channel coding theorem in Section \ref{section_VII}. Correspondingly, By using jointly typical encoding, we prove the semantic rate distortion coding theorem in Section \ref{section_VIII}. In addition, we extend the semantic information measures to the continuous case in Section \ref{section_IX}. Especially, we derive the semantic capacity of Gaussian channel and the rate distortion of Gaussian source. In Section \ref{section_X}, we prove the semantic source channel coding theorem. Finally, Section \ref{section_XI} concludes the paper.

\section{Semantic Communication System and Synonymous mapping}
\label{section_II}
In this section, we first introduce the system model of semantic communication and describe the design criteria of semantic information transmission. Then we clarify the function of synonymous mapping and emphasize its key role in semantic communication.

\subsection{Notation Conventions}
In this paper, calligraphy letters, such as $\mathcal{X}$ and $\mathcal{Y}$, are mainly used to denote sets, and the cardinality of $\mathcal{X}$ is defined as $\left|\mathcal{X}\right|$. The Cartesian product of $\mathcal{X}$ and $\mathcal{Y}$ is written as $\mathcal{X}\times \mathcal{Y}$. Let $\mathcal{X}^n$ denote the $n$-th Cartesian power of $\mathcal{X}$ and $\prod_{k=1}^{n}\mathcal{X}_k$ denote the Cartesian product of $n$ sets $\mathcal{X}_1,\cdots,\mathcal{X}_n$. We write $u^n$ to denote an $n$-dimensional vector $\left(u_1,u_2,\cdots,u_n\right)$. We use the summation convention $N_{[i_1:i_m]}$ to denote the integer summation of $N_{i_1}+\cdots+N_{i_j}+\cdots+N_{i_m}$, where $\forall N_{i_j}\in \mathbb{N}$.

We use $f:\mathcal{X}\to\mathcal{Y}$ to denote a mapping from $\mathcal{X}$ to $\mathcal{Y}$. Furthermore. the extended mapping $f^n:\mathcal{X}^n\to\mathcal{Y}^n$ denotes an element-wise sequential mapping from $\mathcal{X}^n$ to $\mathcal{Y}^n$.

We use $d_{\text{H}}(u^n,v^n)$ to denote the Hamming distance between the binary vector $u^n$ and $v^n$. Given $\forall {a^n}, {b^n}\in \mathbb{R}^n$, let $\left\|a^n-b^n\right\|$ denote the Euclidian distance between the vector $a^n$ and $b^n$.

Throughout this paper, $\log\left(\cdot\right)$ means ``logarithm to base $2$'' and $\ln\left(\cdot\right)$ stands for the natural logarithm. Let $(x)^{+}=\max(x,0)$ be the non-negative part of $x$. Let $(\cdot)^T$ denote the transpose operation of the vector. Let $\mathbb{E}(Z)$ and ${\rm{Var}}(Z)$ denote the expectation and the variance of the random variable $Z$ respectively.

\subsection{Semantic Communication System}
The block diagram of semantic communication system is presented in Fig. \ref{Semantic_communication_system}. Compared with the system model of classic communication in Fig. \ref{model_classic_comm}, the semantic communication system extends the range of information processing and adds extra modules, such as semantic source, semantic destination, synonymous mapping and demapping.

\begin{figure*}[htbp]
\setlength{\abovecaptionskip}{0.cm}
\setlength{\belowcaptionskip}{-0.cm}
  \centering{\includegraphics[scale=1]{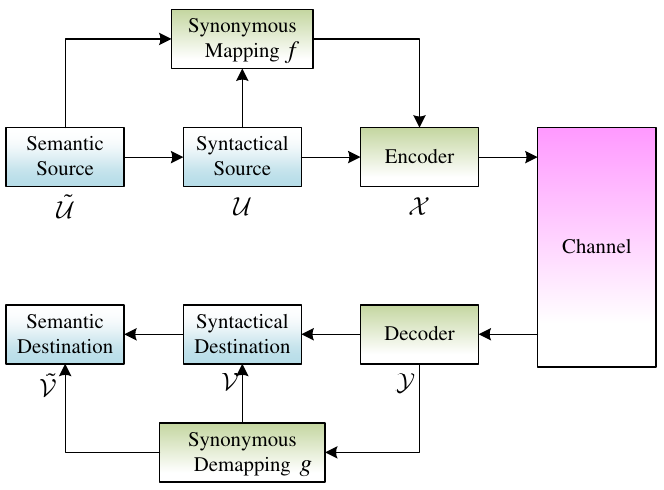}}
  \caption{The block diagram of semantic communication system.}\label{Semantic_communication_system}
\end{figure*}

At the transmitting side, the semantic source $\tilde U$ generates a semantic sequence and drives the syntactical source $U$ produces a message. With the help of synonymous mapping $f$, the encoder transforms the syntactical message into a codeword and sends to the channel. On the other hand, the receiver performs the reverse operations. The decoder recovers the codeword from the received signal and feeds it into the syntactical destination $V$ to reconstruct the message. Then the semantic destination $\tilde V$ obtains the message and reasons the meaning.

According to the general understanding of semantic information, the system design of semantic communication should obey the following criteria.
\begin{enumerate}[(i)]
  \item  \textbf{Invisibility of semantic information}
  Generally, semantic information is implied in syntactic information and cannot be directly observed. We only observe various data, such as text, speech, image or video so as to understand or infer the meaning of the message. Therefore, we highlight that the semantic information is invisible but perceptible. So only the syntactic information can be encoded and indirectly reveals the semantic information.
  \item  \textbf{Synonym of semantic information}
  The relationship between semantic information and syntactic information is diverse and complex. However, whether it is single-modal or multi-modal data, synonym, that is, one meaning has multiple manifestations, is a popular phenomenon. The basic characteristic of synonyms is an one-to-many mapping, which means that we can use a set of synonymous data to indicate the same semantic information. Commonly, this mapping is predefined or deterministic.
\end{enumerate}

By the first criterion, the semantic source/destination is a virtual module and implied behind the syntactic source/destination. Although Shannon stated that semantic information is irrelevant to the engineering problem, it is only suitable for the communication in Level A. If we design the communication system of Level B, the semantic source/dstination will indirectly affect the syntactical encoding and decoding process.

Correspondingly, by the second criterion, the synonymous mapping is deterministic rather than statistic, which is constructed by the common knowledge of transmitter/receiver or the system requirement. For the source coding, this mapping is mainly determined by the understanding of message. Although people have different opinions, they all use the same background knowledge. Therefore, knowledge base in the semantic communication system is a representative example. On the other hand, for the channel coding, this mapping can be created based on the transmission requirement. Unlike the traditional communication system, aided by the synonymous mapping, the reliability requirement of semantic communication can be relaxed from one-bit-no-error to tolerance of some error bits. This point will be further explained in Section \ref{section_VII}.
\subsection{Synonymous mapping}
Now we formally introduce the definition of synonymous mapping as follows.
\begin{definition}
Given a syntactic information set $\mathcal{U}=\left\{u_1,\cdots,u_i,\cdots,u_N\right\}$ and the corresponding semantic information set $\tilde {\mathcal{U}}=\left\{\tilde{u}_1,\cdots,\tilde{u}_{i_s},\cdots,\tilde{u}_{\tilde{N}}\right\}$, the synonymous mapping $f:\tilde{\mathcal{U}}\to\mathcal{U}$ is defined as the one-to-many mapping between $\tilde {\mathcal{U}}$ and $\mathcal{U}$.
\end{definition}

Generally, the size of set $\tilde {\mathcal{U}}$ is no more than $\mathcal{U}=\left\{u_1,u_2,\cdots,u_N\right\}$, that is, $\tilde{N}\leq N$. Furthermore, under the synonymous mapping $f$, $\mathcal{U}$ is partitioned into a group of synonymous sets $\mathcal{U}_{i_s}=\left\{u_{N_{[1:(i_s-1)]}+1},\cdots,u_{N_{[1:(i_s-1)]}+j},\cdots,u_{N_{[1:(i_s-1)]}+N_{i_s}}\right\}$ and $\forall i_s\neq j_s,\mathcal{U}_{i_s}\bigcap\mathcal{U}_{j_s}=\varnothing$. Therefore, we have $\left|\mathcal{U}_{i_s}\right|=N_{i_s}$ and $\mathcal{U}=\bigcup_{i_s=1}^{\tilde{N}}\mathcal{U}_{i_s}$.

Essentially, the synonymous mapping $f$ generates an equivalence class partition of the syntactic set. So we can construct the quotient set $\mathcal{U}/f=\left\{\mathcal{U}_{i_s}\right\}$.

For an arbitrary element $\tilde{u}_i\in \tilde{\mathcal{U}},i_s=1,2,\cdots, \tilde{N}$, we have $f:\tilde{u}_{i_s}\to\mathcal{U}_{i_s}$. In the case of non-confusion, for this mapping, we can drop out the subscript and present the synonymous set as $\mathcal{U}_{\tilde{u}}$,

\begin{figure}[htbp]
\setlength{\abovecaptionskip}{0.cm}
\setlength{\belowcaptionskip}{-0.cm}
  \centering{\includegraphics[scale=0.8]{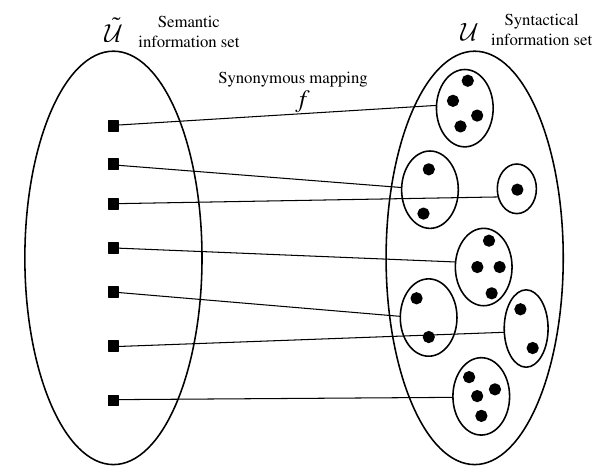}}
  \caption{An example of synonymous mapping between the semantic information set and the syntactic information set.}\label{Synonymous_mapping}
\end{figure}

Figure \ref{Synonymous_mapping} depicts an example of synonymous mapping. Each semantic element can be mapped into an equivalent set of syntactic elements and every set has one or many elements. Further, there is no overlap between any two sets.

\begin{remark}
We find that many presentations of semantic information in the previous works, such as logical probability \cite{Logic_Nilsson,Semantic_Bao}, knowledge presentation, feature map in deep learning, downstream task, etc. can be indicated by the synonymous mapping. In fact, the semantic variable $\tilde{U}$ is also a random variable. However, we emphasize that it is implied behind the random variable $U$ and can deduce the syntactic message $u$. In many applications of deep learning, after the nonlinear processing of neural networks, the semantic features of input source data are extracted and mapped into a latent space. Here, the nonlinear mapping can be regarded as an instance of synonymous mapping and the latent space can be treated as sample space of semantic variable. Therefore, synonymous mapping is the essence of semantic information which reveals the feature extraction and meaning abstraction of syntactic messages. Interestingly, Shannon also noted the synonym phenomenon in \cite{Shannon_Lattice} and he discussed the information lattice structure yet he did not give a general presentation and analysis.

\end{remark}
\section{Semantic Entropy}
\label{section_III}
In this section, we first give the definition of semantic entropy, that is, the measure of semantic information. Then we discuss the property of semantic joint entropy and semantic conditional entropy.
\subsection{Semantic Information Measures}
Let $\tilde {\mathcal{U}}=\left\{\tilde{u}_{i_s}\right\}_{i_s=1}^{\tilde{N}}$ and $\mathcal{U}=\left\{u_i\right\}_{i=1}^{N}$ be a semantic alphabet and a syntactic alphabet respectively. Let $U$ be a discrete random variable with alphabet $\mathcal{U}$ and probability mass function $p(u)=\Pr\left\{U=u\right\},u\in\mathcal{U}$. Given a synonymous mapping $f: \tilde {\mathcal{U}} \to \mathcal{U}$, the semantic information of a semantic symbol $\tilde{u}_{i_s}\in \tilde {\mathcal{U}}$ can be measured as follows,
\begin{equation}
I_s\left(\tilde{u}_{i_s}\right)=-\log \left(p\left(\mathcal{U}_{i_s}\right)\right)=-\log \left(\sum_{i=N_{[1:(i_s-1)]}+1}^{N_{[1:(i_s-1)]}+N_{i_s}}p\left(u_i\right)\right).
\end{equation}
\begin{definition}
Given a discrete random variable $U$, the corresponding semantic variable $\tilde {U}$, and the synonymous mapping $f: \tilde{\mathcal{U}}\to\mathcal{U}$, the semantic entropy of semantic variable $\tilde {U}$ is defined by
\begin{equation}\label{semantic_entropy}
\begin{aligned}
&H_s(\tilde{U})=-\sum_{i_s=1}^{\tilde{N}}p\left(\mathcal{U}_{i_s}\right)\log p\left(\mathcal{U}_{i_s}\right) \\
&=-\sum_{i_s=1}^{\tilde{N}}\sum_{i=N_{[1:(i_s-1)]}+1}^{N_{[1:(i_s-1)]}+N_{i_s}}p\left(u_i\right) \log \left(\sum_{i=N_{[1:(i_s-1)]}+1}^{N_{[1:(i_s-1)]}+N_{i_s}}p\left(u_i\right)\right)\\
&=-\sum_{i_s=1}^{\tilde{N}}\sum_{i\in\mathcal{N}_{i_s}}p\left(u_i\right) \log \left(\sum_{i\in\mathcal{N}_{i_s}}p\left(u_i\right)\right),
\end{aligned}
\end{equation}
\end{definition}
where $\mathcal{N}_{i_s}=\left\{N_{[1:(i_s-1)]}+1,\cdots,N_{[1:(i_s-1)]}+N_{i_s}\right\}$ is the index set associated with $\mathcal{U}_{i_s}$.
Essentially, semantic entropy is a functional of the distribution of $U$ and the synonymous mapping $f$. Similar to information entropy, semantic entropy also indicates the integrity attribute of the variable $\tilde{U}$ rather than single sample. Furthermore, it depends on the synonymous set partition determined by the mapping $f$. In Eq. (\ref{semantic_entropy}), the log is taken to the base 2 and we name it as the semantic binary digit, that is, \textbf{sebit}. So the unit of semantic information is expressed as sebit.

The semantic entropy of $\tilde{\mathcal{U}}$ can also be interpreted as the expectation of $\log\frac{1}{p\left(\mathcal{U}_{\tilde{u}}\right)}$, that is,
\begin{equation}
H_s(\tilde{U})=\mathbb{E}_p \left[\log \frac{1}{p\left(\mathcal{U}_{\tilde{u}}\right)}\right].
\end{equation}

For the semantic entropy, we have the following consequences by the definition.
\begin{lemma}
\begin{equation}
H_s(\tilde{U})\geq 0.
\end{equation}
\end{lemma}
\begin{proof}
Due to $0\leq p(u_i) \leq1$ and $\sum_{i=1}^{N}p(u_i)=1$, we have $0\leq p(\mathcal{U}_{i_s})=\sum_{i\in\mathcal{N}_{i_s}}p\left(u_i\right) \leq1$. So it follows that $ \log \frac{1}{p(\mathcal{U}_{i_s})} \geq 0$. \qedhere
\end{proof}

\begin{lemma}\label{lemma2}
The semantic entropy is no more than the associated information entropy, that is,
\begin{equation}
H_s(\tilde{U})\leq H(U).
\end{equation}
\end{lemma}
\begin{proof}
According to the definition of Eq. (\ref{semantic_entropy}), we have
\begin{equation}
\begin{aligned}
H_s(\tilde{U})-H(U)&=-\sum_{i_s=1}^{\tilde{N}}\sum_{i\in\mathcal{N}_{i_s}}p\left(u_i\right) \log \left(\sum_{i\in\mathcal{N}_{i_s}}p\left(u_i\right)\right)\\
     &+\sum_{i=1}^{N} p\left(u_i\right) \log p\left(u_i\right)\\
     &=\sum_{i_s=1}^{\tilde{N}}\sum_{i\in\mathcal{N}_{i_s}}p\left(u_i\right) \log \frac{p\left(u_i\right)}{\sum_{i\in\mathcal{N}_{i_s}}p\left(u_i\right)}\\
     &\leq \sum_{i=1}^{N} p\left(u_i\right) \log 1 =0.
\end{aligned}
\end{equation}
\end{proof}

\begin{lemma}\label{concave_sementropy}
The semantic entropy $H(\tilde{U})$ is a concave function of $p(u)$.
\end{lemma}

The proof is referred to Appendix \ref{proof_concave_sementropy}.

\begin{theorem}
$H_s(\tilde{U})\leq\log |\tilde{\mathcal{U}}|$, where $|\tilde{\mathcal{U}}|=\tilde{N}$ stands for the number of semantic symbols. The equality holds if and only if the synonymous set $\mathcal{U}_{i_s}$ is uniformly distributed over $\mathcal{U}$ and $U$ has arbitrary distribution over $\mathcal{U}_{i_s}$.
\end{theorem}
\begin{proof}
By the maximal entropy theorem of discrete source, we can easily obtain the conclusion and further derive that the equality holds when $p\left(\tilde{u}_{i_s}\right)=\frac{1}{\tilde{N}}=\sum_{i\in\mathcal{N}_{i_s}}p\left(u_i\right)$. So the elements in $\mathcal{U}_{i_s}$ can take an arbitrary distribution.
\end{proof}

Note that the maximal semantic entropy $\log |\tilde {\mathcal{U}}|$ is no more than the maximum of information entropy $\log \left| \mathcal{U}\right|$. That means the uncertainty of semantic information can be further reduced with the help of synonymous mapping.

\begin{example}
Table \ref{PDF_Semantic} gives a probability distribution of source $U$ and the associated semantic variable $\tilde{U}$ under a synonymous mapping.
\begin{table*}[tp]
\centering
\caption{Probability distribution of syntactic source $U$ and synonymous mapping of semantic source $\tilde{U}$.} \label{PDF_Semantic}
\begin{tabular}{|c|c|c|c|c|c|c|}
  \hline $U$        & \makebox[0.1\textwidth][c] { $u_1$ }    &     \makebox[0.05\textwidth][c] { $u_2$ }        &  \makebox[0.05\textwidth][c] {  $u_3$  } & \makebox[0.1\textwidth][c] { $u_4$ }  & \makebox[0.05\textwidth][c]{ $u_5$}  & \makebox[0.05\textwidth][c] { $u_6$ }  \\
  \hline $p(u)$                  &       $0.3$            &            $0.15$       &       $0.15$         &         $0.2$         &        $0.1$         &        $0.1$ \\
  \hline $\tilde{U}$  & $\tilde{u}_1\to \{u_1\}$ & \multicolumn{2}{c|}{$\tilde{u}_2\to \{u_2,u_3\}$} & $\tilde{u}_3\to \{u_4\}$ & \multicolumn{2}{c|} {$\tilde{u}_4\to\{u_5,u_6\}$}\\
  \hline $p(\tilde{u})$      &  $0.3$    &  \multicolumn{2}{c|} { $0.3$ }  &   $0.2$   &  \multicolumn{2}{c|}{ $0.2$ } \\
  \hline
\end{tabular}
\end{table*}
The entropy of source $U$ is calculated as $H(U)=-\sum_{i=1}^{6}p(u_i)\log p(u_i)=2.471\text{ bits}$. Correspondingly, the semantic entropy of $\tilde{U}$ is calculated as $H_s(\tilde{U})=-\sum_{i_s=1}^{4}p(\tilde{u}_{i_s})\log p(\tilde{u}_{i_s})=1.971\text{ sebits}$. Evidently, we observe that $H_s(\tilde{U})<H(U)$.
\end{example}

\subsection{Semantic joint entropy and semantic conditional entropy}
The definition of semantic entropy can be further extended to a pair of semantic variables. First, we introduce the jointly/conditionally synonymous mapping as following.
\begin{definition}
Given a pair of discrete semantic variables $(\tilde {U},\tilde {V})$ and the corresponding random variable pairs $\left(U,V\right)$, $f_{u}: \tilde {\mathcal{U}} \to \mathcal{U}$ denotes the synonymous mapping from $\tilde {\mathcal{U}}$ to $\mathcal{U}$ and we have $f_{u}: \tilde{u}_{i_s} \to \mathcal{U}_{i_s}$ where $1\leq i_s\leq |\tilde {\mathcal{U}}|={\tilde{N}}_u$ and $1\leq i \leq\left| {\mathcal{U}}\right|=N_u$. Similarly, we can define the synonymous mapping $f_{v}: \tilde {\mathcal{V}} \to \mathcal{V}$. Furthermore, the jointly synonymous mapping $f_{uv}: \tilde {\mathcal{U}}\times \tilde{\mathcal{V}}\to \mathcal{U}\times{\mathcal{V}}$ is defined as
\begin{equation}
f_{uv}: (\tilde{u}_{i_s},\tilde{v}_{j_s}) \to \mathcal{U}_{i_s}\times \mathcal{V}_{j_s}.
\end{equation}

Correspondingly, give a symbol $u_i$, the conditionally synonymous mapping $f_{v|u}: \tilde {\mathcal{V}}|u \to \mathcal{V}|u$ is defined as
\begin{equation}
f_{v|u}: \tilde{v}_{j_s}|u_i \to \mathcal{V}_{j_s}|u_i.
\end{equation}
\end{definition}

\begin{definition}
Given a pair of discrete semantic variables $(\tilde {U},\tilde {V})$ and the corresponding random variable pairs $\left(U,V\right)$ with a joint distribution $p(u,v)$, under the joint mapping $f_{uv}: \tilde {\mathcal{U}}\times \tilde{\mathcal{V}}\to \mathcal{U}\times{\mathcal{V}}$, the semantic joint entropy $H_s(\tilde {U},\tilde {V})$ is defined as
\begin{equation}\label{joint_semantic_entropy}
\begin{aligned}
H_s(\tilde {U},\tilde {V})=-\sum_{i_s=1}^{{\tilde{N}}_u}\sum_{j_s=1}^{{\tilde{N}}_v}&p\left(\mathcal{U}_{i_s}\times\mathcal{V}_{j_s}\right) \log p\left(\mathcal{U}_{i_s}\times\mathcal{V}_{j_s}\right)  \\
=-\sum_{i_s=1}^{{\tilde{N}}_u}\sum_{j_s=1}^{{\tilde{N}}_v} &\sum_{(u_i,v_j)\in \mathcal{U}_{i_s}\times \mathcal{V}_{j_s}} p\left(u_i,v_j\right)\\
\cdot \log &\sum_{(u_i,v_j)\in \mathcal{U}_{i_s}\times \mathcal{V}_{j_s}} p\left(u_i,v_j\right).
\end{aligned}
\end{equation}
\end{definition}

\begin{lemma}\label{lemma3}
The semantic joint entropy is no more than the joint entropy, that is,
\begin{equation}
H_s(\tilde{U},\tilde{V})\leq H(U,V)
\end{equation}
\end{lemma}
\begin{proof}
By the definition of Eq. (\ref{joint_semantic_entropy}), we have
\begin{equation}
\begin{aligned}
&H_s\left(\tilde{U},\tilde{V}\right)-H(U,V)\\
=&-\sum_{i_s=1}^{{\tilde{N}}_u}\sum_{j_s=1}^{{\tilde{N}}_v} \sum_{(u_i,v_j)} p\left(u_i,v_j\right) \log \sum_{(u_i,v_j)} p\left(u_i,v_j\right)\\
+&\sum_{(u_i,v_j)} p\left(u_i,v_j\right) \log p\left(u_i,v_j\right)\\
=&\sum_{(u_i,v_j)} p\left(u_i,v_j\right)\log \frac{p\left(u_i,v_j\right)}{\sum_{(u_i,v_j)} p\left(u_i,v_j\right)}\\
\leq & \sum_{(u_i,v_j)} p\left(u_i,v_j\right)\log 1 =0
\end{aligned}
\end{equation}
\end{proof}

We also define the semantic conditional entropy of a semantic variable given another random variable as follows.
\begin{definition}
Given a pair of discrete semantic variables $(\tilde {U},\tilde {V})$ and the corresponding random variable pairs $\left(U,V\right)\sim p(u,v)$, under the conditional mapping $f_{v|u}: \tilde {\mathcal{V}}|u \to \mathcal{V}|u$, the semantic conditional entropy $H_s(\tilde {V}|U)$ is defined as
\begin{equation}
\begin{aligned}
H_s(\tilde {V}| U)&=-\sum_{i=1}^{N_u}\sum_{j_s=1}^{{\tilde{N}}_v}p\left(u_i\right)p\left(\mathcal{V}_{j_s}\left|u_i\right.\right) \log p\left(\mathcal{V}_{j_s}\left|u_i\right.\right)  \\
&=-\sum_{i=1}^{N_u}\sum_{j_s=1}^{{\tilde{N}}_v} \sum_{(u_i,v_j)} p\left(u_i,v_j\right) \log \sum_{v_j\in \mathcal{V}_{j_s}|u_i} p\left(v_j\left|u_i\right.\right).
\end{aligned}
\end{equation}
\end{definition}
Similarly, we can also define the semantic conditional entropy $H_s(\tilde {U}|V )$.

\begin{lemma}\label{lemma4}
The semantic conditional entropy is no more than the conditional entropy, that is,
\begin{equation}
H_s\left(\tilde {V}| U\right)\leq H\left(V|U\right)
\end{equation}
\end{lemma}
\begin{proof}
The proof follows along the same lines as Lemma \ref{lemma3}.
\end{proof}

Similarly, we also attain the following lemma.
\begin{lemma}\label{lemma41}
If the semantic variables $\tilde{U}$ and $\tilde{V}$ are regarded as a new random variable respectively, the conditional entropy obeys the following relation, that is,
\begin{equation}
H\left(\tilde {V}| \tilde{U}\right)\leq H\left(V|\tilde{U}\right).
\end{equation}
\end{lemma}


The chain rule of entropies with two variables can be indicated by the following theorem.

\begin{theorem}\label{theorem2}
(\textit{Chain Rule of Entropies with Two Variables}):
\begin{equation}\label{chain_rule}
H_s(\tilde{U})+H_s(\tilde{V}\left|U\right.) \leq H_s(\tilde{U},\tilde{V})\leq H(V)+H_s(\tilde{U}\left|V\right.) \leq H(U,V)
\end{equation}
\end{theorem}
\begin{proof}
For the left inequality, due to the property of joint probability, we can write
\begin{equation}
\sum_{(u_i,v_j)\in \mathcal{U}_{i_s}\times  \mathcal{V}_{j_s}}p(u_i,v_j)\leq \left(\sum_{u_i \in \mathcal{U}_{i_s}}p(u_i)\right)\left(\sum_{v_j\in  \mathcal{V}_{i_s}|u_i}p(v_j\left|u_i\right.)\right).
\end{equation}
So we take the negative logarithm and expectation of inequality to prove this inequality. The equality holds if and only if $U$ and $V$ are mutually independent. By using a similar method, we can prove $H_s(\tilde{V})+H_s(\tilde{U}\left|V\right.) \leq H_s(\tilde{U},\tilde{V})$.

Similarly, we can write
\begin{equation}
p\left(v_j\right)\sum_{u_i\in \mathcal{U}_{i_s}|v_j}p\left(u_i\left|v_j\right.\right) \leq \sum_{(u_i,v_j)\in \mathcal{U}_{i_s}\times  \mathcal{V}_{j_s}}p\left(u_i,v_j\right).
\end{equation}
Take the negative logarithm and expectation of inequality, we prove the medium inequality.

Correspondingly, we also write
\begin{equation}
p\left(u_i\right)\sum_{v_j\in \mathcal{V}_{j_s}|u_i}p\left(v_j\left|u_i\right.\right) \leq \sum_{(u_i,v_j)\in \mathcal{U}_{i_s}\times  \mathcal{V}_{j_s}}p\left(u_i,v_j\right)
\end{equation}
and conclude that $H_s(\tilde{U},\tilde{V})\leq H(U) + H_s(\tilde{V}|U)$.

By using similar methods, we write
\begin{equation}
p(v_j\left|u_i\right.) \sum_{u_i\in \mathcal{U}_{i_s}} p\left(u_i\right) \leq \sum_{(u_i,v_j)\in \mathcal{U}_{i_s}\times  \mathcal{V}_{j_s}}p\left(u_i,v_j\right)
\end{equation}
and conclude that $H_s(\tilde{U},\tilde{V})\leq H_s(\tilde{U})+H(V\left|U\right.) $.

In addition, we also write
\begin{equation}
p\left(u_i\left|v_j\right.\right) \sum_{v_j\in \mathcal{V}_{j_s}} p\left(v_j\right) \leq \sum_{(u_i,v_j)\in \mathcal{U}_{i_s}\times  \mathcal{V}_{j_s}}p\left(u_i,v_j\right)
\end{equation}
and conclude that $H_s(\tilde{U},\tilde{V})\leq H_s(\tilde{V})+H(U\left|V\right.)$.

By Lemma \ref{lemma4}, due to $H(V)+H_s(\tilde{U}|V)-H(U, V) = H(V)+H_s(\tilde{U}|V)-H(V)-H(U|V)\leq 0$, we prove the right inequality. Similarly, we can attain that $H(U) + H_s(\tilde{V}|U) \leq H(U,V)$.
\end{proof}

\begin{definition}
Given a semantic sequence $(\tilde {U}_1,\tilde {U}_2,\cdots,\tilde {U}_n)$ and the associated syntactic sequence $\left(U_1,U_2,\cdots,U_n\right)$, $f^n: \tilde{\mathcal{U}}^n\to\mathcal{U}^n$ denotes the sequential synonymous mapping from $\tilde {\mathcal{U}}$ to $\mathcal{U}$ and we have
\begin{equation}
f^n: (\tilde{u}^n) \to \prod_{k=1}^{n} \mathcal{U}_{\tilde{u}_k}.
\end{equation}

Correspondingly, give a subvector $u_1^{k-1}$, the conditionally synonymous mapping $f_{u_k|u_1^{k-1}}: \tilde {\mathcal{U}}|u_1^{k-1} \to \mathcal{U}|u_1^{k-1}$ is defined as
\begin{equation}
f_{u_k|u_1^{k-1}}: \tilde{u}_{k}|u_i^{k-1} \to \mathcal{U}_{\tilde{u}_k}|u_1^{k-1}.
\end{equation}
\end{definition}

\begin{theorem}\label{theorem_chainrule}
(\textit{Chain Rule of Sequential Entropy}):

Given a semantic sequence $(\tilde {U}_1,\tilde {U}_2,\cdots,\tilde {U}_n)$ and the associated syntactic sequence $\left(U_1,U_2,\cdots,U_n\right)$, under the sequential synonymous mapping $f^n$ and conditional synonymous mappings $f_{u_k|u_1^{k-1}}$, we have

\begin{equation}
\begin{aligned}
&\sum_{k=1}^{n} H_s(\tilde{U}_k\left|U_1^{k-1}\right.) \leq H_s(\tilde{U}^n)\\
&\leq \tilde{H}(\tilde{U}_1^{n-1},U_{n}) \leq \cdots \leq \tilde{H}(\tilde{U}_1^m,U_{m+1}^n) \\
&\leq \cdots \leq \tilde{H}(\tilde{U}_1,U_{2}^n)\\
& \leq H(U^n)=\sum_{k=1}^{n} H(U_k\left|U_1^{k-1}\right.)
\end{aligned}
\end{equation}
where $\tilde{H}(\tilde{U}_1^m,U_{m+1}^n) =\sum_{k=1}^{m} H(\tilde{U}_k\left|\tilde{U}_1^{k-1}\right.)+\sum_{k=m+1}^{n} H(U_k\left|\tilde{U}_1^{m},U_{m+1}^{k-1}\right.)$ with $m=n-1,\cdots,2,1$ denote sequential entropies.

\end{theorem}

\begin{proof}
In order to prove inequality $\sum_{k=1}^{n} H_s(\tilde{U}_k\left|U_1^{k-1}\right.) \leq H_s(\tilde{U}^n)$, by using the property of joint probability and synonymous mappings, we can write
\begin{equation}
\sum_{u_1^n\in \prod_{k=1}^{n}\mathcal{U}_{\tilde{u}_k}}p(u_1,u_2,\cdots,u_n)\leq \prod_{k=1}^{n} \sum_{u_k\in  \mathcal{U}_{\tilde{u}_k}|u_1^{k-1}}p(u_k\left|u_{1}^{k-1}\right.).
\end{equation}
So we take the negative logarithm and expectation to prove this inequality. The equality holds if and only if all elements $U_k$ are mutually independent.

By using the chain rule of sequential entropy, after a permutation of sequence, we have $\tilde{H}(\tilde{U}_1^m,U_{m+1}^n)=\tilde{H}(U_{m+1}^n,\tilde{U}_1^{m-1})+H(\tilde{U}_m|U_{m+1}^n,\tilde{U}_1^{m-1})$ and $\tilde{H}(\tilde{U}_1^{m-1},U_{m}^n)=\tilde{H}(U_{m+1}^n,\tilde{U}_1^{m-1})+H(U_m|U_{m+1}^n,\tilde{U}_1^{m-1})$. According to Lemma \ref{lemma41}, it follows that $H(\tilde{U}_m|U_{m+1}^n,\tilde{U}_1^{m-1})\leq H(U_m|U_{m+1}^n,\tilde{U}_1^{m-1})$. So we can attain that
\begin{equation}
\tilde{H}(\tilde{U}_1^m,U_{m+1}^n)\leq \tilde{H}(\tilde{U}_1^{m-1},U_{m}^n), \text{ for } m=n-1,\cdots,2,
\end{equation}
and
\begin{equation}
\tilde{H}(\tilde{U}_1,U_{2}^n)\leq H(U^n).
\end{equation}
\end{proof}

\begin{remark}
In classic information theory, the joint entropy of a pair of random variables is the entropy of one variable plus the conditional entropy of the other variable, that is, $H\left(U,V\right)=H(U)+H(V|U)$. On the other hand, the semantic joint entropy does not satisfy the addition of entropy and degrades to an inequality of information entropy plus semantic conditional entropy, that is, $ H_s(\tilde{U},\tilde{V})\leq H(U)+H_s(\tilde{V}\left|U\right.)$ or $H_s(\tilde{U},\tilde{V})\leq H(V)+H_s(\tilde{U}\left|V\right.)$.
\end{remark}

\begin{example}
Table \ref{JPDF_RV} gives a joint probability distribution of random variable pair $(U,V)$. Table \ref{JSmapping_SRV} illustrates the distribution of the associated semantic variable pair $(\tilde{U},\tilde{V})$ under a joint synonymous mapping $f_{uv}$. Table \ref{CSmapping_SRV} and Table \ref{CSmapping_SRV2} give the conditional distribution of the semantic variables $\tilde{U}|V$ and $\tilde{V}|U$. The marginal distributions of the semantic variable $\tilde{U}$ and $\tilde{V}$ are depicted in Table \ref{Smapping_SRV}.

\begin{table*}[htbp]
\centering
\caption{Joint probability distribution of random variable pair $(U,V)$.} \label{JPDF_RV}
\begin{tabular}{|c|c|c|c|c|c|}
  \hline $(U,V)$     & $(u_1,v_1)$ &  $(u_1,v_2)$  & $(u_1,v_3)$ & $(u_1,v_4)$  & $(u_1,v_5)$  \\
  \hline $p(u,v)$     &   $0.05$       &       $0.1$       &    $0.15$      &       $0$         &       $0$         \\
  \hline $(U,V)$      & $(u_2,v_1)$ &  $(u_2,v_2)$  & $(u_2,v_3)$ & $(u_2,v_4)$  & $(u_2,v_5)$  \\
  \hline $p(u,v)$     &     $0.1$       &     $0.05$       &    $0.05$      &     $0.1$        &       $0$         \\
  \hline $(U,V)$     & $(u_3,v_1)$ &  $(u_3,v_2)$  & $(u_3,v_3)$ & $(u_3,v_4)$  & $(u_3,v_5)$  \\
  \hline $p(u,v)$     &     $0.1$       &     $0.05$       &        $0$       &       $0$         &      $0.05$     \\
  \hline $(U,V)$     & $(u_4,v_1)$ &  $(u_4,v_2)$  & $(u_4,v_3)$ & $(u_4,v_4)$  & $(u_4,v_5)$  \\
  \hline $p(u,v)$     &   $0.05$       &       $0$          &        $0$       &     $0.1$        &      $0.05$     \\
  \hline
\end{tabular}
\end{table*}

\begin{table*}[htbp]
\centering
\caption{Joint synonymous mapping of semantic variable pair $(\tilde{U},\tilde{V})$.} \label{JSmapping_SRV}
\begin{tabular}{|c|c|c|}
  \hline $f_{uv}$  & $({\tilde{u}}_1,{\tilde{v}}_1)\to\{(u_1,v_1),(u_2,v_1)\}$    & $({\tilde{u}}_1,{\tilde{v}}_2)\to\{(u_1,v_2),(u_2,v_2)\}$ \\
  \hline $p(\tilde{u},\tilde{v})$                   &       $0.15$            &            $0.15$        \\
  \hline $f_{uv}$  & $({\tilde{u}}_1,{\tilde{v}}_3)\to\{(u_1,v_3),(u_2,v_3)\}$ & $({\tilde{u}}_1,{\tilde{v}}_4)\to\{(u_1,v_4),(u_1,v_5),(u_2,v_4),(u_2,v_5)\}$ \\
  \hline $p(\tilde{u},\tilde{v})$                   &      $0.2$         &         $0.1$  \\
  \hline $f_{uv}$  & $({\tilde{u}}_2,{\tilde{v}}_1)\to\{(u_3,v_1),(u_4,v_1)\}$ & $(\tilde{u}_2,\tilde{v}_2)\to\{(u_3,v_2),(u_4,v_2)\}$ \\
  \hline $p(\tilde{u},\tilde{v})$                   &  $0.15$    &  $0.05$ \\
  \hline $f_{uv}$  & $(\tilde{u}_2,\tilde{v}_3)\to\{(u_3,v_3),(u_4,v_3)\}$ & $(\tilde{u}_2,\tilde{v}_4)\to\{(u_3,v_4),(u_4,v_4),(u_3,v_5),(u_4,v_5)\}$ \\
   \hline $p(\tilde{u},\tilde{v})$                   &   $0$   &  $0.2$  \\
  \hline
\end{tabular}
\end{table*}

The joint entropy of $(U,V)$ is calculated as $H(U,V)=$ $-\sum_{i=1}^{4}\sum_{j=1}^{5} p(u_i,v_j)\log p(u_i,v_j)$ $=3.5842\text{ bits}$.

The semantic joint entropy of $(\tilde{U},\tilde{V})$ is $H_s(\tilde{U},\tilde{V})=$ $-\sum_{i_s=1}^{2}\sum_{j_s=1}^{4} p(\tilde{u}_{i_s},\tilde{v}_{j_s})\log p(\tilde{u}_{i_s},\tilde{v}_{j_s})$ $=2.7087 \text{ sebits}$.

The conditional entropies are $H(U|V)=$ $-\sum_{i=1}^{4}\sum_{j=1}^{5} p(u_i,v_j)\log p(u_i|v_j)$ $=1.3377 \text{ bits}$ and $H(V|U)=$ $-\sum_{j=1}^{5}\sum_{i=1}^{4} p(u_i,v_j)\log p(v_j|u_i)$ $=1.6132 \text{ bits}$.

Correspondingly, the semantic conditional entropies are $H_s(\tilde{U}|V)=-\sum_{i_s=1}^{2}\sum_{j=1}^{5} p(\tilde{u}_{i_s},v_j)$\\$\cdot\log p(\tilde{u}_{i_s}|v_j)$ $=0.6623 \text{ sebits}$ and $H_s(\tilde{V}|U)$ $=-\sum_{j_s=1}^{4}\sum_{i=1}^{4} p(\tilde{v}_{j_s},u_i)\log p(\tilde{v}_{j_s}|u_i)$ $=1.4755 \text{ sebits}$. Thus we have $H_s(\tilde{U}|V)=0.6623 \text{ sebits}<H(U|V)=1.3377 \text{ bits}$ and $H_s(\tilde{V}|U)=1.4755 \text{ sebits}<H(V|U)=1.6132 \text{ bits}$.

The entropies of random variables $U$ and $V$ are calculated as $H(U)=H(0.3,0.3,0.2,0.2)=1.971 \text{ bits}$ and $H(V)=H(0.3,0.2,0.2,0.2,0.1)=2.2464 \text{ bits}$ respectively. Then, the semantic entropies of $\tilde{U}$ and $\tilde{V}$ are $H_s(\tilde{U})=H_s(0.6,0.4)=0.971 \text{ sebits}$ and $H_s(\tilde{V})=H_s(0.3,0.2,0.2,0.3)=1.971 \text{ sebits}$ respectively.
So it follows that $H_s(\tilde{U},\tilde{V})=2.7087 \text{ sebits}<H(U,V)=3.5842 \text{ bits}$, $H_s(\tilde{V})+H_s(\tilde{U}|V)=2.6633 \text{ sebits}<H_s(\tilde{U},\tilde{V})=2.7087 \text{ sebits}<H(V)+H_s(\tilde{U}|V)=2.9087 \text{ sebits}$, and $H_s(\tilde{U})+H_s(\tilde{V}|U)=2.4465 \text{ sebits}<H_s(\tilde{U},\tilde{V})=2.7087 \text{ sebits}<H(U)+H_s(\tilde{V}|U)=3.4465 \text{ sebits}$.

\begin{table*}[htbp]
\centering
\caption{Conditional synonymous mapping of semantic variable $\tilde{U}|V$.} \label{CSmapping_SRV}
\begin{tabular}{|c|c|c|c|c|c|}
  \hline $\tilde{U}|V$     & $\tilde{u}_1|v_1\to \{u_1,u_2\}|v_1 $ &  $\tilde{u}_1|v_2\to \{u_1,u_2\}|v_2$ & $\tilde{u}_1|v_3\to \{u_1,u_2\}|v_3$ & $\tilde{u}_1|v_4\to \{u_1,u_2\}|v_4$ & $\tilde{u}_1|v_5\to \{u_1,u_2\}|v_5$ \\
  \hline $p(\tilde{u}|v)$         &   $0.5$       &       $0.75$       &    $1$      &       $0.5$         &       $0$         \\
  \hline $\tilde{U}|V$     & $\tilde{u}_2|v_1\to \{u_3,u_4\}|v_1 $ &  $\tilde{u}_2|v_2\to \{u_3,u_4\}|v_2$ & $\tilde{u}_2|v_3\to \{u_3,u_4\}|v_3$ & $\tilde{u}_2|v_4\to \{u_3,u_4\}|v_4$ & $\tilde{u}_2|v_5\to \{u_3,u_4\}|v_5$ \\
  \hline $p(\tilde{u}|v)$         &     $0.5$       &     $0.25$       &        $0$       &       $0.5$         &      $1$     \\
  \hline
\end{tabular}
\end{table*}

\begin{table*}[htbp]
\centering
\caption{Conditional synonymous mapping of semantic variable $\tilde{V}|U$.} \label{CSmapping_SRV2}
\begin{tabular}{|c|c|c|c|c|}
  \hline $\tilde{V}|U$     & $\tilde{v}_1|u_1\to \{v_1\}|u_1 $ &  $\tilde{v}_1|u_2\to \{v_1\}|u_2$ & $\tilde{v}_1|u_3\to \{v_1\}|u_3$ & $\tilde{v}_1|u_4\to \{v_1\}|u_4$  \\
  \hline $p(\tilde{v}|u)$         &   $1/6$       &       $1/3$       &    $0.5$      &       $0.25$           \\
  \hline $\tilde{V}|U$     & $\tilde{v}_2|u_1\to \{v_2\}|u_1 $ &  $\tilde{v}_2|u_2\to \{v_2\}|u_2$ & $\tilde{v}_2|u_3\to \{v_2\}|u_3$ & $\tilde{v}_2|u_4\to \{v_2\}|u_4$  \\
  \hline $p(\tilde{v}|u)$         &   $1/3$       &       $1/6$       &    $0.25$      &       $0$             \\
  \hline $\tilde{V}|U$     & $\tilde{v}_3|u_1\to \{v_3\}|u_1 $ &  $\tilde{v}_3|u_2\to \{v_3\}|u_2$ & $\tilde{v}_3|u_3\to \{v_3\}|u_3$ & $\tilde{v}_3|u_4\to \{v_3\}|u_4$  \\
  \hline $p(\tilde{v}|u)$         &   $0.5$       &       $1/6$       &    $0$           &       $0$             \\
  \hline $\tilde{V}|U$     & $\tilde{v}_4|u_1\to \{v_4,v_5\}|u_1 $ &  $\tilde{v}_4|u_2\to \{v_4,v_5\}|u_2$ & $\tilde{v}_4|u_3\to \{v_4,v_5\}|u_3$ & $\tilde{v}_4|u_4\to \{v_4,v_5\}|u_4$  \\
  \hline $p(\tilde{v}|u)$         &   $0$         &       $1/3$        &    $0.25$      &       $0.75$         \\
  \hline
\end{tabular}
\end{table*}

\begin{table*}[htbp]
\centering
\caption{Synonymous mappings of semantic variables $\tilde{U}$ and $\tilde{V}$.} \label{Smapping_SRV}
\begin{tabular}{|c|c|c|}
  \hline $f_{u}$               & $\tilde{u}_1\to\{u_1,u_2\}$  &  $\tilde{u}_2\to\{u_3,u_4\}$  \\
  \hline $p(\tilde{u})$           &           $0.6$                        &            $0.4$         \\
  \hline
\end{tabular}
\begin{tabular}{|c|c|c|c|c|}
  \hline $f_{v}$         & $\tilde{v}_1 \to\{v_1\}$  & $\tilde{v}_2 \to\{v_2\}$ & $\tilde{v}_3 \to\{v_3\}$ & $\tilde{v}_4 \to\{v_4,v_5\}$ \\
  \hline $p(\tilde{v})$    &               $0.3$               &            $0.2$                   &                $0.2$              &               $0.3$ \\
  \hline
\end{tabular}
\end{table*}

\end{example}
\section{Semantic Relative Entropy and Mutual Information}
\label{section_IV}
In this section, we apply the synonymous mapping to define semantic relative entropy and semantic mutual information. Three semantic relative entropies are measures of the distance between two semantic/syntactic variables. We introduce two measures, such as up/down semantic mutual information to evaluate the reduction in the semantic information of one variable due to the knowledge of the other variable.
\subsection{Semantic Relative Entropy}
In classic information theory, the relative entropy or Kullback Leibler distance $D\left(p\|q\right)=\sum_{u\in\mathcal{U}}p(u)\log\frac{p(u)}{q(u)}$ is used to measure the difference between two probability mass function $p(u)$ and $q(u)$. Similarly, by the synonymous mapping, we can define the semantic relative entropy as following.

\begin{definition}
Given the semantic variable $\tilde {\mathcal{U}}$ and the random variable $\mathcal{U}$ with two probability mass function $p(u)$ and $q(u)$, under the synonymous mapping $f: \tilde {\mathcal{U}}\to  {\mathcal{U}}$, the full semantic relative entropy is defined as
\begin{equation}
D_s\left(p_s\|q_s\right)=\sum_{i_s=1}^{\tilde{N}}\sum_{u_i\in \mathcal{U}_{i_s}}p(u_i) \log \frac{\sum_{u_i\in \mathcal{U}_{i_s}} p(u_i)}{\sum_{u_i\in \mathcal{U}_{i_s}} q(u_i)}.
\end{equation}
Two partial semantic relative entropies are defined as
\begin{equation}
D_s\left(p_s\|q\right)=\sum_{i_s=1}^{\tilde{N}}\sum_{u_i\in \mathcal{U}_{i_s}}p(u_i) \log \frac{\sum_{u_i\in \mathcal{U}_{i_s}} p(u_i)}{q(u_i)},
\end{equation}
and
\begin{equation}
D_s\left(p\|q_s\right)=\sum_{i_s=1}^{\tilde{N}}\sum_{u_i\in \mathcal{U}_{i_s}}p(u_i) \log \frac{p(u_i)}{\sum_{u_i\in \mathcal{U}_{i_s}} q(u_i)}.
\end{equation}
\end{definition}
Hereafter, in order to simplify the presentation, we use $p_s=\sum_{u_i\in \mathcal{U}_{i_s}} p(u_i)$ and $q_s=\sum_{u_i\in \mathcal{U}_{i_s}} q(u_i)$ to stand for two different probability distributions of semantic variable $\tilde{U}$.

Similar to the relative entropy, we have the following basic inequality for the semantic relative entropy.
\begin{theorem}\label{theorem3}
(Semantic information inequality): Let $p(u)$, $q(u)$, $u\in\mathcal{U}$, be two probability mass function. Given the synonymous mapping $f$, we have
\begin{equation}
\left\{\begin{aligned}
&D_s(p_s\|q_s)\geq 0\\
&D_s(p_s\|q)\geq D_s(p_s\|p)\\
&D_s(p\|q_s)\geq D_s(p\|p_s)\\
\end{aligned}
\right.
\end{equation}
with equality if and only if
\begin{equation}
\left\{\begin{aligned}
\sum_{u_i \in \mathcal{U}_{i_s}}p(u_i)&=\sum_{u_i \in \mathcal{U}_{i_s}}q(u_i), &\text{for all } \mathcal{U}_{i_s} \\
p(u_i)&=q(u_i), &\text{for all } \mathcal{U}_{i_s} \\
\sum_{u_i \in \mathcal{U}_{i_s}}p(u_i)&=\sum_{u_i \in \mathcal{U}_{i_s}}q(u_i), &\text{for all } \mathcal{U}_{i_s} \\
\end{aligned}\right.
\end{equation}
\end{theorem}
\begin{proof}
For the first inequality, by Jensen's inequality, we can write
\begin{equation}
\begin{aligned}
-D_s\left(p_s\|q_s\right)&=-\sum_{i_s=1}^{\tilde{N}}\sum_{u_i\in \mathcal{U}_{i_s}}p(u_i) \log \frac{\sum_{u_i\in \mathcal{U}_{i_s}} p(u_i)}{\sum_{u_i\in \mathcal{U}_{i_s}} q(u_i)}\\
&=\sum_{i_s=1}^{\tilde{N}}\sum_{u_i\in \mathcal{U}_{i_s}}p(u_i) \log \frac{\sum_{u_i\in \mathcal{U}_{i_s}} q(u_i)}{\sum_{u_i\in \mathcal{U}_{i_s}} p(u_i)}\\
&\leq \log \sum_{i_s=1}^{\tilde{N}}\sum_{u_i\in \mathcal{U}_{i_s}}p(u_i)  \frac{\sum_{u_i\in \mathcal{U}_{i_s}} q(u_i)}{\sum_{u_i\in \mathcal{U}_{i_s}} p(u_i)}\\
&=\log \sum_{i_s=1}^{\tilde{N}}  {\sum_{u_i\in \mathcal{U}_{i_s}} q(u_i)}\\
&=\log 1 =0.
\end{aligned}
\end{equation}

For the second inequality, we can write
\begin{equation}
\begin{aligned}
D_s\left(p_s\|q\right)&=\sum_{i_s=1}^{\tilde{N}}\sum_{u_i\in \mathcal{U}_{i_s}}p(u_i) \log \frac{\sum_{u_i\in \mathcal{U}_{i_s}} p(u_i)}{q(u_i)}\\
&=\sum_{i_s=1}^{\tilde{N}}\sum_{u_i\in \mathcal{U}_{i_s}}p(u_i) \log \frac{\sum_{u_i\in \mathcal{U}_{i_s}} p(u_i)}{p(u_i)}\\
&+\sum_{i_s=1}^{\tilde{N}}\sum_{u_i\in \mathcal{U}_{i_s}}p(u_i) \log \frac{ p(u_i)}{q(u_i)}\\
&\geq \sum_{i_s=1}^{\tilde{N}}\sum_{u_i\in \mathcal{U}_{i_s}}p(u_i) \log \frac{\sum_{u_i\in \mathcal{U}_{i_s}} p(u_i)}{p(u_i)}\\
&=D_s(p_s\|p)\geq 0.
\end{aligned}
\end{equation}

For the third inequality, we further derive as follows
\begin{equation}
\begin{aligned}
D_s\left(p\|q_s\right)&=\sum_{i_s=1}^{\tilde{N}}\sum_{u_i\in \mathcal{U}_{i_s}}p(u_i) \log \frac{p(u_i)}{\sum_{u_i\in \mathcal{U}_{i_s}} q(u_i)}\\
&=\sum_{i_s=1}^{\tilde{N}}\sum_{u_i\in \mathcal{U}_{i_s}}p(u_i) \log \frac{p(u_i)}{\sum_{u_i\in \mathcal{U}_{i_s}} p(u_i)}\\
&+\sum_{i_s=1}^{\tilde{N}}\sum_{u_i\in \mathcal{U}_{i_s}}p(u_i) \log \frac{\sum_{u_i\in \mathcal{U}_{i_s}} p(u_i)}{\sum_{u_i\in \mathcal{U}_{i_s}} q(u_i)}\\
&\geq \sum_{i_s=1}^{\tilde{N}}\sum_{u_i\in \mathcal{U}_{i_s}}p(u_i) \log \frac{p(u_i)}{\sum_{u_i\in \mathcal{U}_{i_s}} p(u_i)}\\
&=D_s(p\|p_s).
\end{aligned}
\end{equation}

Note that $D_s\left(p\|q_s\right)$ may be negative. In the practical application, we can take the non-negative value, that is, $\left(D_s\left(p\|q_s\right)\right)^{+}$.
\end{proof}

\begin{corollary}
\begin{equation}
D_s(p\|q_s)\leq D_s(p_s\|q_s) \leq D_s(p_s\|q).
\end{equation}
\end{corollary}
\begin{proof}
For the left inequality, we can write
\begin{equation}
\begin{aligned}
&D_s(p\|q_s)-D_s(p_s\|q_s) \\
&=\sum_{i_s=1}^{\tilde{N}}\sum_{u_i\in \mathcal{U}_{i_s}}p(u_i) \log \frac {p(u_i)}{\sum_{u_i\in \mathcal{U}_{i_s}} q(u_i)}\\
&-\sum_{i_s=1}^{\tilde{N}}\sum_{u_i\in \mathcal{U}_{i_s}}p(u_i) \log \frac{\sum_{u_i\in \mathcal{U}_{i_s}}p(u_i)}{\sum_{u_i\in \mathcal{U}_{i_s}} q(u_i)}\\
&=\sum_{i_s=1}^{\tilde{N}}\sum_{u_i\in \mathcal{U}_{i_s}}p(u_i) \log \frac{p(u_i)}{\sum_{u_i\in \mathcal{U}_{i_s}} p(u_i)}\\
&\leq \sum_{i_s=1}^{\tilde{N}}\sum_{u_i\in \mathcal{U}_{i_s}}p(u_i) \log 1=0.\\
\end{aligned}
\end{equation}

For the right inequality, we can write
\begin{equation}
\begin{aligned}
&D_s(p_s\|q_s)-D_s(p_s\|q)\\
&=\sum_{i_s=1}^{\tilde{N}}\sum_{u_i\in \mathcal{U}_{i_s}}p(u_i) \log \frac{\sum_{u_i\in \mathcal{U}_{i_s}}p(u_i)}{\sum_{u_i\in \mathcal{U}_{i_s}} q(u_i)}\\
&-\sum_{i_s=1}^{\tilde{N}}\sum_{u_i\in \mathcal{U}_{i_s}}p(u_i) \log \frac{\sum_{u_i\in \mathcal{U}_{i_s}} p(u_i)}{q(u_i)}\\
&=\sum_{i_s=1}^{\tilde{N}}\sum_{u_i\in \mathcal{U}_{i_s}}p(u_i) \log \frac{q(u_i)}{\sum_{u_i\in \mathcal{U}_{i_s}} q(u_i)}\\
&\leq \sum_{i_s=1}^{\tilde{N}}\sum_{u_i\in \mathcal{U}_{i_s}}p(u_i) \log 1=0.\\
\end{aligned}
\end{equation}
\end{proof}

\begin{corollary}
\begin{equation}
D_s(p\|q_s)\leq D(p\|q) \leq D_s(p_s\|q).
\end{equation}
\end{corollary}
\begin{proof}
For the left inequality, we can write
\begin{equation}
\begin{aligned}
&D_s(p\|q_s)-D(p\|q) \\
&=\sum_{i_s=1}^{\tilde{N}}\sum_{u_i\in \mathcal{U}_{i_s}}p(u_i) \log \frac {p(u_i)}{\sum_{u_i\in \mathcal{U}_{i_s}} q(u_i)}\\
&-\sum_{i_s=1}^{\tilde{N}}\sum_{u_i\in \mathcal{U}_{i_s}}p(u_i) \log \frac{p(u_i)}{q(u_i)}\\
&=\sum_{i_s=1}^{\tilde{N}}\sum_{u_i\in \mathcal{U}_{i_s}}p(u_i) \log \frac{q(u_i)}{\sum_{u_i\in \mathcal{U}_{i_s}} q(u_i)}\\
&\leq \sum_{i_s=1}^{\tilde{N}}\sum_{u_i\in \mathcal{U}_{i_s}}p(u_i) \log 1=0.\\
\end{aligned}
\end{equation}

For the right inequality, we can write
\begin{equation}
\begin{aligned}
&D_s(p\|q)-D_s(p_s\|q)\\
&=\sum_{i_s=1}^{\tilde{N}}\sum_{u_i\in \mathcal{U}_{i_s}}p(u_i) \log \frac{p(u_i)}{q(u_i)}\\
&-\sum_{i_s=1}^{\tilde{N}}\sum_{u_i\in \mathcal{U}_{i_s}}p(u_i) \log \frac{\sum_{u_i\in \mathcal{U}_{i_s}} p(u_i)}{q(u_i)}\\
&=\sum_{i_s=1}^{\tilde{N}}\sum_{u_i\in \mathcal{U}_{i_s}}p(u_i) \log \frac{p(u_i)}{\sum_{u_i\in \mathcal{U}_{i_s}} p(u_i)}\\
&\leq \sum_{i_s=1}^{\tilde{N}}\sum_{u_i\in \mathcal{U}_{i_s}}p(u_i) \log 1=0.\\
\end{aligned}
\end{equation}
\end{proof}

\begin{theorem}\label{convexity_crossentropy}
$D_s(p_s\|q_s)$, $D_s(p_s\|q)$, and $D_s(p\|q_s)$ are convex in the pair $(p,q)$. Equivalently, given $(p_1,q_1)$ and $(p_2,q_2)$ are two pairs of probability mass functions, for all $0\leq\theta\leq1$, we have
\begin{equation}
\left\{\begin{aligned}
&D_s(\theta p_{s,1} + (1-\theta) p_{s,2}\| \theta q_{s,1} + (1-\theta) q_{s,2})\\
& \leq \theta D_s(p_{s,1}\|q_{s,1}) +(1-\theta) D_s(p_{s,2}\|q_{s,2})\\
&D_s(\theta p_{s,1} + (1-\theta) p_{s,2}\| \theta q_{1} + (1-\theta) q_{2})\\
& \leq \theta D_s(p_{s,1}\|q_{1}) +(1-\theta) D_s(p_{s,2}\|q_{2})\\
&D_s(\theta p_{1} + (1-\theta) p_{2}\| \theta q_{s,1} + (1-\theta) q_{s,2})\\
& \leq \theta D_s(p_{1}\|q_{s,1}) +(1-\theta) D_s(p_{2}\|q_{s,2})\\
\end{aligned}\right.
\end{equation}
\end{theorem}
The proof is referred to Appendix \ref{proof_convexity_crossentropy}.

\begin{remark}
In neural network model, relative entropy or cross entropy is an important cost function used to training. For some deep learning application involved semantic information, such as clustering, classification, recognition, we can use semantic relative entropy to alternate the classic counterpart so as to further improve the system performance.
\end{remark}

\subsection{Semantic Mutual Information}
We now introduce up semantic mutual information, which is a partial relative entropy to indicate the large reduction in the semantic information of one variable due to the knowledge of the other.
\begin{definition}\label{definition7}
Consider two semantic variables $\tilde{U}$ and $\tilde{V}$ and two associated random variables $U$ and $V$ with a joint probability mass function $p(u,v)$ and marginal probability mass function $p(u)$ and $p(v)$. Given the jointly synonymous mapping $f_{uv}: \tilde {\mathcal{U}}\times \tilde{\mathcal{V}}\to \mathcal{U}\times{\mathcal{V}}$, the up semantic mutual information $I^s(\tilde{U};\tilde{V})$ is the partial entropy between the joint distribution $p_s\left(u,v\right)$ and the product distribution $p(u)p(v)$, i.e.,
\begin{equation}
\begin{aligned}
I^s(\tilde{U};\tilde{V})=&-\sum_{i_s=1}^{{\tilde{N}}_u}\sum_{j_s=1}^{{\tilde{N}}_v}\sum_{(u_i,v_j)\in \mathcal{U}_{i_s}\times \mathcal{V}_{j_s}}p\left(u_i,v_j\right) \\
&\cdot \log \frac{p\left(u_i\right)p\left(v_j\right)}{\sum_{(u_i,v_j) \in \mathcal{U}_{i_s} \times \mathcal{V}_{j_s}}p\left(u_i,v_j\right)}\\
=&D_s\left(p_s\left(u,v\right)\|p(u)p(v)\right)\\
=&H(U)+H(V)-H_s(\tilde{U},\tilde{V}).
\end{aligned}
\end{equation}

Similarly, the down semantic mutual information $I_s(\tilde{U};\tilde{V})$ is the partial entropy between the joint distribution $p\left(u,v\right)$ and the product distribution $p_s(u)p_s(v)$, i.e.,
\begin{equation}
\begin{aligned}
I_s (\tilde{U};\tilde{V})=&-\sum_{i_s=1}^{{\tilde{N}}_u}\sum_{j_s=1}^{{\tilde{N}}_v}\sum_{(u_i,v_j)\in \mathcal{U}_{i_s}\times \mathcal{V}_{j_s}}p\left(u_i,v_j\right)\\
 &\cdot \log \frac{\sum_{u_i \in \mathcal{U}_{i_s}}p\left(u_i\right) \sum_{v_j \in \mathcal{V}_{j_s}}p\left(v_j\right)}{p\left(u_i,v_j\right)}\\
=&D_s\left(p\left(u,v\right)\|p_s(u)p_s(v)\right)\\
=&H_s(\tilde{U})+H_s(\tilde{V})-H(U,V).
\end{aligned}
\end{equation}

Correspondingly, the full semantic mutual information $\tilde{I}_s(\tilde{U};\tilde{V})$ is defined as,
\begin{equation}
\begin{aligned}
\tilde{I}_s (\tilde{U};\tilde{V})=&-\sum_{i_s=1}^{{\tilde{N}}_u}\sum_{j_s=1}^{{\tilde{N}}_v}\sum_{(u_i,v_j)\in \mathcal{U}_{i_s}\times \mathcal{V}_{j_s}}p\left(u_i,v_j\right) \\
&\cdot \log \frac{\sum_{u_i \in \mathcal{U}_{i_s}}p\left(u_i\right) \sum_{v_j \in \mathcal{V}_{j_s}}p\left(v_j\right)}{\sum_{(u_i,v_j)\in \mathcal{U}_{i_s}\times \mathcal{V}_{j_s}}p\left(u_i,v_j\right)}\\
=&D_s\left(p_s\left(u,v\right)\|p_s(u)p_s(v)\right)\\
=&H_s(\tilde{U})+H_s(\tilde{V})-H_s(\tilde{U},\tilde{V}).
\end{aligned}
\end{equation}
\end{definition}

\begin{theorem}
(Non-negativity of semantic mutual information): For any two semantic variables, $\tilde{U}$, $\tilde{V}$, given the jointly synonymous mapping $f_{uv}$, we have
\begin{align}
\left\{\begin{aligned}
&\tilde{I}_s(\tilde{U};\tilde{V})\geq 0\\
&I^s (\tilde{U};\tilde{V} )\geq H(U,V)-H_s(\tilde{U},\tilde{V})\geq 0\\
&I_s (\tilde{U};\tilde{V} )\geq H_s(\tilde{U},\tilde{V})-H(U,V)
\end{aligned}
\right.
\end{align}
with equality if and only if 
\begin{equation}
\left\{\begin{aligned}
p_s\left(u_i\right) p_s\left(v_j\right)&=p_s\left(u_i,v_j\right)\\
p\left(u_i\right)p\left(v_j\right)&=p\left(u_i,v_j\right)\\
p_s\left(u_i\right) p_s\left(v_j\right)&=p_s\left(u_i,v_j\right)\\
\end{aligned}\right.
\end{equation}
\end{theorem}
\begin{proof}
By using Theorem \ref{theorem3}, we obtain $\tilde{I}_s (\tilde{U};\tilde{V})=D_s\left(p_s\left(u,v\right)\|p_s(u)p_s(v)\right)\geq 0$. So the equality holds if and only if $p_s\left(u_i\right) p_s\left(v_j\right)=p_s\left(u_i,v_j\right)$, which means $\tilde{U}$ and $\tilde{V}$ are independent. Similarly, we can prove the other two inequalities.
\end{proof}

Note that the down semantic mutual information $I_s (\tilde{U};\tilde{V} )$ may be negative. Considering the practical case, we can set $(I_s (\tilde{U};\tilde{V} ))^{+}$.
\begin{theorem}\label{SemMI_concave_convex}
The up/down semantic mutual information is a concave function of $p(u)$ for fixed $p(v|u)$ and a convex function of $p(v|u)$ for fixed $p(x)$.
\end{theorem}

This theorem is proved in Appendix \ref{proof_SemMI_concabe_convex}. We now show the relationship among these mutual information measures.
\begin{corollary}\label{corollary3}
\begin{equation}
\begin{aligned}
I_s (\tilde{U};\tilde{V}) \leq I (U;V )\leq H (V )-H_s (\tilde{V}\left|U\right. ) \leq I^s (\tilde{U};\tilde{V} )
\end{aligned}
\end{equation}
\end{corollary}
\begin{proof}
To prove the left inequality, by Lemma \ref{lemma2}, we can write
\begin{equation}
I_s (\tilde{U};\tilde{V} )-I (U;V )=H_s(\tilde{U})+H_s(\tilde{V})-H(U)-H(V)\leq 0.
\end{equation}
For the medium inequality, due to Lemma \ref{lemma4}, we have
\begin{equation}
I (U;V )-  (H (V )-H_s (\tilde{V}\left|U\right. ) )=-H (V|U )+H_s (\tilde{V}\left|U\right. )\leq 0.
\end{equation}
For the right inequality, due to Theorem \ref{theorem2}, we have
\begin{equation}
\begin{aligned}
 &H (V )-H_s (\tilde{V}\left|U\right. )- I^s (\tilde{U};\tilde{V} )\\
 &= -H_s (\tilde{V}\left|U\right. )-H(U)+H_s(\tilde{U},\tilde{V})\leq 0.
\end{aligned}
\end{equation}
\end{proof}

\begin{corollary}
\begin{equation}
\begin{aligned}
I_s(\tilde{U};\tilde{V}) \leq \tilde{I}_s (\tilde{U};\tilde{V}) \leq H(V)-H_s(\tilde{V}\left|U\right.) \leq I^s(\tilde{U};\tilde{V})
\end{aligned}
\end{equation}
\end{corollary}
\begin{proof}
By Definition \ref{definition7} and Lemma \ref{lemma3}, we can write the left inequality as
\begin{equation}
I_s(\tilde{U};\tilde{V})-\tilde{I}_s(\tilde{U};\tilde{V})=-H(U,V)+H_s(\tilde{U},\tilde{V})\leq 0.
\end{equation}
Correspondingly, by Definition \ref{definition7}, Lemma \ref{lemma2}, and Theorem \ref{theorem2}, we can also write the medium inequality as
\begin{equation}
\begin{aligned}
&\tilde{I}_s (\tilde{U};\tilde{V})-(H(V)-H_s(\tilde{V}\left|U\right.))\\
&=H_s(\tilde{U})+H_s(\tilde{V})-H_s(\tilde{U},\tilde{V})-H(V)+H_s (\tilde{V}\left|U\right. )\\
&\leq H_s(\tilde{U})+H_s (\tilde{V}\left|U\right. )-H_s(\tilde{U},\tilde{V})\leq 0.
\end{aligned}
\end{equation}
\end{proof}

\begin{theorem}\label{chainrule_SMI}
(\textit{Chain Rule of Sequential Mutual Information}):

Given a pair of semantic sequences $(\tilde {U}^n,\tilde {V}^n)$ and the associated pair of syntactic sequences $\left(U^n,V^n\right)$, under the sequential synonymous mapping $f_{uv}^n$ and conditional synonymous mappings $f_{u_k,v^n|u_1^{k-1}}$, for the sequential version of down semantic mutual information, we have
\begin{equation}
\begin{aligned}
I_s(\tilde{U}^n;\tilde{V}^n)&\leq I_s(\tilde{U}_1^{n-1},U_n;\tilde{V}^n)\\
& \leq \cdots \leq I_s(\tilde{U}_1^m,U_{m+1}^n;\tilde{V}^n) \\
&\leq \cdots \leq I_s(\tilde{U}_1,U_{2}^n;\tilde{V}^n) \\
& \leq I(U^n;V^n)=\sum_{k=1}^{n} I(U_k;V^n\left|U_1^{k-1}\right.),
\end{aligned}
\end{equation}
where $I_s(\tilde{U}_1^m,U_{m+1}^n;\tilde{V}^n)=\tilde{H}(\tilde{U}_1^m,U_{m+1}^n)+H_s(\tilde{V}^n)-H(U^n,V^n)$ with $m=n-1,\cdots,2,1$.

Similarly, for the sequential version of up semantic mutual information, we have
\begin{equation}
\begin{aligned}
I(U^n;V^n)&=\sum_{k=1}^{n} I(U_k;V^n\left|U_1^{k-1}\right.)\\
&\leq I^s(\tilde{U}_1,U_{2}^n;\tilde{V}^n)\\
& \leq \cdots \leq I^s(\tilde{U}_1^m,U_{m+1}^n;\tilde{V}^n)\\
& \leq \cdots \leq I^s(\tilde{U}_1^{n-1},U_n;\tilde{V}^n)\leq I^s(\tilde{U}^n;\tilde{V}^n),
\end{aligned}
\end{equation}
where $I^s(\tilde{U}_1^m,U_{m+1}^n;\tilde{V}^n) =H(U^n)+H(V^n)-\tilde{H}(\tilde{U}_1^m,U_{m+1}^n,\tilde{V}^n)$ with $m=1,2,\cdots,n$.
\end{theorem}

\begin{proof}
By Theorem \ref{theorem_chainrule}, we derive that
\begin{equation}
\begin{aligned}
&I_s(\tilde{U}^n;\tilde{V}^n)=H_s(\tilde{U}^n)+H_s(\tilde{V}^n)-H(U^n,V^n)\\
&\leq H_s(\tilde{U}_1^{n-1}, U_n)+H_s(\tilde{V}^n)-H(U^n,V^n)\\
&\leq \cdots \leq H_s(\tilde{U}_1^m, U_{m+1}^n)+H_s(\tilde{V}^n)-H(U^n,V^n)\\
&\leq H_s(\tilde{U}_1^{m-1}, U_{m}^n)+H_s(\tilde{V}^n)-H(U^n,V^n)\\
&\leq \cdots \leq H_s(\tilde{U}_1, U_2^n)+H_s(\tilde{V}^n)-H(U^n,V^n)\\
&\leq H(U^n)+H(V^n)-H(U^n,V^n) = I(U^n;V^n).
\end{aligned}
\end{equation}
So we prove the first chain of inequalities.

Similarly, by Theorem \ref{theorem_chainrule}, we can attain that
\begin{equation}
\begin{aligned}
&I(U^n;V^n)=H(U^n)+H(V^n)-H(U^n,V^n)\\
&\leq H(U^n)+H(V^n)-\tilde{H}(\tilde{U}_1,U_2^n,\tilde{V}^n)\\
&\leq \cdots \leq H(U^n)+H(V^n)-\tilde{H}(\tilde{U}_1^{m-1},U_{m}^n,\tilde{V}^n)\\
&\leq H(U^n)+H(V^n)-\tilde{H}(\tilde{U}_1^{m},U_{m+1}^n,\tilde{V}^n)\\
&\leq \cdots \leq H(U^n)+H(V^n)-\tilde{H}(\tilde{U}_1^{n-1},U_{n},\tilde{V}^n)\\
&\leq H(U^n)+H(V^n)-H_s(\tilde{U}^n,\tilde{V}^n)=I^s(\tilde{U}^n;\tilde{V}^n).
\end{aligned}
\end{equation}


So we prove the second chain of inequalities.
\end{proof}

\begin{remark}
In classic information theory, the mutual information satisfies the following equalities
\begin{equation}
\begin{aligned}
I(U;V)&=H(U)-H(U|V)\\
         &=H(V)-H(V|U)\\
         &=H(U)+H(V)-H(U;V).
\end{aligned}
\end{equation}
On the contrary, in semantic information theory, the semantic mutual information only obeys the degraded inequalities
\begin{equation}
\left\{\begin{aligned}
&I_s (\tilde{U};\tilde{V}) \leq \tilde{I}_s (\tilde{U};\tilde{V} ) \leq H(V)-H_s(\tilde{V}\left|U\right.) \leq I^s (\tilde{U};\tilde{V})\\
&I_s (\tilde{U};\tilde{V} ) \leq \tilde{I}_s (\tilde{U};\tilde{V}) \leq H(U)-H_s(\tilde{U}\left|V\right.) \leq I^s (\tilde{U};\tilde{V}).
\end{aligned}\right.
\end{equation}
Since the size relationship between mutual information and semantic mutual information is uncertain, we only use the up and down semantic mutual information to indicate the upper bound and the lower bound respectively. In fact, by Theorem \ref{theorem2}, the up semantic mutual information $I^s(\tilde{U};\tilde{V})$ can be further upper bounded by $H(U)+H(V)-H_s(\tilde{U})-H_s(\tilde{V}\left|U\right.)$. However, since this bound has not a clear physical meaning, we still use $I^s(\tilde{U};\tilde{V})$ as the upper bound.
\end{remark}

\begin{theorem}
(\textit{Semantic Processing Inequality}):

Given a Markov chain $U\to V \to W$, then we have
\begin{equation}
I_s(\tilde{U};\tilde{W})\leq I_s(\tilde{U};W)\leq I(U;W)\leq I(U;V) \leq I^s(\tilde{U};V) \leq I^s(\tilde{U};\tilde{V}).
\end{equation}
\end{theorem}
\begin{proof}
By the chain rule of semantic mutual information (Theorem \ref{chainrule_SMI}), we can expand (down) mutual information in two different ways,
\begin{equation}
\begin{aligned}
I(U;VW)&=I(U;V)+I(U;W|V)\\
             &=I(U;W)+I(U;V|W)\\
             &\geq I_s(\tilde{U};W)+I(U;V|W)\\
             &\geq I_s(\tilde{U};\tilde{W}).
\end{aligned}
\end{equation}
Since $U$ and $W$ are conditionally independent give $V$, we have $I(U;W|V)=0$. Due to $I(U;V|W)\geq 0$, we have
\begin{equation}
I(U;V)\geq I(U;W) \geq I_s(\tilde{U};W) \geq I_s(\tilde{U};\tilde{W}).
\end{equation}

Similarly, by the same chain rule, we can expand (up) mutual information in two different ways,
\begin{equation}
\begin{aligned}
I(U;VW)&=I(U;V)+I(U;W|V)\\
             &=I(U;W)+I(U;V|W)\\
             &\leq I^s(\tilde{U};V)+I(U;W|V)
             &\leq I^s(\tilde{U};\tilde{V}).
\end{aligned}
\end{equation}
Thus, we have
\begin{equation}
I(U;W) \leq I(U;V)\leq I^s(\tilde{U};V) \leq I^s(\tilde{U};\tilde{V})
\end{equation}
and complete the proof.
\end{proof}

\begin{example}
According to the joint distribution in Table \ref{JPDF_RV}, we calculate the mutual information between $U$ and $V$ as $I(U;V)=H(U)+H(V)-H(U,V)=0.6332 \text{ bits}$. Correspondingly, by using the joint distribution in Table \ref{JSmapping_SRV}, we can compute the up semantic mutual information as $I^s(\tilde{U};\tilde{V})=H(U)+H(V)-H_s(\tilde{U};\tilde{V})=1.5087 \text{ sebits}$.

Similarly, by using the distribution in Table \ref{Smapping_SRV}, we compute the down semantic mutual information as $I_s(\tilde{U};\tilde{V})=H_s(U)+H_s(V)-H(U;V)=-0.6422 \text{ sebits}$. Consider the non-negative requirement, we set $I_s(\tilde{U};\tilde{V})=\max\{-0.6422,0\}=0 \text{ sebits}$. By using the distribution in Table \ref{CSmapping_SRV} and Table \ref{CSmapping_SRV2}, we calculate that $H(V)-H_s(\tilde{V}\left|U\right.)=0.7709 \text{ sebits}$ and $H(U)-H_s(\tilde{U}\left|V\right.)=1.3087 \text{ sebits}$. So we conclude that $I_s(\tilde{U};\tilde{V})=0 \text{ sebits}<I(U;V)=0.6332 \text{ bits}<H(V)-H_s(\tilde{V}\left|U\right.)=0.7709 \text{ sebits}<H(U)-H_s(\tilde{U}\left|V\right.)=1.3087 \text{ sebits}<I^s(\tilde{U};\tilde{V})=1.5087 \text{ sebits}$.
\end{example}

\section{Semantic Channel Capacity and Semantic Rate Distortion}
\label{section_V}
In this section, we introduce the semantic channel capacity and semantic rate distortion. We use the maximum up semantic mutual information to indicate the former. Similarly, the minimum down semantic mutual information is used to present the latter.
\subsection{Semantic Channel Capacity}
Generally, the discrete channel of semantic communication can be modeled as the form of a five-tuple, that is, $\left\{\tilde{\mathcal{X}},\mathcal{X},\mathcal{Y},\tilde{\mathcal{Y}}, p(Y|X)\right\}$. Here, $\mathcal{X}$ and $\mathcal{Y}$ are the input and output syntactical alphabet. And $\tilde{\mathcal{X}}$ and $\tilde{\mathcal{Y}}$ are the corresponding input and output semantic alphabet. Furthermore, $p(Y|X)$ is the channel transition probability.

\begin{definition}
Given a discrete memoryless channel $\left\{\tilde{\mathcal{X}},\mathcal{X},\mathcal{Y},\tilde{\mathcal{Y}}, p(Y|X)\right\}$, the semantic channel capacity is defined as
\begin{equation}
\begin{aligned}
C_s=&\max_{f_{xy}}\max_{p(x)} I^s(\tilde{X};\tilde{Y})\\
=&\max_{f_{xy}}\max_{p(x)}\sum_{i_s}\sum_{j_s}\sum_{(x_i,y_j)\in \mathcal{X}_{i_s}\times \mathcal{Y}_{j_s}}p\left(x_i\right)p\left(y_j\left|x_i\right.\right) \\
&\cdot \log \frac{\sum_{(x_i,y_j) \in \mathcal{X}_{i_s} \times \mathcal{Y}_{j_s}}p\left(x_i,y_j\right)}{p\left(x_i\right)p\left(y_j\right)}\\
=&\max_{f_{xy}}\max_{p(x)}\left[H(X)+H(Y)-H_s(\tilde{X},\tilde{Y})\right]
\end{aligned}
\end{equation}
where the maximum is taken over all possible input distribution $p(x)$ and the jointly synonymous mapping $f_{xy}:\tilde{\mathcal{X}}\times\tilde{\mathcal{Y}}\to\mathcal{X}\times\mathcal{Y}$.
\end{definition}

By Corollary \ref{corollary3}, the channel mutual information is no more than the up semantic mutual information, that is, $I(X;Y)\leq I^s(\tilde{X};\tilde{Y})$. Thus, we conclude that the channel capacity is no more than the semantic capacity, i.e., $C\leq C_s$.

\subsection{Semantic Rate Distortion}
Similar to the classic rate distortion theory, given a semantic/syntactic source distribution and a distortion measure, we also need to investigate the minimum rate description required to achieve a specific distortion.
\begin{definition}
Given a discrete source $X\sim p(x), x\in\mathcal{X}$ and the associated semantic source $\tilde{X}$, the decoder outputs an estimate $\hat{X}$ with the associated semantic variable $\hat{\tilde{X}}$. Under the synonymous mapping $f_{x}$ and $f_{\hat{x}}$, the semantic distortion measure is a mapping
\begin{equation}
d_s: \tilde{\mathcal{X}}\times \hat{\tilde{\mathcal{X}}}\to \mathbb{R}^{+}
\end{equation}
from the Cartesian product of semantic source alphabet and reconstruction alphabet into the set of non-negative real numbers. The semantic distortion $d_s(\tilde{x}_{i_s},\hat{\tilde{x}}_{j_s})=d_s(\mathcal{X}_{i_s},\hat{\mathcal{X}}_{j_s})$ is a measure of the cost of representing the semantic symbol $\tilde{x}_{i_s}$ by the symbol $\hat{\tilde{x}}_{j_s}$, equivalently, which is a cost of representing the syntactical symbol set $\mathcal{X}_{i_s}$ by the reproduction set $\hat{\mathcal{X}}_{j_s}$. For the traditional lossy source coding, the semantic distortion can be measured using Hamming distortion or mean squared error (MSE) distortion. On the other hand, for the lossy coding with neural network model, the semantic distortion can also be evaluated by using word error rate (WER), structure similarity index measure (SSIM), learned perceptual image patch similarity (LPIPS) etc.

So the average semantic distortion is defined as
\begin{equation}
\begin{aligned}
\bar{d}_s&=\mathbb{E}\left[d_s(\tilde{x},\hat{\tilde{x}})\right]\\
&=\sum_{i_s}\sum_{j_s}\sum_{(x_i,\hat{x}_j)\in \mathcal{X}_{i_s}\times \hat{\mathcal{X}}_{j_s}}p\left(x_i,\hat{x}_j\right)d_s(\mathcal{X}_{i_s},\hat{\mathcal{X}}_{j_s}).
\end{aligned}
\end{equation}
Furthermore, the test channel set $P_D$ is defined as
\begin{equation}
P_D=\left\{p(\hat{x}\left|x\right.): \bar{d}_s=\mathbb{E}\left[d_s\left(\tilde{x},\hat{\tilde{x}}\right)\right] \leq D\right\}.
\end{equation}
\end{definition}

Next, we give the formal definition of semantic rate distortion.
\begin{definition}
Given an i.i.d. source $X$ with distribution $p(x)$, the associated semantic source $\tilde{X}$, and the semantic distortion function $d_s(\tilde{x}_{i_s},\hat{\tilde{x}}_{j_s})$, the semantic rate distortion is defined as,
\begin{equation}
\begin{aligned}
R_s(D)=&\min_{f_{x},f_{\hat{x}}}\min_{p(\hat{x}\left|x\right.)\in P_D}I_s(\tilde{X};\hat{\tilde{X}})\\
=&\min_{f_{x},f_{\hat{x}}}\min_{p(\hat{x}\left|x\right.)\in P_D} \sum_{i_s}\sum_{j_s}\sum_{(x_i,\hat{x}_j)\in {\mathcal{X}_{i_s}}\times {\hat{\mathcal{X}}_{j_s}}}p\left(x_i,\hat{x}_j\right) \\
&\cdot \log \frac{p\left(x_i,\hat{x}_j\right)}{\sum_{x_i \in {\mathcal{X}_{i_s}}}p\left(x_i\right) \sum_{\hat{x}_j \in {\hat{\mathcal{X}}_{j_s}}}p\left(\hat{x}_j\right)}\\
=&\min_{f_{x},f_{\hat{x}}}\min_{p(\hat{x}\left|x\right.)\in P_D} \left[H_s(\tilde{X})+H_s(\hat{\tilde{X}})-H(X,\hat{X})\right].
\end{aligned}
\end{equation}
\end{definition}

\begin{lemma}(Convexity of $R_s(D)$):
The semantic rate distortion function $R_s(D)$ is a non-increasing convex function of $D$.
\end{lemma}
\begin{proof}
Similar to syntactic rate distortion $R(D)$, $R_s(D)$ is also a non-increasing function in $D$.

Consider two rate distortion pairs $(R_1,D_1)$ and $(R_2,D_2)$, the corresponding joint distribution achieving these pair are $p_1(x,\hat{x})=p(x)p_1(\hat{x}|x)$ and $p_2(x,\hat{x})=p(x)p_2(\hat{x}|x)$ respectively. Consider the weighted distribution $p_{\theta}=\theta p_1+(1-\theta) p_2$ and the weighted distortion $D_\theta=\theta D_1+(1-\theta) D_2$. Due to the convexity of semantic mutual information, we have
\begin{equation}
\begin{aligned}
R_s(D_\theta) &\leq I_{s,p_\theta}(\tilde{X};\hat{\tilde{X}})\\
                          &\leq \theta I_{s,p_1}(\tilde{X};\hat{\tilde{X}})+(1-\theta) I_{s,p_2}(\tilde{X};\hat{\tilde{X}})\\
                          &=\theta R_s(D_1)+(1-\theta) R_s(D_2).
\end{aligned}
\end{equation}
So we complete the proof.
\end{proof}

According to Corollary \ref{corollary3}, the down semantic mutual information is no more than the mutual information, that is, $I_s(\tilde{X};\tilde{Y})\leq I(X;Y) $. Thus, we conclude that the semantic rate distortion is no more than the classic rate distortion, i.e., $R_s(D)\leq R(D)$.

\section{Semantic Lossless Source Coding}
\label{section_VI}
In this section, we discuss the semantic lossless source coding. First, we investigate the asymptotic equipartition property (AEP) of semantic coding and introduce the synonymous typical set. Then we prove the semantic source coding theorem and give the optimal length of semantic coding. Finally, we design the semantic Huffman coding to demonstrate the advantage of semantic data compression.

\subsection{Asymptotic Equipartition Property and Synonymous Typical Set}
Similar to classic information theory, asymptotic equipartition property (AEP) is also an important tool to prove the coding theorem in semantic information theory.

\begin{definition}
Given a discrete random variable $U$ with the distribution $p(u)$ and the corresponding semantic variable $\tilde {U}$, when we consider an i.i.d. semantic sequence $(\tilde {U}_1,\tilde {U}_2,\cdots,\tilde {U}_n)$ and the associated syntactic sequences $\left(U_1,U_2,\cdots,U_n\right)$, the sequential synonymous mapping is defined as $f^{n}: \tilde{\mathcal{U}}^{n}\to\mathcal{U}^{n}$. That is to say, given a sequence $\tilde {u}^n$, we have $f^n\left(\tilde {u}^n\right)=\prod_{k=1}^{n}\mathcal{U}_{\tilde{u}_{k}}$.
\end{definition}

\begin{theorem}\label{theorem5}
(Semantic and Syntactic AEP):
If $(\tilde {U}_1,\tilde {U}_2,\cdots,\tilde {U}_n,\cdots)$ is an i.i.d. semantic sequence and the associated syntactic sequence is $\left(U_1,U_2,\cdots,U_n,\cdots\right)$, given the sequential synonymous mapping $f^{n}: \tilde{\mathcal{U}}^{n}\to\mathcal{U}^{n}$, then
\begin{equation}
\left\{\begin{aligned}
&\lim_{n\to\infty}-\frac{1}{n}\log p(\tilde {U}_1,\tilde {U}_2,\cdots,\tilde {U}_n)=H_s(\tilde {U})\\
&\lim_{n\to\infty}-\frac{1}{n}\log p\left(U_1,U_2,\cdots,U_n\right)=H\left(U\right)
\end{aligned}
\right.
\end{equation}
\end{theorem}
\begin{proof}
For the first equality, since the $\tilde{U}_{k}$ are i.i.d., so are $\log p(\tilde{U}_{k})$. By the weak law of large numbers, we have
\begin{equation}
\begin{aligned}
-\frac{1}{n}\log p(\tilde {U}_1,\tilde {U}_2,\cdots,\tilde {U}_n)&=-\frac{1}{n}\sum_{k=1}^{n}\log p(\tilde{U}_{k})\\
&\to \mathbb{E} \left[-\log p(\tilde{U})\right]\\
&=\mathbb{E} \left[-\log \sum_{u\in \mathcal{U}_{\tilde{u}}}p(u)\right]\\
&=H_s(\tilde{U}).
\end{aligned}
\end{equation}
The second equality holds from the classic information theory \cite[Theorem 3.1.1]{Book_Cover}.
 \end{proof}

Base on semantic and syntactic AEP, under the sequential synonymous mapping $f^n$, we have the following synonymous AEP.
\begin{theorem}
(Synonymous AEP):
If $(\tilde {U}_1,\tilde {U}_2,\cdots,\tilde {U}_n,\cdots)$ is an i.i.d. semantic sequence and the associated syntactic sequence is $\left(U_1,U_2,\cdots,U_n,\cdots\right)$, under the sequential synonymous mapping $f^n: \tilde{\mathcal{U}}^n\to\mathcal{U}^n$, we have
\begin{equation}
\begin{aligned}
&\lim_{n\to\infty}-\frac{1}{n}\left[\log p\left(U_1,U_2,\cdots,U_n\right)-\log p(\tilde {U}_1,\tilde {U}_2,\cdots,\tilde {U}_n)\right]\\
&=H\left(U\right)-H_s(\tilde {U}).
\end{aligned}
\end{equation}
\end{theorem}
\begin{proof}
By Theorem \ref{theorem5}, we have
\begin{equation}
\begin{aligned}
&-\frac{1}{n}\left[\log p\left(U_1,U_2,\cdots,U_n\right)-\log p(\tilde {U}_1,\tilde {U}_2,\cdots,\tilde {U}_n)\right]\\
&=-\frac{1}{n}\sum_{k=1}^{n}\left[\log p(U_{k})-\log p(\tilde{U}_{k})\right]\\
&\to \mathbb{E} \left[-\log p(U)\right]-\mathbb{E} \left[-\log p(\tilde{U})\right]\\
&=\mathbb{E}\left[-\log p(u)\right]-\mathbb{E} \left[-\log \sum_{u\in \mathcal{U}_{\tilde{u}}}p(u)\right]\\
&=H(U)-H_s(\tilde{U}).
\end{aligned}
\end{equation}
\end{proof}

Similar to the definition of syntactically typical set $A_{\epsilon}^{(n)}$ in \cite{Book_Cover}, we now introduce the semantically typical set $\tilde{A}_{\epsilon}^{(n)}$ and the synonymous typical set $B_{\epsilon}^{(n)}\left(\tilde{u}^n\right)$ respectively.
\begin{definition}\label{definition12}
Given the typical sequences $\left\{\tilde{u}^n\right\}$ with respect to the distribution $p(u)$, the semantically typical set ${\tilde{A}}_{\epsilon}^{(n)}$ is defined as the set of $n$-sequences with empirical semantic entropy $\epsilon$-close to the true semantic entropy, that is,
\begin{equation}
{\tilde{A}}_{\epsilon}^{(n)}=\left\{\tilde{u}^n\in \tilde{\mathcal{U}}^n: \left|-\frac{1}{n}\log p\left(\tilde{u}^n\right)-H_s(\tilde{U})\right|<\epsilon\right\},
\end{equation}
where
\begin{equation}
p\left(\tilde{u}^n\right)=\prod_{k=1}^{n} p(\tilde{u}_k)=\prod_{k=1}^{n} \sum_{u_k\in \mathcal{U}_{\tilde{u}_k}}p(u_k).
\end{equation}
\end{definition}

\begin{definition}
Given a specific typical sequence $\left\{\tilde{u}^n\right\}$, the synonymous typical set $B_{\epsilon}^{(n)}\left(\tilde{u}^n\right)$ with the syntactically typical sequences $\left\{u^n\right\}$ is defined as the set of $n$-sequences with the difference of empirical entropies  $\epsilon$-close to the difference of true entropies, that is,
\begin{equation}
\begin{aligned}
B_{\epsilon}^{(n)}\left(\tilde{u}^n\right)=\bigg\{ u^n\in \mathcal{U}^n: &\left|-\frac{1}{n}\log p\left(u^n\right)-H(U)\right|<\epsilon,\\
&\left|-\frac{1}{n}\log p\left(\tilde{u}^n\right)-H_s(\tilde{U})\right|<\epsilon,\\
&\left|-\frac{1}{n}\log p(\tilde{u}^n\to u^n)\right.\\
&\left.\left.-\left(H(U)-H_s(\tilde{U})\right)\right|<\epsilon\right\},
\end{aligned}
\end{equation}
where $p(u^n)\!\!=\!\!\prod_{k=1}^{n} p(u_k)$ and $ p(\tilde{u}^n)\!\!=\!\!\prod_{k=1}^{n} \sum_{u_k\in \mathcal{U}_{\tilde{u}_k}}\!\!\!\!p(u_k)$. Here the probability $p(\tilde{u}^n\to u^n)$ is defined as
\begin{equation}
p(\tilde{u}^n\to u^n)=\left\{
\begin{aligned}
&\frac{p\left(u^n\right)}{p\left(\tilde{u}^n\right)},&\text{ if } u^n=f^n(\tilde{u}^n),\\
&0, &\text{otherwise}.
\end{aligned}\right.
\end{equation}
\end{definition}

Under the sequential synonymous mapping $f^n:\tilde{\mathcal{U}}^n\to{\mathcal{U}}^n$, the syntactically typical set $A_{\epsilon}^{(n)}$ can be further partitioned into multiple synonymous typical sets $B_{\epsilon}^{(n)}\left(\tilde{u}^n\right)$, that is,
\begin{equation}
A_{\epsilon}^{(n)}=\bigcup_{\tilde{u}^n\in \tilde{A}_{\epsilon}^{(n)}} B_{\epsilon}^{(n)}\left(\tilde{u}^n\right).
\end{equation}
And for $\forall \tilde{u}^n$, $\tilde{v}^n \in \tilde{A}_{\epsilon}^{(n)}$, $\tilde{u}^n \neq \tilde{v}^n$, we have $B_{\epsilon}^{(n)}\left(\tilde{u}^n\right) \bigcap B_{\epsilon}^{(n)}\left(\tilde{v}^n\right)=\varnothing$.

Essentially, $B_{\epsilon}^{(n)}\left(\tilde{u}^n\right)$ is an equivalence class of the synonymous typical sequences. Thus for the syntactically typical set $A_{\epsilon}^{(n)}$, we can obtain a quotient set $A_{\epsilon}^{(n)}/f^n=\left\{B_{\epsilon}^{(n)}(\tilde{u}^n)\right\}$ and construct an one-to-one mapping between $\tilde{A}_{\epsilon}^{(n)}$ and $A_{\epsilon}^{(n)}/f^n$.

We now discuss the properties of syntactically and semantically typical sets. In order to the convenient reading, we first rewrite the properties of syntactically typical set as follows.

\begin{theorem}
(\cite[Theorem 3.1.2]{Book_Cover}):
\begin{enumerate}[(1)]
    \item If $\left(u_1,u_2,\cdots,u_n\right)\in A_{\epsilon}^{(n)}$, then $H(U)-\epsilon \leq -\frac{1}{n}\log p\left(u_1,u_2,\cdots,u_n\right) \leq H(U)+\epsilon$ or equivalently
         \begin{equation}\label{AEPinequality}
        2^{-n\left(H(U)+\epsilon\right)}\leq p\left(u_1,u_2,\cdots,u_n\right) \leq 2^{-n\left(H(U)-\epsilon\right)}.
        \end{equation}
    \item $\text{Pr}\left\{A_{\epsilon}^{(n)}\right\}>1-\epsilon$ for $n$ sufficiently large.
    \item $\left(1-\epsilon\right) 2^{n\left(H(U)-\epsilon\right)}\leq \left|A_{\epsilon}^{(n)}\right| \leq 2^{n\left(H(U)+\epsilon\right)}$ for $n$ sufficiently large.
  \end{enumerate}
\end{theorem}

Then as the consequence of the semantic AEP, we give the properties of semantically typical set as below.
\begin{theorem}\label{theorem8}
\text{ }
\begin{enumerate}[(1)]
    \item If $\left(\tilde{u}_1,\tilde{u}_2,\cdots,\tilde{u}_n\right)\in \tilde{A}_{\epsilon}^{(n)}$, then $H_s(\tilde{U})-\epsilon \leq -\frac{1}{n}\log p\left(\tilde{u}_1,\tilde{u}_2,\cdots,\tilde{u}_n\right) \leq H_s(\tilde{U})+\epsilon$, equivalently,
        \begin{equation}\label{SAEPinequality}
        2^{-n\left(H_s(\tilde{U})+\epsilon\right)}\leq p\left(\tilde{u}_1,\tilde{u}_2,\cdots,\tilde{u}_n\right) \leq 2^{-n\left(H_s(\tilde{U})-\epsilon\right)}.
        \end{equation}
    \item $\text{Pr}\left\{\tilde{A}_{\epsilon}^{(n)}\right\}>1-\epsilon$ for sufficiently large $n$.
    \item $\left(1-\epsilon\right) 2^{n\left(H_s(\tilde{U})-\epsilon\right)}\leq\left|\tilde{A}_{\epsilon}^{(n)}\right| \leq 2^{n\left(H_s(\tilde{U})+\epsilon\right)}$ for sufficiently large $n$.
\end{enumerate}
\end{theorem}
\begin{proof}
The proof of property (1) is directly from the definition of $\tilde{A}_{\epsilon}^{(n)}$. By Theorem \ref{theorem5}, the probability of the event $\tilde{U}^n \in \tilde{A}_{\epsilon}^{(n)}$ tends to $1$ as $n\to \infty$. So for $\forall \epsilon>0$, $\exists n_0$, for $\forall n \geq n_0$, we have
\begin{equation}
\text{Pr}\left\{ \left|-\frac{1}{n}\log p(\tilde{U}^n)-H_s(\tilde{U})\right|<\epsilon\right\}>1-\epsilon.
\end{equation}
Then we complete the proof of property (2).

After summing over the set $\tilde{A}_{\epsilon}^{(n)}$, by using property (2), Eq. (\ref{SAEPinequality}) can be rewritten as
\begin{equation}
\left\{ \begin{aligned}
&\left|\tilde{A}_{\epsilon}^{(n)}\right| 2^{-n\left(H_s(\tilde{U})+\epsilon\right)}\leq \sum_{\tilde{u}^n\in \tilde{A}_{\epsilon}^{(n)}} p\left(\tilde{u}^n\right)\leq 1\\
&1-\epsilon \leq \sum_{\tilde{u}^n\in \tilde{A}_{\epsilon}^{(n)}} p\left(\tilde{u}^n\right) \leq \left|\tilde{A}_{\epsilon}^{(n)}\right| 2^{-n\left(H_s(\tilde{U})-\epsilon\right)}
\end{aligned}\right.
\end{equation}

Hence, we can write
\begin{equation}\label{equation61}
\left|\tilde{A}_{\epsilon}^{(n)}\right| \geq \left(1-\epsilon\right) 2^{n\left(H_s(\tilde{U})-\epsilon\right)}
\end{equation}
and
\begin{equation}\label{equation62}
\left|\tilde{A}_{\epsilon}^{(n)}\right| \leq 2^{n\left(H_s(\tilde{U})+\epsilon\right)}
\end{equation}
respectively and complete the third property.
\end{proof}

Furthermore, we give the properties of synonymous typical set as below.
\begin{theorem}\label{SynTSet_Theorem}
\text{ }
\begin{enumerate}[(1)]
    \item Given a semantic sequence $\left(\tilde{u}_1,\tilde{u}_2,\cdots,\tilde{u}_n\right)\in \tilde{A}_{\epsilon}^{(n)}$, if $\left(u_1, u_2,\cdots,u_n\right)\in B_{\epsilon}^{(n)}\left(\tilde{u}^n\right)$, then $H(U)-H_s(\tilde{U})-\epsilon \leq -\frac{1}{n}\log \frac{p\left(u^n\right)}{p\left(\tilde{u}^n\right)} \leq H(U)-H_s(\tilde{U})+\epsilon$, equivalently,
        \begin{equation}
        2^{-n\left(H(U)-H_s(\tilde{U})+\epsilon\right)}\leq \frac{p\left(u^n\right)}{p\left(\tilde{u}^n\right)} \leq 2^{-n\left(H(U)-H_s(\tilde{U})-\epsilon\right)}.
        \end{equation}
    \item $ 2^{n\left(H(U)-H_s(\tilde{U})-\epsilon\right)}\leq\left|B_{\epsilon}^{(n)}\left(\tilde{u}^n\right)\right| \leq 2^{n\left(H(U)-H_s(\tilde{U})+\epsilon\right)}$ for sufficiently large $n$.
\end{enumerate}
\end{theorem}
\begin{proof}
The proof of property (1) is directly from the definition of $B_{\epsilon}^{(n)}\left(\tilde{u}^n\right)$.

To prove the left inequality of property (2), we write
\begin{equation}
\begin{aligned}
p\left(\tilde{u}^n\right)&=\sum_{u^n\in B_{\epsilon}^{(n)}\left(\tilde{u}^n\right)} p\left(u^n\right) \\
&\leq p\left(\tilde{u}^n\right) \sum_{u^n \in B_{\epsilon}^{(n)}\left(\tilde{u}^n\right)}2^{-n\left(H(U)-H_s(\tilde{U})-\epsilon\right)}\\
& =p\left(\tilde{u}^n\right) \left|B_{\epsilon}^{(n)}\left(\tilde{u}^n\right)\right| 2^{-n\left(H(U)-H_s(\tilde{U})-\epsilon\right)}.
\end{aligned}
\end{equation}
So we have
\begin{equation}
\left|B_{\epsilon}^{(n)}\left(\tilde{u}^n\right)\right|\geq 2^{n\left(H(U)-H_s(\tilde{U})-\epsilon\right)}.
\end{equation}

On the other hand, for the right inequality, we can write
\begin{equation}
\begin{aligned}
p\left(\tilde{u}^n\right)&=\sum_{u^n\in B_{\epsilon}^{(n)}\left(\tilde{u}^n\right)} p\left(u^n\right) \\
&\geq p\left(\tilde{u}^n\right) \sum_{u^n \in B_{\epsilon}^{(n)}\left(\tilde{u}^n\right)}2^{-n\left(H(U)-H_s(\tilde{U})+\epsilon\right)}\\
& =p\left(\tilde{u}^n\right) \left|B_{\epsilon}^{(n)}\left(\tilde{u}^n\right)\right| 2^{-n\left(H(U)-H_s(\tilde{U})+\epsilon\right)}.
\end{aligned}
\end{equation}
Similarly, we have
\begin{equation}
\left|B_{\epsilon}^{(n)}\left(\tilde{u}^n\right)\right| \leq 2^{n\left(H(U)-H_s(\tilde{U})+\epsilon\right)}
\end{equation}
and complete the proof.
\end{proof}

\begin{remark}
Both the probability of syntactically typical set $A_{\epsilon}^{(n)}$ and the probability of semantically typical set $\tilde{A}_{\epsilon}^{(n)}$ trends to $1$ with sufficiently large $n$. Furthermore, all the syntactically and semantically typical sequences are almost equiprobability. Therefore, the number of syntactically typical sequences is about $\left|A_{\epsilon}^{(n)}\right|\thickapprox2^{n(H(U)\pm\epsilon)}$ and that of semantically typical sequences is about $\left|\tilde{A}_{\epsilon}^{(n)}\right|\thickapprox2^{n(H_s(\tilde{U})\pm\epsilon)}$. Simultaneously, the number of synonymous typical sequences is about $\left|B_{\epsilon}^{(n)}\left(\tilde{u}^n\right)\right|\thickapprox2^{n(H(U)-H_s(\tilde{U})\pm\epsilon)}$. In fact, all the synonymous typical sets have almost the same number of typical sequences. Hence, hereafter, in the case of non-confusion, we abbreviate $B_{\epsilon}^{(n)}\left(\tilde{u}^n\right)$ to $B_{\epsilon}^{(n)}$.
\end{remark}
\subsection{Semantic Source Coding Theorem}
We now discuss the semantic source coding. As shown in Fig. \ref{Semantic_Source_coding}, in the side of transmitter, with the help of synonymous mapping $f^n$, the semantic index $i_s$ is mapped into a syntactic sequence $U^n$ and encoded into a source codeword $X^n$. In the other side, the decoder decides the codeword $X^n$ to an estimated syntactic sequence $\hat{U}^n$. After de-synonymous mapping $g^n$, we obtain an index estimation of semantic sequence $\hat{i}_s$.

\begin{figure*}[htbp]
\setlength{\abovecaptionskip}{0.cm}
\setlength{\belowcaptionskip}{-0.cm}
  \centering{\includegraphics[scale=0.8]{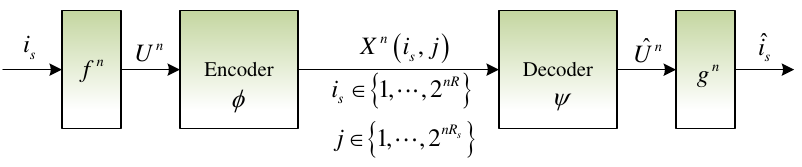}}
  \caption{Semantic lossless source encoder and decoder.}\label{Semantic_Source_coding}
\end{figure*}

\begin{definition}
An $\left(M,n\right)$ code for semantic source coding consists of the following parts:
\begin{enumerate}[(1)]
   \item A semantic index set $\mathcal{I}_s=\left\{1,\cdots,i_s,\cdots, M_s\right\}$ and a syntactic index set $\mathcal{I}=\left\{1,\cdots,i,\cdots, M\right\}$.
   \item By the synonymous mapping $f^n: \tilde{\mathcal{U}}^n \to \mathcal{U}^n$, one semantic sequence in the set $\tilde{\mathcal{U}}^n$ is mapped into a syntactic sequence.
   \item An encoding function $\phi: \mathcal{U}^n\to  \mathcal{X}^n$ generates the set of codewords, namely, codebook, $\mathcal{C}=\left\{X^n(1),X^n(2),\cdots,X^n(M)\right\}$. Due to synonymous mapping, this codebook can be partitioned into synonymous codeword subsets $\mathcal{C}_s(i_s)=\{X^n(i_s,j),i_s\in \mathcal{I}_s,j=1,2,\cdots,\frac{M}{M_s}\}$, where $X^n(i_s,j)$ denotes the $j$-th codeword of the $i_s$-th subset.
   \item A decoding function $\psi: \mathcal{X}^n \to \mathcal{U}^n$ outputs the decision syntactic sequence $\hat{U}^n$.
   \item After de-mapping, $g^n (\hat{U}^n)=\hat{i}_s$, the estimated semantic index is obtained. Note that both $\psi$ and $g^n$ are deterministic.
\end{enumerate}
\end{definition}

Under the synonymous mapping $f^n$, $\mathcal{C}_s$ is an equivalence class of the synonymous codewords. So we can construct the quotient set $\mathcal{C}/f^n=\left\{\mathcal{C}_s\right\}$. Without loss of generality, we can assume all the synonymous sets have the same number of codewords, that is, $\left|\mathcal{C}_s\right|=\frac{M}{M_s}=2^{nR_s}$, where $R_s$ is called the rate of synonymous codeword set. So we have $\left|\mathcal{C}/f^n\right|=M_s$. Let $R=\frac{1}{n}\log_2 M_s$ denote the semantic code rate. Furthermore, let $R'=R+R_s=\frac{1}{n}\log_2 M$ be the syntactic code rate.

\begin{definition}
The decoding error probability is defined as
\begin{equation}
P_e^{(n)}=\text{Pr}(g^n (\mathcal{C}_s )\neq i_s).
\end{equation}
\end{definition}

We now give the formal description of semantic source coding theorem.
\begin{theorem}
Given the semantic source $\tilde{U}$ and the syntactic source $U$ with the synonymous mapping $f:\tilde{\mathcal{U}}\to \mathcal{U}$, for each code rate $R>H_s(\tilde{U})$, there exists a series of $\left(2^{n(R+R_s)},n\right)$ codes, when code length $n$ tends to sufficiently large, the error probability is close to zero, i.e. $P_e^{(n)}\to 0$. On the contrary, if $R<H_s(\tilde{U})$, then for any $\left(2^{n(R+R_s)},n\right)$ code, the error probability tends to $1$ with $n$ sufficiently large.
\end{theorem}
\begin{proof}
First, we prove the direct part of the theorem. We select $\epsilon>0$ and construct a one-to-one mapping from $\tilde{A}_{\epsilon}^{(n)}$ to $\mathcal{C}/f^n$. So for sufficiently large $n$, by Theorem \ref{theorem8}, we have
\begin{equation}
\left(1-\epsilon\right) 2^{n\left(H_s(\tilde{U})-\epsilon\right)}\leq\left|\mathcal{C}/f^n\right|=2^{nR}=\left|\tilde{A}_{\epsilon}^{(n)}\right| \leq 2^{n\left(H_s(\tilde{U})+\epsilon\right)}.
\end{equation}

Therefore, the semantic code rate satisfies
\begin{equation}
\frac{1}{n}\log\left(1-\epsilon\right)+H_s(\tilde{U})-\epsilon \leq R \leq H_s(\tilde{U})+\epsilon .
\end{equation}
Also by Theorem \ref{theorem8}, it follows that
\begin{equation}
P_e^{(n)}=\text{Pr}(\tilde{U}^n\notin \tilde{A}_{\epsilon}^{(n)})<\epsilon.
\end{equation}

By Theorem \ref{SynTSet_Theorem}, the size of synonymous set satisfies
\begin{equation}
1\leq\left|\mathcal{C}_s\right|=2^{nR_s}\leq \left|B_{\epsilon}^{(n)}\left(\tilde{u}^n\right)\right| \leq 2^{n\left(H(U)-H_s(\tilde{U})+\epsilon\right)}.
\end{equation}
Therefore the syntactic code rate satisfies
\begin{equation}\label{equation74}
\frac{1}{n}\log\left(1-\epsilon\right)+R_s+H_s(\tilde{U})-\epsilon \leq R' \leq H_s(\tilde{U})+R_s+\epsilon .
\end{equation}

When the synonymous set $\mathcal{C}_s$ only has one codeword to represent the semantically typical sequence, that is, $R_s=0$, letting $\epsilon\to 0$, both the code rate $R$ and $R'$ tend to $H_s(\tilde{U})$, while $P_e^{(n)}$ tends to $0$. On the other hand, substituting $R_s=H(U)-H_s(\tilde{U})$ into (\ref{equation74}) and letting $\epsilon\to 0$, $R\to H_s(\tilde{U})$ and $R'\to H(U)$. So we prove the direct part of theorem.

Next we prove the converse part. Consider any code with code length $n$ and the number of synonymous sets satisfies $2^{nR}\leq 2^{n\left(H_s(\tilde{U})-\zeta\right)}$. So some sets are used to represent the semantically typical sequences $g^n\left(\mathcal{C}_s\right)\in \tilde{A}_{\epsilon}^{(n)}$ and the rest sets to represent the semantic non-typical sequences $g^n\left(\mathcal{C}_s\right)\notin \tilde{A}_{\epsilon}^{(n)}$. By Theorem \ref{theorem8}, for sufficiently large $n$, the probability of semantic sequences covered by the code is upper bounded by
\begin{equation}
\begin{aligned}
1-P_e^{(n)}&=2^{n\left(H_s(\tilde{U})-\zeta\right)}2^{-n\left(H_s(\tilde{U})-\epsilon\right)}+\text{Pr}(g^n\left(\mathcal{C}_s\right)\notin \tilde{A}_{\epsilon}^{(n)})\\
&< 2^{-n\left(\zeta-\epsilon\right)}+\epsilon.
\end{aligned}
\end{equation}
Therefore, we can write the error probability as
\begin{equation}
P_e^{(n)} > 1-2^{-n\left(\zeta-\epsilon\right)}-\epsilon.
\end{equation}
This inequality holds when $n\to\infty$ for $\epsilon>0$ and $\zeta>\epsilon$. So we have $P_e^{(n)}>1-2\epsilon$. Additionally, when $n\to\infty$ and $\epsilon\to 0$, $P_e^{(n)} \to 1$. So we complete the proof of the semantic source coding theorem.
\end{proof}

\begin{remark}
The semantic source coding theorem is an extended version of the counterpart in the classic information theory. The limitation of semantic compression rate $R$ is $H_s(\tilde{U})$ while the corresponding syntactic rate $R'$ is $H_s(\tilde{U})+R_s$. If the synonymous set is represented by only one codeword, i.e., $R_s=0$, then compression rate $R'$ tends to the semantic entropy $H_s(\tilde{U})$. On the other hand, if the synonymous mapping is ignored and all the synonymous typical sequences are distinct, thus $R_s=H(U)-H_s(\tilde{U})$ and the compression rate $R'$ achieves the information entropy $H(U)$. Since the distinction between codewords in the synonymous set is no longer regarded as an error, by using semantic coding, the source data can be further compressed and the efficiency is improved.
\end{remark}

\subsection{Semantic Source Coding Method}
Similar to the classic source coding, the variable length coding is desired for semantic source coding. Thus we also obtain the semantic version of Kraft inequality as following.
\begin{theorem}
(Semantic Kraft Inequality):
Given a discrete random variable $U\in \mathcal{U}=\left\{u_i\right\}_{i=1}^{N}$, the corresponding semantic variable $\tilde {U}\in \tilde{\mathcal{U}}=\left\{\tilde{u}_{i_s}\right\}_{i_s=1}^{\tilde{N}}$, and the synonymous mapping $f: \tilde{\mathcal{U}}\to\mathcal{U}$. For any prefix code over an alphabet of size $F$ exists if and only if the codeword length $l_1,l_2,\cdots, l_{\tilde{N}}$ satisfies
\begin{equation}
\sum_{i_s=1}^{\tilde{N}}F^{-l_{i_s}}\leq 1.
\end{equation}
\end{theorem}

The proof is similar to that of classic Kraft inequality and omitted. It should be noted that the semantic Kraft inequality has the same form as that of classic counterpart. However, since the semantic prefix code is performed over a synonymous set rather than a single syntactic symbol, the number of codewords is less than classic prefix code, that is, $\tilde{N}\leq N$.

Furthermore, we can obtain the average code length of optimal semantic source code as follows.
\begin{theorem}
Given the syntactic source distribution $p(u)$ and the synonymous mapping $f: \tilde{\mathcal{U}}\to\mathcal{U}$, let $l_1^*,l_2^*,\cdots, l_{\tilde{N}}^*$ denote the optimal code lengths with an $F$-ary alphabet, then the expected length $\bar{L}^*$ of the optimal semantic code satisfies
\begin{equation}
\frac{H_s(\tilde{U})}{\log F}\leq \bar{L}^* < \frac{H_s(\tilde{U})}{\log F}+1.
\end{equation}
\end{theorem}
\begin{proof}
Assign the code length as $l_{i_s}=\left\lceil-\log_F \sum_{i\in \mathcal{U}_{i_s}} p(u_i)\right\rceil$. Similar to the classic version, by semantic Kraft inequality, we have
\begin{equation}
\frac{H_s(\tilde{U})}{\log F} \leq \bar{L}^*\leq \sum_{i_s=1}^{\tilde{N}} p(\tilde{u}_{i_s}) l_{i_s}< \frac{H_s(\tilde{U})}{\log F}+1.
\end{equation}
So we prove the theorem.
\end{proof}

\begin{example}
(Semantic Huffman Coding):
Now we describe an example of semantic Huffman coding. For a syntactic source $U$ with four symbols $u_1,u_2,u_3,u_4$, the probability distribution is listed in Table \ref{Syntactic_HC}. The information entropy is $H(U)=1.75$ bits. By using Huffman coding, the codewords are shown in Table \ref{Syntactic_HC} and the average code length is $\bar{L}=1.75$ bits $=H(U)$.

\begin{table}[tp]
\centering
\caption{Probability distribution and Huffman codes of syntactic source $U$.} \label{Syntactic_HC}
\begin{tabular}{|c|c|c|c|c|}
  \hline Syntactic symbol      &       $u_1$          &          $u_2$         &         $u_3$         &        $u_4$         \\
  \hline Probability               &  $\frac{1}{2}$    &    $\frac{1}{4}$   &   $\frac{1}{8}$   &   $\frac{1}{8}$  \\
  \hline Syn. HC  &           0              &            10             &          110            &           111         \\
  \hline
\end{tabular}
\end{table}

If we give a synonymous mapping $f$, that is, $\tilde{u}_1\to \{u_1\}$, $\tilde{u}_2\to \{u_2\}$, $\tilde{u}_3\to \{u_3,u_4\}$, the probability distribution of semantic source $\tilde{U}$ is listed in Table \ref{Semantic_HC}. So the semantic entropy is calculated as $H_s(\tilde{U})=1.5$ sebits. Correspondingly, by the semantic Huffman codewords listed in Table \ref{Semantic_HC}, the average code length is $\bar{L}_s=1.5$ sebits $=H_s(\tilde{U})$. Distinctly, due to synonymous mapping, the average code length of semantic Huffman code is smaller than that of traditional Huffman code.

\begin{table}[tp]
\centering
\caption{Probability distribution and Huffman codes of semantic source $\tilde{U}$.} \label{Semantic_HC}
\begin{tabular}{|c|c|c|c|}
  \hline Semantic symbol      &       $\tilde{u}_1\to\{u_1\}$     &          $\tilde{u}_2\to \{u_2\}$         &         $\tilde{u}_3\to \{u_3,u_4\}$   \\
  \hline Probability               &              $\frac{1}{2}$            &                 $\frac{1}{4}$                 &                $\frac{1}{4}$                  \\
  \hline Sem. HC  &                         0                    &                           10                         &                         11                           \\
  \hline
\end{tabular}
\end{table}

Given a syntactic sequence $\mathbf{u}=(u_1u_1u_3u_4u_2u_3u_2)$, by Table \ref{Syntactic_HC}, the syntactic Huffman coding is $\mathbf{x}=(001101111011010)$. On the other hand, by Table \ref{Semantic_HC}, the semantic Huffman coding is $\mathbf{x}_s=(001111101110)$. Hence, the length of syntactic coding is $L(\mathbf{x})=15$ bits and the length of semantic coding is $L(\mathbf{x}_s)=12$ sebits so that the latter is smaller than the former, i.e., $L(\mathbf{x}_s)<L(\mathbf{x})$. Certainly, since the decoder can select arbitrary symbol from the set $\{u_3,u_4\}$ when decoding $\tilde{u}_3$, the result may be $\hat{\mathbf{u}}=(u_1u_1u_3u_3u_2u_4u_2)$. Although such decoding sequence is different from the original one $\mathbf{u}=(u_1u_1u_3u_4u_2u_3u_2)$ in the syntactic sense, the semantic information of the decoding results still keeps the same since $u_3$ and $u_4$ have the same meaning.
\end{example}

\begin{remark}
For the method of semantic source coding, we have two kinds of design thought. The first kind thought is to modify the traditional source coding, such as Huffman, arithmetic or universal coding. By using an elaborate synonymous mapping, these coding methods can be devised to further improve the compression efficiency. For the second thought, based on the deep learning method, we can construct a neural network model to perform semantic source coding. In this model, the synonymous mapping and semantic coding can be integrated and optimized to approach the theoretic limitation.
\end{remark}

\section{Semantic Channel Coding}
\label{section_VII}
In this section, we investigate the semantic channel coding. First, we introduce the jointly asymptotic equipartition property (JAEP) in the semantic sense and define the jointly synonymous typical set. Then we prove the semantic channel coding theorem by using JAEP and jointly synonymous typical set, which states that the semantic capacity, $\max_{f_{xy}}\max_{p(x)}I^s({\tilde{X};\tilde{Y}})$, the maximum up semantic mutual information, is the largest achievable rate of semantic communication. Finally, we consider the semantic channel decoding problem and propose the maximum likelihood group (MLG) decoding algorithm. A simple example of semantic Hamming code is analyzed based on the MLG decoding rule.

\subsection{Jointly Asymptotic Equipartition Property and Jointly Synonymous Typical Set}
Given the semantic channel model $\left\{\tilde{\mathcal{X}},\mathcal{X},\mathcal{Y},\tilde{\mathcal{Y}}, p(Y|X)\right\}$, $\mathcal{X}$ and $\mathcal{Y}$ are the input and output syntactical alphabet and $\tilde{\mathcal{X}}$ and $\tilde{\mathcal{Y}}$ are the corresponding input and output semantic alphabet. Furthermore, $p(Y|X)$ is the channel transition probability and let $f_{xy}:\tilde{\mathcal{X}}\times\tilde{\mathcal{Y}}\to\mathcal{X}\times\mathcal{Y}$ be the jointly synonymous mapping.

Based on this channel model, we extend the jointly asymptotic equipartition property (JAEP) to the semantic sense and use it to prove the channel coding theorem in semantic information theory.

\begin{definition}
Given the semantic channel model, the jointly synonymous mapping over the sequence pairs is defined as $f_{xy}^{n}: \tilde{\mathcal{X}}^{n}\times \tilde{\mathcal{Y}}^{n} \to \mathcal{X}^{n}\times \mathcal{Y}^{n}$. Equivalently, given a sequence pair $\left(\tilde {x}^n,\tilde {y}^n\right)$, we have $f_{xy}^n\left(\tilde {x}^n, \tilde {y}^n \right)=\prod_{k=1}^{n}{\mathcal{X}}_{\tilde{x}_{k}}\times {\mathcal{Y}}_{\tilde{y}_{k}}$.
\end{definition}

To facilitate the reader's understanding, we first cite the definition of syntactically joint typical set $A_{\epsilon}^{(n)}$ in \cite{Book_Cover}.
\begin{definition}\label{Syn_JTS}
Given the jointly typical sequences $\left\{x^n,y^n\right\}$ with respect to the distribution $p(x^n,y^n)$, the syntactically jointly typical set $A_{\epsilon}^{(n)}$ is defined as the set of $n$-sequence pairs with empirical entropies $\epsilon$-close to the true entropies, that is,
\begin{equation}
\begin{aligned}
A_{\epsilon}^{(n)}=& \bigg\{\left(x^n,y^n\right) \in {\mathcal{X}}^n\times \mathcal{Y}^n:  \\
                                & \left|-\frac{1}{n}\log p\left(x^n\right)-H(X)\right|<\epsilon,\\
                                & \left|-\frac{1}{n}\log p\left(y^n\right)-H(Y)\right|<\epsilon,\\
                                & \left. \left|-\frac{1}{n}\log p\left(x^n,y^n\right)-H(X,Y)\right|<\epsilon \right\}
\end{aligned}
\end{equation}
where $p\left(x^n,y^n\right)=\prod_{k=1}^{n} p(x_k,y_k)$.
\end{definition}

Next, we introduce the semantically jointly typical set $\tilde{A}_{\epsilon}^{(n)}$.
\begin{definition}\label{Sem_JTS}
Given the jointly typical sequences $\left\{\tilde{x}^n,\tilde{y}^n\right\}$ with respect to the distribution $p(x^n,y^n)$, the semantically jointly typical set ${\tilde{A}}_{\epsilon}^{(n)}$ is defined as the set of $n$-sequence pairs with empirical entropies $\epsilon$-close to the true entropies, that is,
\begin{equation}
\begin{aligned}
{\tilde{A}}_{\epsilon}^{(n)}=&\bigg\{\left(\tilde{x}^n,\tilde{y}^n\right)\in \tilde{\mathcal{X}}^n\times \tilde{\mathcal{Y}}^n:\\
                                              &\left|-\frac{1}{n}\log p\left(\tilde{x}^n \right)-H_s(\tilde{X})\right|<\epsilon, \\
                                              &\left|-\frac{1}{n}\log p\left(\tilde{y}^n \right)-H_s(\tilde{Y})\right|<\epsilon, \\
                                              & \left.\left|-\frac{1}{n}\log p\left(\tilde{x}^n,\tilde{y}^n\right)-H_s(\tilde{X},\tilde{Y})\right|<\epsilon\right\}
\end{aligned}
\end{equation}
where
\begin{equation}
p\left(\tilde{x}^n,\tilde{y}^n\right)=\prod_{k=1}^{n} p(\tilde{x}_k,\tilde{y}_k)=\prod_{k=1}^{n} \sum_{(x_k,y_k)\in {\mathcal{X}_{\tilde{x}_k}\times \mathcal{Y}_{\tilde{y}_k}}}p(x_k,y_k).
\end{equation}
\end{definition}

The semantic and syntactic sequence mapping is depicted in Fig. \ref{Sem_Syn_Seq_Mapping}. As described in Section \ref{section_VI}, under the synonymous mapping $f_x^n$ and $f_y^n$, semantic source set $\tilde {\mathcal{X}}$ and semantic destination set $\tilde {\mathcal{Y}}$ are mapped into synonymous typical sets. Similarly, under the joint mapping $f_{xy}^n$, the typical sequences in these sets can further compose the jointly synonymous typical set $B_{\epsilon}^{(n)}\left(\tilde{x}^n,\tilde{y}^n\right)$.
\begin{figure*}[htbp]
\setlength{\abovecaptionskip}{0.cm}
\setlength{\belowcaptionskip}{-0.cm}
  \centering{\includegraphics[scale=0.9]{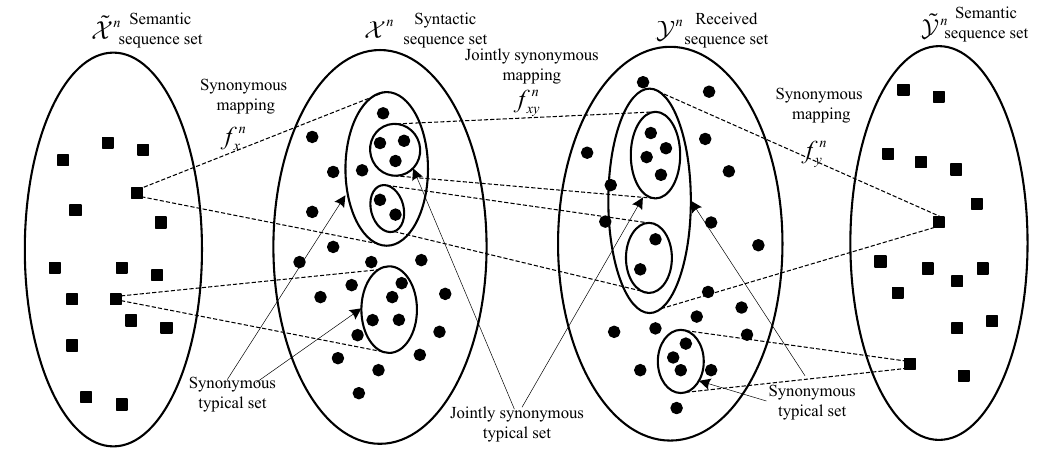}}
  \caption{Semantic and syntactic sequence mapping over the channel.}\label{Sem_Syn_Seq_Mapping}
\end{figure*}

\begin{definition}
Given a specifically typical sequence pair $\left\{\tilde{x}^n,\tilde{y}^n\right\}$, the jointly synonymous typical set $B_{\epsilon}^{(n)}\left(\tilde{x}^n,\tilde{y}^n\right)$ with the syntactically jointly typical sequences $\left\{x^n,y^n\right\}$ is defined as the set of $n$-sequence pairs with the difference of empirical entropies $\epsilon$-close to the difference of true entropies, that is,
\begin{equation}
\begin{aligned}
&B_{\epsilon}^{(n)}\left(\tilde{x}^n,\tilde{y}^n\right)\\
=&\bigg\{\left(x^n,y^n\right) \in {\mathcal{X}}^n\times \mathcal{Y}^n:  \\
                                                              &\left|-\frac{1}{n}\log p\left(\tilde{x}^n \right)-H_s(\tilde{X})\right|<\epsilon,\\
                                                              &\left|-\frac{1}{n}\log p\left(\tilde{y}^n \right)-H_s(\tilde{Y})\right|<\epsilon, \\
                                                              &\left|-\frac{1}{n}\log p\left(\tilde{x}^n,\tilde{y}^n\right)-H_s(\tilde{X},\tilde{Y})\right|<\epsilon,\\
                                                              &\left|-\frac{1}{n}\log p\left(x^n \right)-H(X)\right|<\epsilon,\\
                                                              &\left|-\frac{1}{n}\log p\left(y^n \right)-H(Y)\right|<\epsilon, \\
                                                              &\left|-\frac{1}{n}\log p\left(x^n,y^n\right)-H(X,Y)\right|<\epsilon,\\
&\left|-\frac{1}{n}\log p((\tilde{x}^n,\tilde{y}^n)\to (x^n,y^n))\right.\\
&\left.\left.-\left(H(X,Y)-H_s(\tilde{X},\tilde{Y})\right)\right|<\epsilon\right\},
\end{aligned}
\end{equation}
where the probability $p((\tilde{x}^n,\tilde{y}^n)\to (x^n,y^n))$ is defined as
\begin{equation}
\begin{aligned}
&p((\tilde{x}^n,\tilde{y}^n)\to (x^n,y^n))\\
&=\left\{
\begin{aligned}
&\frac{p\left(x^n,y^n\right)}{p\left(\tilde{x}^n,\tilde{y}^n\right)},&\text{ if } f_{xy}^n(\tilde{x}^n,\tilde{y}^n)=(x^n,y^n),\\
&0,&\text{otherwise}.
\end{aligned}\right.
\end{aligned}
\end{equation}
\end{definition}

Under the jointly synonymous mapping $f_{xy}^n:\tilde{\mathcal{X}}^n \times \tilde{\mathcal{Y}}^n\to {\mathcal{X}}^n\times {\mathcal{Y}}^n$, the syntactically jointly typical set $A_{\epsilon}^{(n)}$ can be partitioned into jointly synonymous typical sets $B_{\epsilon}^{(n)}\left(\tilde{x}^n,\tilde{y}^n\right)$, that is,
\begin{equation}
A_{\epsilon}^{(n)}=\bigcup_{(\tilde{x}^n,\tilde{y}^n)\in \tilde{A}_{\epsilon}^{(n)}} B_{\epsilon}^{(n)}\left(\tilde{x}^n,\tilde{y}^n\right).
\end{equation}
For $\forall (\tilde{u}^n, \tilde{v}^n),(\tilde{x}^n, \tilde{y}^n)\in \tilde{A}_{\epsilon}^{(n)}$, $(\tilde{u}^n,\tilde{v}^n)\neq (\tilde{x}^n,\tilde{y}^n)$, we have $B_{\epsilon}^{(n)}\left(\tilde{x}^n,\tilde{y}^n\right) \bigcap B_{\epsilon}^{(n)}\left(\tilde{u}^n,\tilde{v}^n\right)=\varnothing$.

Similarly, $B_{\epsilon}^{(n)}\left(\tilde{x}^n,\tilde{y}^n\right)$ also is an equivalence class of the jointly synonymous typical sequences. Thus we can construct a quotient set $A_{\epsilon}^{(n)}/f_{xy}^n=\left\{B_{\epsilon}^{(n)}(\tilde{x}^n,\tilde{y}^n)\right\}$ from the syntactically jointly typical set $A_{\epsilon}^{(n)}$, and establish an one-to-one mapping between $\tilde{A}_{\epsilon}^{(n)}$ and $A_{\epsilon}^{(n)}/f_{xy}^n$.

We now discuss the properties of syntactically and semantically jointly typical sets. We first rewrite the properties of syntactically jointly typical set in \cite{Book_Cover} as follows.

\begin{theorem}\label{Syn_JAEP_theorem}
(Syntactically joint AEP \cite[Theorem 7.6.1]{Book_Cover}):
Let $\left(X^n,Y^n\right)$ be a sequence pair with length $n$ obeying the i.i.d. distribution $p(x^n,y^n)$, then
\begin{enumerate}[(1)]
    \item $\text{Pr}((X^n,Y^n)\in A_{\epsilon}^{(n)})>1-\epsilon$ for $n$ sufficiently large. Or equivalently, if $\left(x^n,y^n\right)\in A_{\epsilon}^{(n)}$, then
         \begin{equation}\label{Syn_JAEPinequality}
        2^{-n\left(H(X,Y)+\epsilon\right)}\leq p\left(x^n,y^n\right) \leq 2^{-n\left(H(X,Y)-\epsilon\right)}.
        \end{equation}
    \item $\left(1-\epsilon\right) 2^{n\left(H(X,Y)-\epsilon\right)}\leq \left|A_{\epsilon}^{(n)}\right| \leq 2^{n\left(H(X,Y)+\epsilon\right)}$ for $n$ sufficiently large.
    \item If $\dot{X}^n$ and $\dot{Y}^n$ are independent sequences with the same marginals as $x^n$ and $y^n$, i.e., $(\dot{X}^n,\dot{Y}^n)\sim p(x^n)p(y^n)$, for $n$ sufficiently large, we have
        \begin{equation}
        \begin{aligned}
        (1-\epsilon)2^{-n(I(X;Y)+3\epsilon)}& \leq \text{Pr}\left((\dot{X}^n,\dot{Y}^n)\in A_{\epsilon}^{(n)}\right) \\
                                                               &\leq 2^{-n(I(X;Y)-3\epsilon)}.
        \end{aligned}
        \end{equation}
  \end{enumerate}
\end{theorem}

Then as the consequence of the semantic JAEP, we present the properties of semantically jointly typical set as following.
\begin{theorem}\label{Sem_JAEP_theorem}
(Semantically joint AEP): Let $(\tilde{X}^n,\tilde{Y}^n)$ be a semantic sequence pair with length $n$ drawn i.i.d. according to $p(\tilde{x}^n,\tilde{y}^n)$, by using jointly synonymous mapping $f_{xy}^n$, the associated syntactic sequence pair is $(X^n, Y^n)\sim p(x^n,y^n)$, then
\begin{enumerate}[(1)]
    \item $\text{Pr}((\tilde{X}^n,\tilde{Y}^n)\in \tilde{A}_{\epsilon}^{(n)})>1-\epsilon$ for $n$ sufficiently large. Or equivalently if $\left(\tilde{x}^n,\tilde{y}^n\right)\in \tilde{A}_{\epsilon}^{(n)}$, then
         \begin{equation}\label{Sem_JAEPinequality}
        2^{-n\left(H_s(\tilde{X},\tilde{Y})+\epsilon\right)}\leq p\left(\tilde{x}^n,\tilde{y}^n\right) \leq 2^{-n\left(H_s(\tilde{X},\tilde{Y})-\epsilon\right)}.
        \end{equation}
    \item $\left(1-\epsilon\right) 2^{n\left(H_s(\tilde{X},\tilde{Y})-\epsilon\right)}\leq \left|\tilde{A}_{\epsilon}^{(n)}\right| \leq 2^{n\left(H_s(\tilde{X},\tilde{Y})+\epsilon\right)}$ for $n$ sufficiently large.
    \item If $\dot{X}^n$ and $\dot{Y}^n$ are independent sequences with the same marginals as $X^n$ and $Y^n$, i.e., $(\dot{X}^n,\dot{Y}^n)\sim p(x^n)p(y^n)$,  and the corresponding semantic sequence is $(\dot{\tilde{X}}^n,\dot{\tilde{Y}}^n)$, for $n$ sufficiently large, we have
        \begin{equation}\label{jointly_typical_equality}
        \begin{aligned}
        (1-\epsilon)2^{-n(I^s(\tilde{X};\tilde{Y})+3\epsilon)}& \leq \text{Pr}((\dot{\tilde{X}}^n,\dot{\tilde{Y}}^n)\in \tilde{A}_{\epsilon}^{(n)}) \\
                                                                                        & \leq 2^{-n(I^s(\tilde{X};\tilde{Y})-3\epsilon)}.
        \end{aligned}
        \end{equation}
\end{enumerate}
\end{theorem}
\begin{proof}
Similar to Theorem \ref{Syn_JAEP_theorem}, the proof of property (1) is directly from the weak law of large numbers and Definition \ref{Sem_JTS}.

According to property (1), Eq. (\ref{Sem_JAEPinequality}) can be rewritten as
\begin{equation}
\left\{ \begin{aligned}
&\left|\tilde{A}_{\epsilon}^{(n)}\right| 2^{-n\left(H_s(\tilde{X},\tilde{Y})+\epsilon\right)}\leq \sum_{(\tilde{x}^n,\tilde{y}^n) \in \tilde{A}_{\epsilon}^{(n)}} p\left(\tilde{x}^n,\tilde{y}^n\right)\leq 1\\
&1-\epsilon \leq \sum_{(\tilde{x}^n,\tilde{y}^n)\in \tilde{A}_{\epsilon}^{(n)}} p\left(\tilde{x}^n,\tilde{y}^n\right) \leq \left|\tilde{A}_{\epsilon}^{(n)}\right| 2^{-n\left(H_s(\tilde{X},\tilde{Y})-\epsilon\right)}
\end{aligned}\right.
\end{equation}

So the cardinality of semantically jointly typical set can be written as
\begin{equation}\label{equation_SJTS_down}
\left|\tilde{A}_{\epsilon}^{(n)}\right| \geq \left(1-\epsilon\right) 2^{n\left(H_s(\tilde{X},\tilde{Y})-\epsilon\right)}
\end{equation}
and
\begin{equation}\label{equation_SJTS_up}
\left|\tilde{A}_{\epsilon}^{(n)}\right| \leq 2^{n\left(H_s(\tilde{X},\tilde{Y})+\epsilon\right)}
\end{equation}
respectively. We complete the proof of the second property.

Now assume $\dot{X}^n$ and $\dot{Y}^n$ are independent but have the same marginals as $X^n$ and $Y^n$, and after de-mapping, $(\dot{X}^n,\dot{Y}^n)$ is mapped into a semantic sequence pair $(\dot{\tilde{X}}^n,\dot{\tilde{Y}}^n)$, so we can establish an one-to-one mapping $(x^n,y^n) \leftrightarrow (\dot{x}^n,\dot{y}^n) \leftrightarrow (\dot{\tilde{x}}^n,\dot{\tilde{y}}^n)$, then we have
\begin{equation}
\begin{aligned}
&\text{Pr}\left((\dot{\tilde{X}}^n,\dot{\tilde{Y}}^n)\in \tilde{A}_{\epsilon}^{(n)}\right) \\
&=\text{Pr}\left((\dot{X}^n,\dot{Y}^n)\in B_{\epsilon}^{(n)},\dot{X}^n\in A_{\epsilon}^{(n)}(X^n),\dot{Y}^n\in A_{\epsilon}^{(n)}(Y^n)\right)\\
&= \sum_{ (x^n,y^n) \leftrightarrow (\dot{\tilde{x}}^n,\dot{\tilde{y}}^n)\in \tilde{A}_{\epsilon}^{(n)}} p(x^n)p(y^n)\\
&\leq 2^{n\left(H_s(\tilde{X},\tilde{Y})+\epsilon\right)}2^{-n\left(H(X)-\epsilon\right)}2^{-n\left(H(Y)-\epsilon\right)}\\
&= 2^{-n\left(I^s(\tilde{X},\tilde{Y})-3\epsilon\right)}.
\end{aligned}
\end{equation}

Using a similar thought, we can also derive that
\begin{equation}
\begin{aligned}
&\text{Pr}\left((\dot{\tilde{X}}^n,\dot{\tilde{Y}}^n)\in \tilde{A}_{\epsilon}^{(n)}\right)\\
&= \sum_{ \tilde{A}_{\epsilon}^{(n)}} p(x^n)p(y^n)\\
&\geq  (1-\epsilon)2^{n(H_s(\tilde{X};\tilde{Y})-\epsilon)}2^{-n\left(H(X)+\epsilon\right)}2^{-n\left(H(Y)+\epsilon\right)}\\
&= (1-\epsilon) 2^{-n\left(I^s(\tilde{X},\tilde{Y})+3\epsilon\right)}.
\end{aligned}
\end{equation}
Thus we complete the proof of the theorem.
\end{proof}

Furthermore, we give the properties of jointly synonymous typical set as below.
\begin{theorem}\label{SJTS_theorem}
\text{ }
\begin{enumerate}[(1)]
    \item Given a semantic sequence pair $\left(\tilde{x}^n,\tilde{y}^n\right)\in \tilde{A}_{\epsilon}^{(n)}$, if $\left(x^n, y^n\right)\in B_{\epsilon}^{(n)}\left(\tilde{x}^n,\tilde{y}^n\right)$, then
        \begin{equation}
        \begin{aligned}
        2^{-n\left(H(X,Y)-H_s(\tilde{X},\tilde{Y})+\epsilon\right)}&\leq \frac{p\left(x^n,y^n\right)}{p\left(\tilde{x}^n,\tilde{y}^n\right)}\\
         &\leq 2^{-n\left(H(X,Y)-H_s(\tilde{X},\tilde{Y})-\epsilon\right)}.
        \end{aligned}
        \end{equation}
    \item $ 2^{n\left(H(X,Y)-H_s(\tilde{X},\tilde{Y})-\epsilon\right)}\leq\left|B_{\epsilon}^{(n)}\left(\tilde{x}^n,\tilde{y}^n\right)\right| \leq 2^{n\left(H(X,Y)-H_s(\tilde{X},\tilde{Y})+\epsilon\right)}$ for sufficiently large $n$.
\end{enumerate}
\end{theorem}
\begin{proof}
The proof of property (1) is directly from the definition of $B_{\epsilon}^{(n)}\left(\tilde{x}^n,\tilde{y}^n\right)$.

To prove the left inequality of property (2), we write
\begin{equation}
\begin{aligned}
&p\left(\tilde{x}^n,\tilde{y}^n\right)\\
&=\sum_{(x^n,y^n)\in B_{\epsilon}^{(n)}\left(\tilde{x}^n,\tilde{y}^n\right)} p\left(x^n,y^n\right)\\
&\leq p\left(\tilde{x}^n,\tilde{y}^n\right) \sum_{(x^n,y^n) \in B_{\epsilon}^{(n)}\left(\tilde{x}^n,\tilde{y}^n\right)}2^{-n\left(H(X,Y)-H_s(\tilde{X},\tilde{Y})-\epsilon\right)}\\
& =p\left(\tilde{x}^n,\tilde{y}^n\right) \left|B_{\epsilon}^{(n)}\left(\tilde{x}^n,\tilde{y}^n\right)\right| 2^{-n\left(H(X,Y)-H_s(\tilde{X},\tilde{Y})-\epsilon\right)}.
\end{aligned}
\end{equation}
So it follows that
\begin{equation}
\left|B_{\epsilon}^{(n)}\left(\tilde{x}^n,\tilde{y}^n\right)\right|\geq 2^{n\left(H(X,Y)-H_s(\tilde{X},\tilde{Y})-\epsilon\right)}
\end{equation}

On the other hand, we can write
\begin{equation}
\begin{aligned}
&p\left(\tilde{x}^n,\tilde{y}^n\right)\\
&=\sum_{(x^n,y^n)\in B_{\epsilon}^{(n)}\left(\tilde{x}^n,\tilde{y}^n\right)} p\left(x^n,y^n\right) \\
&\geq p\left(\tilde{x}^n,\tilde{y}^n\right) \sum_{(x^n,y^n) \in B_{\epsilon}^{(n)}\left(\tilde{x}^n,\tilde{y}^n \right)}2^{-n\left(H(X,Y)-H_s(\tilde{X},\tilde{Y})+\epsilon\right)}\\
& =p\left(\tilde{x}^n,\tilde{y}^n\right) \left|B_{\epsilon}^{(n)}\left(\tilde{x}^n,\tilde{y}^n\right)\right| 2^{-n\left(H(X,Y)-H_s(\tilde{X},\tilde{Y})+\epsilon\right)}.
\end{aligned}
\end{equation}
Similarly, we have
\begin{equation}
\left|B_{\epsilon}^{(n)}\left(\tilde{x}^n,\tilde{y}^n\right)\right| \leq 2^{n\left(H(X,Y)-H_s(\tilde{X},\tilde{Y})+\epsilon\right)}
\end{equation}
and complete the proof.
\end{proof}

\begin{figure*}[htbp]
\setlength{\abovecaptionskip}{0.cm}
\setlength{\belowcaptionskip}{-0.cm}
  \centering{\includegraphics[scale=0.7]{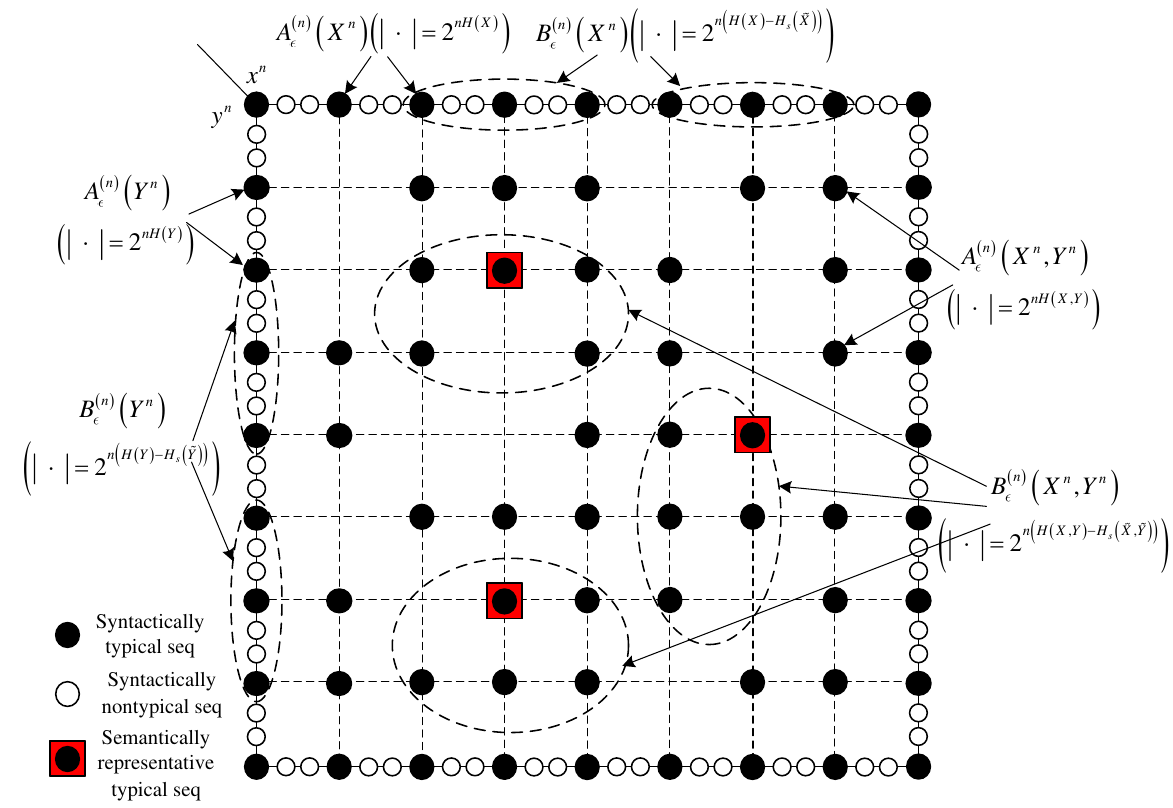}}
  \caption{The relationship of typical sets.}\label{Typical_set_relation}
\end{figure*}

The relationships of typical sets, such as typical set, jointly typical set and synonymous typical set, are shown in Fig. \ref{Typical_set_relation}. By using synonymous mapping $f_x^n \text{ or } f_y^n$, syntactically typical set $A_{\epsilon}^{(n)}(X^n) \text{ or } A_{\epsilon}^{(n)}(Y^n)$ can be partitioned into many synonymous typical sets $B_{\epsilon}^{(n)}(X^n)\text{ or } B_{\epsilon}^{(n)}(Y^n)$. \footnote{Note that $B_{\epsilon}^{(n)}(X^n)$ denotes the set consisting the syntactically typical sequences $X^n$ whereas $B_{\epsilon}^{(n)}(\tilde{x}^n)$ denotes the set induced by the semantic typical sequence $\tilde{x}^n$. Correspondingly, $B_{\epsilon}^{(n)}(Y^n)$ and $B_{\epsilon}^{(n)}(\tilde{y}^n)$, $B_{\epsilon}^{(n)}(X^n,Y^n)$ and $B_{\epsilon}^{(n)}(\tilde{x}^n,\tilde{y}^n)$ have similar distinction.} The syntactically jointly typical set $A_{\epsilon}^{(n)}(X^n,Y^n)$ consists of some syntactically typical sequences simultaneously belonging to the typical set $A_{\epsilon}^{(n)}(X^n)$ and $A_{\epsilon}^{(n)}(Y^n)$. Note that not all pairs of syntactically typical $X^n$ and $Y^n$ are jointly typical.

The probability of jointly typical set $A_{\epsilon}^{(n)}(X^n,Y^n)$ with sufficiently large $n$ is close to 1. In addition, since all the jointly typical sequences are almost equiprobability, the number of syntactically jointly typical sequences is about $\left|A_{\epsilon}^{(n)}(X^n,Y^n)\right|\thickapprox2^{n(H(X,Y)\pm\epsilon)}$.
Similar conclusion holds for the jointly typical set $\tilde{A}_{\epsilon}^{(n)}(\tilde{X}^n,\tilde{Y}^n)$ and the corresponding number of jointly typical sequences is about $\left|\tilde{A}_{\epsilon}^{(n)}(\tilde{X}^n,\tilde{Y}^n)\right|\thickapprox2^{n(H_s(\tilde{X},\tilde{Y})\pm\epsilon)}$.

Furthermore, under the joint mapping $f_{xy}^n$, as circled by dashed lines, some jointly typical sequences constitute the jointly synonymous typical sets $B_{\epsilon}^{(n)}(X^n,Y^n)$. Since the number of jointly synonymous typical sequences is about $2^{n(H(X,Y)-H_s(\tilde{X},\tilde{Y})\pm\epsilon)}$, all the synonymous typical sets have almost the same number of typical sequences. In each set, the black circle marked by a red box denotes the representative typical sequence. If we randomly choose a typical sequence pair, the probability that this pair falls in a jointly synonymous typical set (equivalently represents a semantically jointly typical sequence) is about $2^{-nI^s(\tilde{X};\tilde{Y})}$. This means that there are about $2^{nI^s(\tilde{X};\tilde{Y})}$ distinguishable sequences $\tilde{X}^n$ in the semantic sense.

\subsection{Semantic Channel Coding Theorem}
We now discuss the problem of semantic channel coding. As shown in Fig. \ref{Semantic_Channel_coding}, with the help of synonymous mapping $f^n$, a semantic index $i_s$ is mapped and encoded into the channel codeword $X^n$. Here, the semantic message $\tilde{X}^n(i_s)$ is a broad concept, which can be a real semantic sequence or a syntactic sequence with some semantic constraints. After going through the channel, we obtain the received sequence $Y^n$. Then the decoder outputs the decoded codeword $\hat{X}^n$. After de-synonymous mapping $g^n$, we obtain an estimation of semantic index $\hat{i}_s$.

\begin{figure*}[htbp]
\setlength{\abovecaptionskip}{0.cm}
\setlength{\belowcaptionskip}{-0.cm}
  \centering{\includegraphics[scale=0.9]{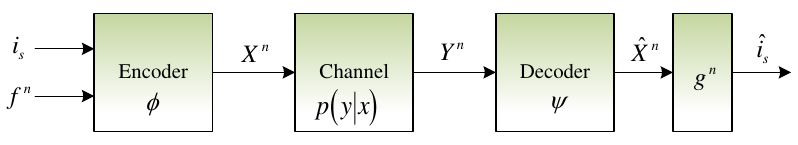}}
  \caption{Block diagram of semantic channel coding.}\label{Semantic_Channel_coding}
\end{figure*}

For a discrete memoryless channel, let $p(Y|X)$ be the channel transition probabilities and $\mathcal{X}$ and $\mathcal{Y}$ denote the input and output syntactical alphabet of channel respectively. Then the channel transition probabilities for the $n$-th extension of the channel can be written as
\begin{equation}
p(y^n\left|x^n\right.)=\prod_{k=1}^{n}p(y_k\left|x_k\right.).
\end{equation}

\begin{definition}
An $\left(M,n\right)$ code for the semantic channel $\left\{\tilde{\mathcal{X}},\mathcal{X},\mathcal{Y},\tilde{\mathcal{Y}}, p(Y|X)\right\}$ consists of the following parts:
\begin{enumerate}[(1)]
   \item A semantic index set $\mathcal{I}_s=\left\{1,\cdots,i_s,\cdots, M_s\right\}$ and a syntactic index set $\mathcal{I}=\left\{1,\cdots,i,\cdots, M\right\}$.
   \item An encoding function $\phi: \tilde{\mathcal{X}}^n\to  \mathcal{X}^n$ generates the set of codewords, namely, codebook, $\mathcal{C}=\left\{X^n(1),X^n(2),\cdots,X^n(M)\right\}$. Due to the synonymous mapping $f^n$, this codebook can be partitioned into synonymous codeword subsets $\mathcal{C}_s$.
   \item A decoding function $\psi: \mathcal{Y}^n \to \mathcal{X}^n$ outputs the decision syntactic codeword $\hat{X}^n$.
   \item After de-mapping, $g^n(\hat{X}^n)=\hat{i}_s$, the estimated semantic index is obtained. Note that both $\psi$ and $g^n$ are deterministic.
\end{enumerate}
\end{definition}

According to the synonymous mapping $f^n$, $\mathcal{C}_s$ is an equivalence class consisting of the synonymous codewords. So we can construct a quotient set $\mathcal{C}/f^n=\left\{\mathcal{C}_s\right\}$ with $\left|\mathcal{C}/f^n\right|=M_s$. Let $R=\frac{1}{n}\log_2 M_s$ denote the code rate of semantic label of channel coding. We configure each synonymous set with the same number of codewords, that is, $\left|\mathcal{C}_s\right|=2^{nR_s}=\frac{M}{M_s}$, where $R_s$ is named as the rate of synonymous set. Furthermore, let $R_t=R+R_s=\frac{1}{n}\log_2 M$ being the total code rate.

\begin{definition}
Assume a semantic index $i_s$ is mapped into a syntactic codeword $X^n(i)\in \mathcal{C}_s(i_s)$, the conditional decoding error probability given the index $i_s$ is defined as
\begin{equation}
\begin{aligned}
\lambda_{i_s}&=\text{Pr}\left(g^n\left(\psi \left(Y^n\right)\right)\neq i_s \left|X^n=X^n(i)\leftrightarrow\tilde{X}^n=\tilde{X}^n(i_s)\right.\right)\\
      &=\text{Pr}\left(\hat{X}^n(i) \notin \mathcal{C}_s(i_s) \left|X^n=X^n(i)\leftrightarrow\tilde{X}^n=\tilde{X}^n(i_s)\right.\right)\\
      &=\sum_{y^n}p\left(y^n(i)\left|x^n(i)\right.\right) I\left(\psi(y^n)\notin \mathcal{C}_s(i_s)\right)
\end{aligned}
\end{equation}
where $I(\cdot)$ is the indicator function.
Assume the index $i_s$ is chosen uniformly on the set $\mathcal{I}_s$, the average error probability $P_e^{(n)}$ for an $(M,n)$ code is defined as
\begin{equation}
P_e^{(n)}=\frac{1}{M_s} \sum_{i_s=1}^{M_s} \lambda_{i_s}.
\end{equation}

\end{definition}

We now give the formal description of semantic channel coding theorem.
\begin{theorem}\label{Sem_CCT}
(Semantic Channel Coding Theorem)

Given the semantic channel $\left\{\tilde{\mathcal{X}},\mathcal{X},\mathcal{Y},\tilde{\mathcal{Y}}, p(Y|X)\right\}$, for each code rate $R_t<C_s$, there exists a sequence of $\left(2^{n(R+R_s)},n\right)$ codes consisting of synonymous codeword set with the rate $0\leq R_s \leq H(X,Y)-H_s(\tilde{X},\tilde{Y})$, when code length $n$ tends to sufficiently large, the error probability tends to zero, i.e. $P_e^{(n)}\to 0$.

On the contrary, if $R_t>C_s$, then for any $\left(2^{n{(R+R_s)}},n\right)$ code, the error probability tends to $1$ with sufficiently large $n$.
\end{theorem}
\begin{proof}
We first prove the achievability part of the theorem and the converse will be left in the next part.

Given the source distribution $p(x)$ and the synonymous mapping $f_x$, a $\left(2^{n{(R+R_s)}},n\right)$ code can be generated randomly according to the distribution $p(x)$. Note that $2^{n{(R+R_s)}}$ codewords can be independently generated based on the distribution
\begin{equation}
p(x^n)=\prod_{k=1}^{n} p(x_k).
\end{equation}
Furthermore, these codewords can be uniformly divided into $2^{nR}$ groups according to the synonymous mapping
\begin{equation}
f_x^n(\tilde{x}^n)=\prod_{k=1}^{n} \mathcal{X}_{\tilde{x}_k}.
\end{equation}

Thus all the $2^{n{(R+R_s)}}$ codewords can be listed as a matrix
\begin{equation}
\mathcal{C}=\left[ \begin{matrix}
   {{x}_{1}}(1) & {{x}_{2}}(1) & \cdots  & {{x}_{n}}(1)  \\
   \vdots  & \vdots  & \vdots  & \vdots   \\
   {{x}_{1}}({{2}^{n{{R}_{s}}}}) & {{x}_{2}}({{2}^{n{{R}_{s}}}}) & \cdots  & {{x}_{n}}({{2}^{n{{R}_{s}}}})  \\
   \vdots  & \vdots  & \ddots  & \vdots   \\
   {{x}_{1}}({{2}^{n(R+R_s)}}) & {{x}_{2}}({{2}^{n(R+R_s)}}) & \cdots  & {{x}_{n}}({{2}^{n(R+R_s)}})  \\
\end{matrix} \right]
  \begin{matrix}
   \begin{aligned}
   & \left.
     \begin{aligned}
   &  \\
   &  \\
   &  \\
  \end{aligned} \right\}{{\mathcal{C}}_{s}}(1) \\
     & \vdots  \\
   \end{aligned}  \\
   {{\mathcal{C}}_{s}}({{2}^{n{{R}}}})  \\
\end{matrix}
\end{equation}
The probability of generating the synonymous codeword set $\mathcal{C}_s(1)$ is
\begin{equation}
\text{Pr}\left(\mathcal{C}_s(1)\right)=\prod_{i=1}^{2^{nR_s}}\prod_{k=1}^{n} p(x_k(i))
\end{equation}
All the codeword sets $\mathcal{C}_s(i_s)$ have the same generating probability. Furthermore, the probability of generating a particular code $\mathcal{C}$ is
\begin{equation}
\text{Pr}\left(\mathcal{C}\right)=\prod_{i=1}^{2^{n(R+R_s)}}\prod_{k=1}^{n} p(x_k(i)).
\end{equation}

Similar to the idea in \cite{Book_Cover}, we also use jointly typical decoding for the semantic channel code. If a codeword $\hat{X}^n(i)$ is decided, it must satisfy the following conditions.
\begin{enumerate}
   \item $(\hat{X}^n(i),Y^n)$ is syntactically jointly typical, due to $g^n(\hat{X}^n(i))=\hat{i}_s$, $(\tilde{X}^n(\hat{i}_s),\tilde{Y}^n)$ is also semantically jointly typical. Equivalently, $(\hat{X}^n(i),Y^n)$ is jointly synonymous typical. Hence, the decision codeword $\hat{X}^n(i)$ may be not equal to the transmit codeword $X^n(i)$ but both $\hat{X}^n(i)$ and $X^n(i)$ belong to $\mathcal{C}_s(\hat{i}_s)$.
   \item There is no other index $m$, satisfying $(\hat{X}^n(m),Y^n) \in A_{\epsilon}^{(n)}$ or $g^n(\hat{X}^n(m))=\hat{i}_s$.
\end{enumerate}

We now calculate the average error probability of jointly typical decoding. Generally, this error probability should be averaged over all the codebooks and all the codewords. However, based on the symmetry of the code construction, due to averaging over all codes, the error probability is not dependent on the specific the semantic index $i_s$ and the syntactic index $i$. We relabel the codewords in a synonymous set as $X^n(i_s,j)\in \mathcal{C}_s(i_s), j\in \left\{1,2,\cdots,2^{nR_s}\right\}$.

Without loss of generality, we can assume $i_s=1$ and $i=1$, that is, the codeword $X^n(1,1)$ is sent. Therefore, the average error probability can be written as
\begin{equation}
\begin{aligned}
\text{Pr}(\mathcal{E})&=\frac{1}{2^{nR}}\sum_{i_s=1}^{2^{nR}}\sum_{\mathcal{C}}P(\mathcal{C})\lambda_{i_s}(\mathcal{C})\\
&=\sum_{\mathcal{C}}P(\mathcal{C})\lambda_{1}(\mathcal{C})\\
&=\text{Pr}\left(\mathcal{E}\left|X^{n}(1,1)\right.\right).
\end{aligned}
\end{equation}

Given the received sequence $Y^n$ when sending the first codeword $X^n(1,1)$, we define the following events:
\begin{equation}
\begin{aligned}
E_{i_s}=\left\{\left(X^n(i_s,j),Y^n\right)\in B_{\epsilon}^{(n)}, X^n(i_s,j)\in \mathcal{C}_s(i_s)\right\}&,\\
i_s\in \mathcal{I}_s, j\in \left\{1,2,\cdots,2^{nR_s}\right\}&.
\end{aligned}
\end{equation}
Here, the event $E_{i_s}$ means that the codewords in the $i_s$-th synonymous set $ \mathcal{C}_s(i_s)$ and $Y^n$ are jointly synonymous typical.

When sending the first codeword $X^n(1,1)$ and receiving the received sequence $Y^n$, by using jointly typical decoding, two kinds of error will occur. The first error event is $E_1^c$ which means that all the codewords in $\mathcal{C}_s(1)$ and $Y^n$ are not jointly typical. By the syntactically joint AEP, we have
\begin{equation}\label{Prb_Errorevent1}
P(E_1^c)\leq \epsilon, \text{for sufficiently large }n.
\end{equation}

On the other hand, the second error event is $E_{i_s},i_s\in \left\{2,\cdots,2^{nR}\right\}$ which means that a codeword in a wrong synonymous set is jointly typical with the received sequence $Y^n$. Due to the code construction process, $X^n(1,1)$, $X^n(i_s,j),(i_s\neq 1)$, and $Y^n$ are mutually independent. Hence, by using semantically joint AEP (Theorem \ref{Sem_JAEP_theorem}), the probability that $X^n(i_s,j)$ and $Y^n$ are jointly synonymous typical is written as
\begin{equation}\label{Prb_Errorevent2}
P(E_{i_s})\leq \sum_{j=1}^{2^{nR_s}}2^{-n(I^s(\tilde{X};\tilde{Y})-3\epsilon)}, i_s\in \left\{2,\cdots,2^{nR}\right\}.
\end{equation}

Consequently, combining (\ref{Prb_Errorevent1}) and (\ref{Prb_Errorevent2}), the error probability can be derived as
\begin{equation}
\begin{aligned}
\text{Pr}(\mathcal{E})&=\text{Pr}\left(\mathcal{E}\left|X^{n}(1,1)\right.\right)\\
                                 &\leq P(E_{1}^c) + \sum_{i_s=2}^{2^{nR}}P(E_{i_s})\\
                                 &\leq \epsilon+\sum_{i_s=2}^{2^{nR}}\sum_{j=1}^{2^{nR_s}}2^{-n(I^s(\tilde{X};\tilde{Y})-3\epsilon)}\\
                                 & = \epsilon+\left(2^{nR}-1\right)2^{nR_s}2^{-n(I^s(\tilde{X};\tilde{Y})-3\epsilon)}\\
                                 &\leq \epsilon+2^{-n\left(I^s(\tilde{X};\tilde{Y})-R-R_s-3\epsilon\right)}\\
                                 &\leq 2\epsilon.
\end{aligned}
\end{equation}
This formula holds for sufficiently large $n$ and $I^s(\tilde{X};\tilde{Y})-R-R_s-3\epsilon>0$.

Therefore, if the code rate satisfies $R_t=R+R_s<I^s(\tilde{X};\tilde{Y})$, the error probability can tend to zero with the suitable $\epsilon$ and $n$. In particular, by setting $R_s=0$ (i.e., each synonymous set contains only one codeword), the above condition reduces to $R<I_s(\tilde{X};\tilde{Y})$. Hence, for any semantic rate $R<C_s$?, there always exists a code with vanishing error probability.

In addition, by Theorem \ref{SJTS_theorem}, the size of synonymous codeword set satisfies
\begin{equation}
\begin{aligned}
1\leq 2^{nR_s}& \leq  2^{n\left(H(X;Y)-H_s(\tilde{X};\tilde{Y})\right)}\\
                        &= 2^{n\left(I^s(\tilde{X};\tilde{Y})-I(X;Y)\right)}
\end{aligned}
\end{equation}
for sufficiently large $n$. Thus we have $0 \leq R_s \leq (I^s(\tilde{X};\tilde{Y})-I(X;Y))$ and derive that
\begin{equation}\label{equation112}
I(X;Y)+R_s\leq I^s(\tilde{X};\tilde{Y}).
\end{equation}

Then the total code rate $R_t=R+R_s$ can be upper bounded by
\begin{equation}
\begin{aligned}
R_t &\leq I(X;Y)+I^s(\tilde{X};\tilde{Y})-I(X;Y)\\
    &= I^s(\tilde{X};\tilde{Y}).
\end{aligned}
\end{equation}
In summary, for any code rate below the semantic capacity $C_s=\max_{f_{xy}}\max_{p(x)}I^s(\tilde{X};\tilde{Y})$, we can construct a code with the error probability being close to zero for sufficiently large $n$. This proves the achievability of theorem.
\end{proof}

\begin{remark}
In the classic channel coding theorem, in order to satisfying the requirement of reliable communication, the code rate must be lower than the channel capacity $C$. On the contrary, in the semantic channel coding, the code rate can be further increased to the semantic channel capacity $C_s$ under the condition of keeping the semantic reliability. Using synonymous codeword set to present the semantic sequence is the key technique to achieve this goal. By (\ref{equation112}), due to $C+R_s\leq C_s$, with the increasing of the number of codewords in a synonymous codeword set, the semantic code rate gradually grows and approaches the semantic capacity. Similar to the classic counterpart, although the proof of the semantic channel coding theorem is also an existence method due to using the random coding, it may provide some hints to construct the channel codes approaching the semantic capacity.
\end{remark}

In order to prove the converse, we first illustrate the relationship between sequential syntactic mutual information and sequential semantic mutual information.
\begin{lemma}\label{lemma8}
Assume $\tilde{X}^n$ is the transmitted semantic sequence over a discrete memoryless channel and the received sequence is $Y^n$, we have
\begin{equation}
I(\tilde{X}^n; Y^n)\leq I^s(\tilde{X}^n;\tilde{Y}^n), \text{ for all } p(x^n).
\end{equation}
\end{lemma}
\begin{proof}
Due to the definition of discrete memoryless channel, we can write the sequential mutual information as
\begin{equation}
\begin{aligned}
&I(\tilde{X}^n; Y^n)-I^s(\tilde{X}^n;\tilde{Y}^n)\\
&=H(\tilde{X}^n)+H(Y^n)-H(\tilde{X}^n,Y^n)\\
&-\left[H(X^n)+H(Y^n)-H_s(\tilde{X}^n,\tilde{Y}^n)\right]\\
&= \sum_{k=1}^{n} \left[H_s(\tilde{X}_k)-H(X_k)+H_s(\tilde{X}_k,\tilde{Y}_k)-H(\tilde{X}_k,Y_k)\right]\leq 0.
\end{aligned}
\end{equation}
Due to $H(\tilde{X}_k)=H_s(\tilde{X}_k)\leq H(X_k)$ and $H_s(\tilde{X}_k,\tilde{Y}_k)\leq H(\tilde{X}_k,Y_k)$ (Theorem \ref{theorem2}), we prove the lemma.
\end{proof}

Then we investigate the sequential semantic mutual information many times using of discrete memoryless channel.
\begin{lemma}\label{lemma9}
Assume $X^n$ is the transmitted codeword over a discrete memoryless channel and the received sequence is $Y^n$, under a jointly synonymous mapping $f_{xy}^n$, we have
\begin{equation}
I^s(\tilde{X}^n;\tilde{Y}^n)\leq nC_s, \text{ for all } p(x^n).
\end{equation}
\end{lemma}
\begin{proof}
Due to the definitions of discrete memoryless channel and up-semantic mutual information, we can write
\begin{equation}
\begin{aligned}
I^s(\tilde{X}^n;\tilde{Y}^n)&=H(X^n)+H(Y^n)-H_s(\tilde{X}^n;\tilde{Y}^n)\\                                  
                                         &\leq \sum_{k=1}^{n} \left[H(X_k)+H(Y_k)\right]-\sum_{k=1}^{n} H_s(\tilde{X}_k,\tilde{Y}_k)\\
                                         &=\sum_{k=1}^{n}\left[H(X_k)+H(Y_k)- H_s(\tilde{X}_k,\tilde{Y}_k)\right]\\
                                         &\leq nC_s.
\end{aligned}
\end{equation}
The first inequality is from the property of sequential entropy and the second is from the definition of semantic channel capacity.
\end{proof}

We now prove the converse to the semantic channel coding theorem.
\begin{proof} (Converse to Theorem \ref{Sem_CCT}, (Semantic Channel Coding Theorem)):

Let $W_s$ denote a semantic index uniformly drawn from $\left\{1,2,\cdots,2^{nR}\right\}$. The error probability can be written as $P_e^{(n)}=\text{Pr}(\hat{W}_s\neq W_s)$. So we have
\begin{equation}
\begin{aligned}
nR &=H(W_s)=H(W_s|Y^n)+I(W_s;Y^n)   \\
    &\overset{(a)}{\leq} H(W_s|Y^n)+I(\tilde{X}^n(W_s);Y^n)  \\
    &\overset{(b)}{\leq} 1+P_e^{(n)} nR +I(\tilde{X}^n(W_s);Y^n) \\
    &\overset{(d)}{\leq} 1+P_e^{(n)} nR +nC_.
\end{aligned}
\end{equation}
Inequality $(a)$ holds since $\tilde{X}^n(W_s)$ is the function of $W_s$ and $(b)$ is from Fano's inequality. In equality $(c)$ is from Lemma \ref{lemma8}.

So the total code rate satisfies
\begin{equation}\label{Semrate_inequality}
\begin{aligned}
R_t=R+R_s& \leq C+H(X,Y)-H_s(\tilde{X},\tilde{Y})+P_e^{(n)} R+\frac{1}{n}\\
                   & \leq C_s+P_e^{(n)} R+\frac{1}{n},
\end{aligned}
\end{equation}
which means $R_t\leq C_s$ for $n\to\infty$.
On the other hand, we can rewrite (\ref{Semrate_inequality}) as
\begin{equation}
\begin{aligned}
P_e^{(n)}&\geq \frac{R_t}{R}-\frac{1}{nR}-\frac{C_s}{R}\\
               &\geq \frac{R_t}{R}\left( 1- \frac{1}{nR_t}-\frac{C_s}{R_t} \right).
\end{aligned}
\end{equation}
This formula indicates that if $R_t>C_s$, the error probability is larger than zero for sufficiently large $n$ and the reliable transmission of semantic information can not be fulfilled. So we complete the proof.
\end{proof}

In classic communication systems, the channel capacity is a fundamental limitation of data reliable transmission. In the past seventy years, people invented many powerful channel codes to approach capacity, such as turbo, LDPC and polar codes. Similarly, the semantic capacity is also a key parameter for the semantic transmission. In the future, the construction of channel codes approaching semantic capacity will become one core issue of semantic communication.

\subsection{Semantic Channel Coding Method}
We now investigate the semantic channel coding method. Given a $\left(2^{n(R+R_s)},n\right)$ channel code with length $n$ and semantic rate $R$, the codebook $\mathcal{C}$ can be divided into $2^{nR}$ synonymous codeword groups $\mathcal{C}_s(i_s), i_s\in \{1,2,\cdots,2^{nR}\}$ and each group has $2^{nR_s}$ synonymous codewords. Consider the synonymous mapping, we propose a new method, named as maximum likelihood group (MLG) decoding algorithm to decode this semantic code. The basic idea is to calculate the likelihood probability of the received signal on a synonymous group and compare all the group likelihood probabilities so as to select a group with the maximum probability as the final decoding result.

\begin{definition}
Assume one codeword $x^n\in \mathcal{C}_s(i_s)$ is transmitted on the discrete memoryless channel with the transition probability $p(y^n|x^n)$, the group likelihood probability is defined as
\begin{equation}
P(y^n|\mathcal{C}_s(i_s))\triangleq \prod_{l=1}^{2^{nR_s}}p(y^n|x^n(i_s,l)).
\end{equation}
\end{definition}

So the maximum likelihood group decoding rule is written as
\begin{equation}
\begin{aligned}
\hat{i}_s &=\text{arg}\max_{i_s} P(y^n|\mathcal{C}_s(i_s)) \\
              &=\text{arg}\max_{i_s} \prod_{l=1}^{2^{nR_s}}p(y^n|x^n(i_s,l)) .
\end{aligned}
\end{equation}
Equivalently, this rule can also presented as a logarithmic version,
\begin{equation}
\hat{i}_s =\text{arg}\max_{i_s} \sum_{l=1}^{2^{nR_s}}\ln p(y^n|x^n(i_s,l)) .
\end{equation}

Hence, we can calculate all the group likelihood probabilities and select one group with the maximum probability as the final decision. The index $\hat{i}_s$ indicates the estimation of semantic information $\hat{\tilde{x}}^n$.

Next, we discuss the MLG decoding in the additive white Gaussian noise (AWGN) channel. When a signal is transmitted over the AWGN channel, the received signal can be represented by an equivalent low-pass signal sampled at time $k$:
\begin{equation}
y_k= s_k+z_k
\end{equation}
where $s_k=\left\{\pm\sqrt{E_s}\right\}$ is the binary phase shifted key (BPSK) signal, $z_k$ is a sample of a zero-mean complex Gaussian noise process with variance $\sigma^2=N_0/2$. Let $E_s$ be the symbol energy and $N_0$ denote the single-sided power spectral density of the additive white noise. So the symbol signal-to-noise ratio (SNR) is defined as $\frac{E_s}{N_0}$.

Assume one codeword $x^n(i_s,l)$ is mapped into a transmitted signal vector $s^n(i_s,l)=$\\$\sqrt{E_s}\left(1^n-2x^n(i_s,l)\right)$ where $1^n$ is the all-one vector with the length of $n$, by using MLG rule, we can write
\begin{equation}
\begin{aligned}\label{MLG_rule}
\hat{i}_s &=\text{arg}\max_{i_s} \sum_{l=1}^{2^{nR_s}}\ln p(y^n|x^n(i_s,l)) \\
              &=\text{arg}\max_{i_s} \sum_{l=1}^{2^{nR_s}}\ln \left[\frac{1}{\sqrt{2\pi \sigma^2}} e^{-\frac{\left\|y^n-s^n(i_s,l)\right\|^2}{2\sigma^2}}\right]\\
              &=\text{arg}\min_{i_s} \sum_{l=1}^{2^{nR_s}} \left\|y^n-s^n(i_s,l)\right\|^2\\
              &=\text{arg}\min_{i_s} d_{\text{E}}^2 \left(y^n, \mathcal{C}_s(i_s)\right),
\end{aligned}
\end{equation}
where $d_{\text{E}}^2 \left(y^n, \mathcal{C}_s(i_s)\right)=\sum_{l=1}^{2^{nR_s}} \left\|y^n-s^n(i_s,l)\right\|^2$ is the squared Euclidian distance between the receive vector $y^n$ and the code group $\mathcal{C}_s(i_s)$. Thus the MLG rule in AWGN channel is transformed into the minimum distance grouping decoding rule.

Now we investigate the group-wise error probability (GEP) of semantic channel code.
\begin{theorem}
Given a semantic channel code $\mathcal{C}$ with equipartition code groups $\mathcal{C}_s$, the GEP between $\mathcal{C}_s(i_s)$ and $\mathcal{C}_s(j_s)$ is upper bounded by
\begin{equation}
P\left({\mathcal{C}_s}(i_s)\to{\mathcal{C}_s}(j_s)\right) \leq \exp\left\{-d_\text{GH}(\mathcal{C}_s(i_s),\mathcal{C}_s(j_s)) \frac{E_s}{N_0}\right\},
\end{equation}
where $d_\text{GH}$ denotes the group Hamming distance which is defined as following
\begin{equation}
\begin{aligned}
& d_\text{GH}(\mathcal{C}_s(i_s),\mathcal{C}_s(j_s))=\\
&\min_{m} \frac{\left[\sum_{l=1}^{2^{nR_s}} d_H(x^n(i_s,m),x^n(j_s,l))-\sum_{l=1,l\neq m}^{2^{nR_s}} d_H( x^n(i_s,m),x^n(i_s,l))\right]^2}
       { \|\sum_{l=1}^{2^{nR_s}} \left(x^n(j_s,l)-x^n(i_s,l)\right)\|^2 }.
\end{aligned}
\end{equation}
\end{theorem}

\begin{proof}
Assume one codeword $x^n(1,1)$ in the code group $\mathcal{C}_s(1)$ is transmitted, the received signal vector can be represented as follows
\begin{equation}\label{equation125}
y^n=s^n(1,1)+z^n,
\end{equation}
where $z^n \sim \mathcal{N}(\mathbf{0},\sigma^2\mathbf{I})$ is the Gaussian noise vector.

Suppose a codeword $x^n(j_s,l)\in \mathcal{C}_s(j_s), j_s\neq 1$ is mapped into the signal vector $s^n(j_s,l)$. By using the MLG rule, if a group-wise error occurs, the Euclidian distance between the received vector and the transmitted signal vector group satisfy the inequality
\begin{equation}
d_{\text{E}}^2 \left(y^n, \mathcal{C}_s(i_s)\right)>d_{\text{E}}^2 \left(y^n, \mathcal{C}_s(j_s)\right).
\end{equation}
Substituting (\ref{MLG_rule}) and (\ref{equation125}) into this inequality, we have
\begin{equation}
\begin{aligned}
&\sum_{l=1}^{2^{nR_s}} \left\|y^n-s^n(1,l)\right\|^2\\
&>\sum_{l=1}^{2^{nR_s}} \left\|y^n-s^n(j_s,l)\right\|^2 \Rightarrow\\
&\left\|z^n\right\|^2+\sum_{l=2}^{2^{nR_s}} \left\|s^n(1,1)-s^n(1,l)+z^n\right\|^2\\
&>\sum_{l=1}^{2^{nR_s}} \left\|s^n(1,1)-s^n(j_s,l)+z^n\right\|^2.
\end{aligned}
\end{equation}
After some manipulations, the error decision region can be written as
\begin{equation}
\begin{aligned}
\mathcal{H}= \left\{ z^n:  \left[\sum_{l=1}^{2^{nR_s}} \left(s^n(j_s,l)-s^n(1,l)\right)\right] (z^n)^T >\right. &\\
  \frac{1}{2}\left[\sum_{l=1}^{2^{nR_s}} \left\|s^n(1,1)-s^n(j_s,l)\right\|^2\right.&\\
  \left. \left.-\sum_{l=2}^{2^{nR_s}} \left\|s^n(1,1)-s^n(1,l)\right\|^2\right]\right\}&.
\end{aligned}
\end{equation}
Let $d_{\text{E}}^2 \left(s^n(1,1), \mathcal{C}_s(j_s)\right)=\sum_{l=1}^{2^{nR_s}} \left\|s^n(1,1)-s^n(j_s,l)\right\|^2$ denote the distance between the transmit vector $s^n(1,1)$ and the code group $\mathcal{C}_s(j_s)$ and $d_{\text{E}}^2 \left(s^n(1,1), \mathcal{C}_s(1)\right)=\sum_{l=2}^{2^{nR_s}} \left\|s^n(1,1)-s^n(1,l)\right\|^2$ denote the inner distance of code group $\mathcal{C}_s(1)$.

So the codeword-to-group error probability can be derived as
\begin{equation}
\begin{aligned}
  &P\left(x^n(1,1)\to{\mathcal{C}_s}(j_s)\right)\\
  &= Q\left[ \sqrt {\frac{(d_{\text{E}}^2 \left(s^n(1,1), \mathcal{C}_s(j_s)\right)-d_{\text{E}}^2 \left(s^n(1,1), \mathcal{C}_s(1)\right))^2}{\left\|\sum_{l=1}^{2^{nR_s}} \left(s^n(j_s,l)-s^n(1,l)\right)\right\|^22N_0}} \right], \\
\end{aligned}
\end{equation}
where $Q(x)=\frac{1}{\sqrt{2\pi}}\int_x^{\infty} e^{-t^2/2}dt$ is the tail distribution function of the standard normal distribution.

Furthermore, due to $\left(s^n(1,1)-s^n(j_s,l)\right)=2\sqrt{E_s}\left(x^n(j_s,l)-x^n(1,1)\right)$, then we have ${\left\| s^n(1,1)-s^n(j_s,l) \right\|}^2$\\$=4 E_s d_H(x^n(1,1),x^n(j_s,l))$. Additionally, we can derive that $\left\|\sum_{l=1}^{2^{nR_s}} \left(s^n(j_s,l)-s^n(1,l)\right)\right\|^2=4 E_s\left\|\sum_{l=1}^{2^{nR_s}} \left(x^n(j_s,l)-x^n(1,l)\right)\right\|^2$$=4 E_s\|\Delta(\mathcal{C}_s(j_s),\mathcal{C}_s(1))\|^2$. Thus the error probability can be further written as
\begin{equation}
  P\left(x^n(1,1)\to{\mathcal{C}_s}(j_s)\right)= Q\left[ \sqrt {d_\text{GH}(x^n(1,1),\mathcal{C}_s(j_s)) \frac{2E_s}{N_0}} \right],
\end{equation}
where $d_\text{GH}(x^n(1,1),\mathcal{C}_s(j_s))=$$[\sum_{l=1}^{2^{nR_s}} d_H(x^n(1,1),x^n(j_s,l))$\\$-\sum_{l=2}^{2^{nR_s}} d_H( x^n(1,1),x^n(1,l))]^2$$/\|\Delta(\mathcal{C}_s(j_s),\mathcal{C}_s(1))\|^2$ denotes the codeword-to-group Hamming distance.

Furthermore, using the inequality $Q(x)\leq e^{-\frac{x^2}{2}}$, the codeword-to-group error probability can be upper bounded by
\begin{equation}
P\left(x^n(1,1)\to{\mathcal{C}_s}(j_s)\right) \leq e^ {-d_\text{GH}(x^n(1,1),\mathcal{C}_s(j_s)) \frac{E_s}{N_0}}.
\end{equation}

Averaging over all the codewords of the group $\mathcal{C}_s(1)$, we obtain the upper bound of GEP as follows
\begin{equation}
\begin{aligned}
P\left({\mathcal{C}_s}(1)\to{\mathcal{C}_s}(j_s)\right) &\leq \sum_{l=1}^{2^{nR_s}} \frac{1}{2^{nR_s}}e^{-d_\text{GH}(x^n(1,l),\mathcal{C}_s(j_s)) \frac{E_s}{N_0}}\\
& \leq \exp\left\{-d_\text{GH}(\mathcal{C}_s(1),\mathcal{C}_s(j_s)) \frac{E_s}{N_0}\right\}.
\end{aligned}
\end{equation}
So we complete the proof.
\end{proof}

In the ML decoding, the minimum Hamming distance $d_{\text{H,min}}$ determines the error performance of one linear channel code. Similarly, in the MLG decoding, the minimum group Hamming distance $d_{\text{GH,min}}=\min d_{\text{GH}}(\mathcal{C}_s(i_s),\mathcal{C}_s(j_s))$ dominates the performance of semantic channel code.

\begin{example}
We now give an example of semantic code constructed based on (7,4) Hamming code with synonymous mapping and MLG decoding. The codebook is shown in Table \ref{Semantic_Hammingcode}. All the sixteen codewords are divided into eight code groups and each group has two synonymous codewords. For an instance, $\mathcal{C}_s(1)$ has two codewords $(0000000)$ and $(1101000)$ and its semantic sequence is $(000)$. So this code can be regarded as a (7,3) semantic Hamming code with code rate $R=\frac{3}{7}$ and $R_s=\frac{1}{7}$.

By using ML decoding, the union bound of the error probability is
\begin{equation}
P_e \leq \sum_{d=d_{\text{H,min}}}^{n} A_d Q\left(\sqrt{2d \frac{E_s}{N_0}}\right) \leq \sum_{d=d_{\text{H,min}}}^{n} A_d e^{-d\frac{E_s}{N_0}}.
\end{equation}

Since the minimum Hamming distance of this code is $d_{\text{H,min}}=3$ and distance spectrum is $\{A_3=8,A_4=6,A_7=1\}$, the error probability of ML decoding is upper bounded by
\begin{equation}\label{ML_UB}
P_e^{\text{ML}} \leq 8e^{-3\frac{E_s}{N_0}}+6e^{-4\frac{E_s}{N_0}}+e^{-7\frac{E_s}{N_0}}.
\end{equation}

Let $\{A_{d_1,d_2}\}$ denote the group distance spectrum and $d_1$ and $d_2$ mean the codeword-to-group Hamming distance. By using MLG decoding, the union bound of the error probability is
\begin{equation}
\begin{aligned}
P_e &\leq \sum_{d_1,d_2=d_{\text{GH,min}}}^{n} \frac{A_{d_1,d_2}}{2}\left[Q\left(\sqrt{2d_1 \frac{E_s}{N_0}}\right)+Q\left(\sqrt{2d_2 \frac{E_s}{N_0}}\right) \right]\\
&\leq \sum_{d_1,d_2=d_{\text{GH,min}}}^{n} \frac{A_{d_1,d_2}}{2}\left(e^{-d_1\frac{E_s}{N_0}}+e^{-d_2\frac{E_s}{N_0}}\right).
\end{aligned}
\end{equation}

So the minimum group Hamming distance of this code is $d_{\text{GH,min}}=2$ and group distance spectrum is $\{A_{2,2}=6,A_{4,4}=1\}$. The corresponding upper bound of the MLG decoding is
\begin{equation}\label {MLG_UB}
P_e^{\text{MLG}} \leq 6e^{-2\frac{E_s}{N_0}}+e^{-4\frac{E_s}{N_0}}.
\end{equation}

Compare with (\ref{ML_UB}) and (\ref{MLG_UB}), we find that the minimum distance of semantic Hamming code is decreased. However, for a long code length and well-designed synonymous mapping, the error performance of MLG decoding will be better than that using ML decoding.

\begin{table}[tp]
\centering
\caption{The codebook of (7,3) semantic Hamming code with synonymous mapping.} \label{Semantic_Hammingcode}
\begin{tabular}{|c|c|c|c|}
  \hline Index $i_s$ &  Semantic sequence   & Hamming code group $\mathcal{C}_s(i_s)$   \\
  \hline         1         &             000               &           \{0000000,   1101000\}           \\
  \hline         2         &             001               &           \{0110100,   1011100\}           \\
  \hline         3         &             010               &           \{1110010,   0011010\}          \\
  \hline         4         &             011               &           \{1000110,   0101110\}           \\
  \hline         5         &             100               &           \{1010001,   0111001\}          \\
  \hline         6         &             101               &           \{1100101,   0001101\}           \\
  \hline         7         &             110               &           \{0100011,   1001011\}           \\
  \hline         8         &             111               &           \{0010111,   1111111\}           \\
  \hline
\end{tabular}
\end{table}
\end{example}

\begin{remark}
From the viewpoint of practical application, semantic channel codes are a new kind of channel codes. Synonymous mapping provides a valuable idea for the construction and decoding of semantic channel codes. Unlike the traditional channel codes, semantic channel codes should optimize the minimum group Hamming distance. How to design an optimal synonymous mapping to cleverly partition the code group is significant for the design of semantic codes. Non-equipartition mapping may be more flexible than the equipartition mapping. On the other hand, the optimal decoding of semantic codes is the MLG rule rather ML rule. However, due to the exponent complexity of MLG decoding algorithm, it is not practical for application. So we should pursuit lower complexity decoding algorithms for the semantic channel codes in the future.
\end{remark}

\section{Semantic Lossy Source Coding}
\label{section_VIII}
In this section, we mainly discuss the semantic lossy source coding. Firstly, we investigate the semantic distortion measure and extend the concept of jointly typical sequence in the semantic sense. Then we prove the semantic rate-distortion coding theorem by using JAEP and synonymous typical set, which means that the semantic rate distortion function, $R_s(D)=\min_{f_x,f_{\hat{x}}}\min_{p(Y|X): \mathbb{E}d_s(\tilde{X},\hat{\tilde{X}})\leq D}I_s(\tilde{X};\hat{\tilde{X}})$, namely the minimum down semantic mutual information, is the lowest compression rate achieving by semantic lossy source coding.

\subsection{Semantic Distortion and Jointly Typical Set}
Assume a discrete source $X\sim p(x), x\in\mathcal{X}$ with the associated semantic source $\tilde{X}$ produces a sequence $X^n\sim p(x^n)$, after the quantization and reproduction, equivalently, through a test channel with the transition probabilities $p(\hat{x}^n|x^n)$, we obtain the representative sequence $\hat{X}^n\sim p(\hat{x}^n)$. Here, $p(\hat{x}^n)=\sum_{p(x^n)}p(x^n)p(\hat{x}^n|x^n)$ denotes the distribution of reproduction sequences.

\begin{definition}
For a semantic sequence $\tilde{x}^n$, under the synonymous mapping $f_{x}^n$, the associated source sequence is $x^n$. After going through a test channel with the transition probabilities $p(\hat{x}^n|x^n)$, we obtain the reconstruction sequence $\hat{x}^n$. Under the de-synonymous mapping $g_{\hat{x}}^n$, the associated semantic sequence is $\hat{\tilde{x}}^n$. Thus the semantic distortion between sequences $\tilde{x}^n$ and $\hat{\tilde{x}}^n$ is defined by
\begin{equation}
\begin{aligned}
d_s(\tilde{x}^n,\hat{\tilde{x}}^n)&=\frac{1}{n}\sum_{k=1}^{n}d_s(\tilde{x}_k,\hat{\tilde{x}}_k)\\
                                                  &=\frac{1}{n}\sum_{k=1}^{n}d_s(\mathcal{X}_{\tilde{x}_k},\hat{\mathcal{X}}_{\hat{\tilde{x}}_k}).
\end{aligned}
\end{equation}

Then the average semantic distortion over the semantic sequences is defined as
\begin{equation}
\begin{aligned}
\mathbb{E}\left[d_s (\tilde{X}^n,\hat{\tilde{X}}^n)\right]=\sum_{(x^n,\hat{x}^n)}p\left(x^n\right)p(\hat{x}^n|x^n)d_s(\tilde{x}^n,\hat{\tilde{x}}^n).
\end{aligned}
\end{equation}
\end{definition}

Given the synonymous mapping $f_x^n$, let $A_{\epsilon}^{(n)}(X^n)$ and $\tilde{A}_{\epsilon}^{(n)}(\tilde{X}^n)$ denote the syntactically and semantically typical set of source sequence respectively. Correspondingly, the synonymous typical set is denoted as $B_{\epsilon}^{(n)}(\tilde{x}^n\to X^n)$. According to Theorem \ref{SynTSet_Theorem}, we have
\begin{equation}
 2^{n\left(H(X)-H_s(\tilde{X})-\epsilon\right)}\leq\left|B_{\epsilon}^{(n)}(\tilde{x}^n\to X^n)\right| \leq 2^{n\left(H(X)-H_s(\tilde{X})+\epsilon\right)}
\end{equation}
for sufficiently large $n$.

Furthermore, given a test channel with the transition probabilities $p(\hat{x}^n|x^n)$, the distribution of reproduction sequence is written as $p(\hat{x}^n)=\sum_{p(x^n)}p(x^n)p(\hat{x}^n|x^n)$. Similarly, given the synonymous mapping $f_{\hat{x}}^n$,  let $A_{\epsilon}^{(n)}(\hat{X}^n)$ and $\tilde{A}_{\epsilon}^{(n)}(\hat{\tilde{X}}^n)$ denote the syntactically and semantically typical set of reproduction sequence respectively. Moreover, the associated synonymous typical set is addressed as $B_{\epsilon}^{(n)}(\hat{\tilde{x}}^n\to\hat{X}^n)$. By Theorem \ref{SynTSet_Theorem}, we also have
\begin{equation}
2^{n\left(H(\hat{X})-H_s(\hat{\tilde{X}})-\epsilon\right)} \leq\left|B_{\epsilon}^{(n)}(\hat{\tilde{x}}^n\to\hat{X}^n)\right|\\
 \leq 2^{n\left(H(\hat{X})-H_s(\hat{\tilde{X}})+\epsilon\right)}
\end{equation}
for sufficiently large $n$.
Since all the synonymous typical sets have almost the same size, we can abbreviate them as $B_{\epsilon}^{(n)}(\hat{X}^n)$ and $B_{\epsilon}^{(n)}(X^n)$ respectively.

Given the joint distribution $p(x,\hat{x})$ on $\mathcal{X}\times \hat{\mathcal{X}}$, let $A_{\epsilon}^{(n)}(X^n,\hat{X}^n)$ denote the the syntactically jointly typical set of source sequence and reproduction sequence and $\tilde{A}_{\epsilon}^{(n)}$ be the semantically jointly typical set.

The typical sequence mapping for lossy source coding is depicted in Fig. \ref{SemSyn_Mapping_lossy_coding}. Under the synonymous mappings $f_x^n$ and $f_{\hat{x}}^n$, the semantic source sequence $\tilde {x}^n$ and semantic reconstruction sequence $\hat{\tilde {x}}^n$ are separately mapped into synonymous typical sets $B_{\epsilon}^{(n)}(X^n)$ and $B_{\epsilon}^{(n)}(\hat{X}^n)$ . Furthermore, by the conditional probability $p(\hat{x}^n|x^n)$, the typical sequences in these sets can compose the jointly typical sequence.
\begin{figure*}[htbp]
\setlength{\abovecaptionskip}{0.cm}
\setlength{\belowcaptionskip}{-0.cm}
  \centering{\includegraphics[scale=0.9]{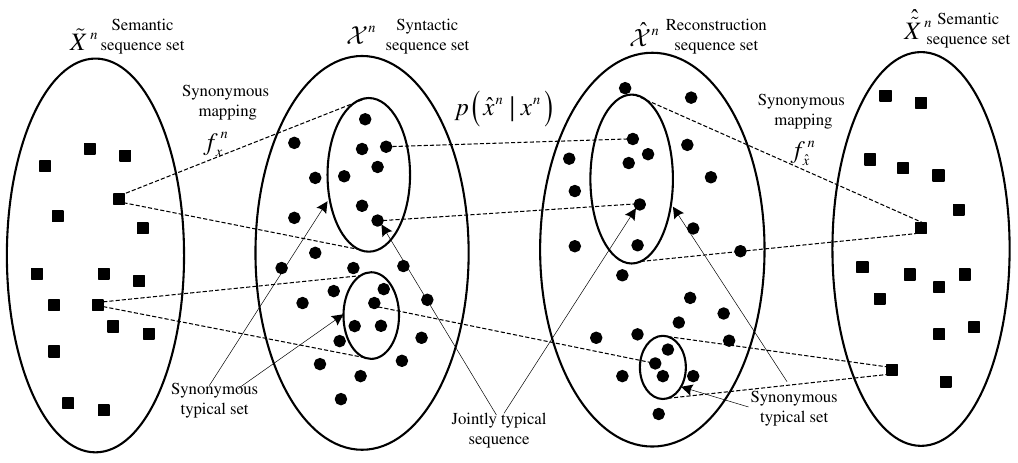}}
  \caption{Typical sequence mapping for lossy source coding.}\label{SemSyn_Mapping_lossy_coding}
\end{figure*}

Then as the consequence of the JAEP, we present the properties of semantically jointly typical set as following.
\begin{theorem}\label{Distortion_JAEP_theorem}
Let $(X^n, \hat{X}^n)$ be a sequence pair with length $n$ drawn i.i.d. according to $p(x^n,\hat{x}^n)$. By using synonymous mapping $f_x^n$ ($f_{\hat{x}}^n$), the semantic sequence $\tilde{X}^n$ ($\hat{\tilde{X}}^n$) is mapped into a syntactic sequence $X^n$ ($\hat{X}^n$).
\begin{enumerate}[(1)]
    \item $\text{Pr}\left((\tilde{X}^n,\hat{\tilde{X}}^n)\in \tilde{A}_{\epsilon}^{(n)}\right)>1-\epsilon$ for $n$ sufficiently large.
    \item $\left(1-\epsilon\right) 2^{n\left(H_s(\tilde{X},\hat{\tilde{X}})-\epsilon\right)}\leq \left|\tilde{A}_{\epsilon}^{(n)}\right| \leq 2^{n\left(H_s(\tilde{X},\hat{\tilde{X}})+\epsilon\right)}$ for $n$ sufficiently large.
    \item Given $\tilde{Z}^n$ and $\hat{\tilde{Z}}^n$ are two independent semantic sequences, if $Z^n$ and $\hat{Z}^n$ are two associated syntactic sequences with the same distributions as $X^n$ and $\hat{X}^n$, i.e., $Z^n\sim p(x^n)$ and $\hat{Z}^n\sim p(\hat{x}^n)$, for $n$ sufficiently large, we have
        \begin{equation}\label{seperate_typical_equality}
        \begin{aligned}
        (1-\epsilon)2^{-n(I_s(\tilde{X};\hat{\tilde{X}})+3\epsilon)}&\leq \text{Pr}\left((\tilde{Z}^n,\hat{\tilde{Z}}^n)\in \tilde{A}_{\epsilon}^{(n)}\right)\\
        & \leq 2^{-n(I_s(\tilde{X};\hat{\tilde{X}})-3\epsilon)}.
        \end{aligned}
        \end{equation}
\end{enumerate}
\end{theorem}
\begin{proof}
Property (1) and (2) are restatements of semantically jointly typical set.

Assume $\tilde{Z}^n$ and $\hat{\tilde{Z}}^n$ are two independent semantic sequences, the associated syntactic sequences $Z^n$ and $\hat{Z}^n$ are independent but have the same distributions as $X^n$ and $\hat{X}^n$. Thus we can establish two one-to-one mappings ${\tilde{z}}^n \leftrightarrow z^n \leftrightarrow x^n$ and $\hat{\tilde{z}}^n \leftrightarrow \hat{z}^n \leftrightarrow \hat{x}^n$, then,
\begin{equation}
\begin{aligned}
&\text{Pr}\left((\tilde{Z}^n,\hat{\tilde{Z}}^n)\in {\tilde{A}}_{\epsilon}^{(n)}\right) \\
&=\text{Pr}\left(Z^n\in B_{\epsilon}^{(n)}(X^n),\hat{Z}^n\in B_{\epsilon}^{(n)}(\hat{X}^n), (Z^n,\hat{Z}^n)\in {A}_{\epsilon}^{(n)}\right) \\
&= \sum_{ (\tilde{x}^n,\hat{\tilde{x}}^n) \leftrightarrow ({x}^n,\hat{x}^n)\in {A}_{\epsilon}^{(n)}} p(\tilde{x}^n)p(\hat{\tilde{x}}^n)\\
&\leq 2^{n\left(H(X,\hat{X})+\epsilon\right)}2^{-n\left(H_s(\tilde{X})-\epsilon\right)}2^{-n\left(H_s(\hat{\tilde{X}})-\epsilon\right)}\\
&= 2^{-n\left(I_s(\tilde{X};\tilde{Y})-3\epsilon\right)}.
\end{aligned}
\end{equation}

Using a similar method, we can also derive that
\begin{equation}
\begin{aligned}
&\text{Pr}\left((\tilde{Z}^n,\hat{\tilde{Z}}^n)\in {\tilde{A}}_{\epsilon}^{(n)}\right) \\
&=\sum_{ {A}_{\epsilon}^{(n)}} p(\tilde{x}^n)p(\hat{\tilde{x}}^n)\\
&\geq  (1-\epsilon)2^{n(H(X;\hat{X})-\epsilon)}2^{-n\left(H_s(\tilde{X})+\epsilon\right)}2^{-n\left(H_s(\hat{\tilde{X}})+\epsilon\right)}\\
&= (1-\epsilon) 2^{-n\left(I_s(\tilde{X};\hat{\tilde{X}})+3\epsilon\right)}.
\end{aligned}
\end{equation}
Thus we complete the proof of the theorem.
\end{proof}

\begin{remark}
It should be noted that the probability of Eq. (\ref{seperate_typical_equality}) in Theorem \ref{Distortion_JAEP_theorem} is different from that of Eq. (\ref{jointly_typical_equality}) in Theorem \ref{Sem_JAEP_theorem}. For the former, it indicates the probability of selecting two synonymous typical sequences (equivalently, representing two semantically typical sequences) constituting a syntactically jointly typical pair which is used to evaluate the error probability of jointly typical encoding. On the other hand, for the latter, it represents the probability of selecting two syntactically typical sequences consisting a jointly synonymous typical pair (equivalently, a semantically jointly typical pair) so that it reveals the error probability of jointly typical decoding.
\end{remark}

\subsection{Semantic Rate Distortion Coding Theorem}
We now investigate the problem of semantic lossy source coding. As depicted in Fig. \ref{Semantic_lossy_source_coding}, with the help of synonymous mapping $f_{x}^n$, a semantic index $i_s$ is mapped into the syntactic source sequence $X^n(i)$. Considering the distortion requirement, the encoder selects a suitable codeword $\hat{X}^n(j)$ to represent the source sequence $X^n(i)$ then send the index $i$ to the receiver. Then in the side of receiver, the decoder outputs the reproduction sequence $\hat{X}^n(j)$. After de-mapping $g^n$, we obtain an estimation of semantic index $\hat{i}_s$.

\begin{figure*}[htbp]
\setlength{\abovecaptionskip}{0.cm}
\setlength{\belowcaptionskip}{-0.cm}
  \centering{\includegraphics[scale=0.9]{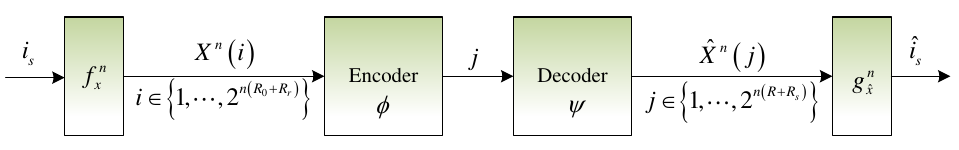}}
  \caption{Block diagram of semantic lossy source coding.}\label{Semantic_lossy_source_coding}
\end{figure*}

For a lossy source coding system, let $p(\hat{X}|X)$ be the conditional probability and $\mathcal{X}$ and $\hat{\mathcal{X}}$ denote the source and reproduction syntactical alphabet respectively. Then the conditional probabilities for the $n$-th extension can be written as
\begin{equation}
p(\hat{x}^n\left|x^n\right.)=\prod_{k=1}^{n}p(\hat{x}_k\left|x_k\right.).
\end{equation}

\begin{definition}
An $\left(M,n\right)$ code for the semantic lossy source coding consists of the following parts:
\begin{enumerate}[(1)]
   \item Two semantic index sets $\mathcal{I}_s=\left\{1,\cdots, M_s\right\}$ and $\mathcal{I}_r=\left\{1,\cdots, M_r\right\}$. Two syntactic index sets $\mathcal{I}=\left\{1,\cdots, M\right\}$ and $\mathcal{I}'=\left\{1,\cdots, M'\right\}$.
   \item A synonymous mapping $f_{x}^n: {\tilde{\mathcal{X}}}^n\to {\mathcal{X}}^n$ generates the set of source sequences, namely, semanticbook, $\mathcal{S}=\{X^n(1),X^n(2),\cdots,X^n(M')\}$. This semanticbook can be partitioned into synonymous sequences subsets $\mathcal{S}_r$.
   \item An encoding function $\phi: {\mathcal{X}}^n\to  \hat{\mathcal{X}}^n$ generates the set of codewords, namely, codebook, $\mathcal{C}=\{\hat{X}^n(1),\hat{X}^n(2),\cdots,\hat{X}^n(M)\}$. Due to the synonymous mapping $f_{\hat{x}}^n$, this codebook can be partitioned into synonymous codeword subsets $\mathcal{C}_s$.
   \item An decoding function $\psi: \hat{\mathcal{X}}^n \to \hat{\mathcal{X}}^n$ outputs the decision syntactic codeword $\hat{X}^n$.
   \item After de-mapping, $g_{\hat{x}}^n(\hat{X}^n)=\hat{i}_s$, the estimated semantic index is obtained. Note that both $\psi$ and $g_{\hat{x}}^n$ are deterministic.
\end{enumerate}
\end{definition}

According to the synonymous mappings $f_{x}^n$ and $f_{\hat{x}}^n$, $\mathcal{S}_r$ is an equivalence class consisting of the synonymous sequences and $\mathcal{C}_s$ is an equivalence class consisting of the synonymous codewords. So we can construct the quotient sets $\mathcal{S}/f_{x}^n=\left\{\mathcal{S}_r\right\}$ with $\left|\mathcal{S}/f_{x}^n\right|=M_r$ and $\mathcal{C}/f_{\hat{x}}^n=\left\{\mathcal{C}_s\right\}$ with $\left|\mathcal{C}/f_{\hat{x}}^n\right|=M_s$. Let $R=\frac{1}{n}\log_2 M_s$ denote the semantic rate of source coding and $R_0=\frac{1}{n}\log_2 M_r$ denote the semantic rate of source sequences. We configure each synonymous set with the same number of sequences or codewords, that is, $\left|\mathcal{S}_r\right|=\frac{M'}{M_r}=2^{nR_r}$ and $\left|\mathcal{C}_s\right|=\frac{M}{M_s}=2^{nR_s}$.

We now give the formal description of semantic rate-distortion coding theorem.
\begin{theorem}\label{Sem_RDT}
(Semantic Rate-Distortion Coding Theorem):

Given an i.i.d. syntactic source $X\sim p(x)$ with the associated semantic source $\tilde{X}$ under the synonymous mapping $f$ and the bounded semantic distortion function $d_s(\tilde{x},\hat{\tilde{x}})$, for each code rate $R>R_s(D)$, there exists a sequence of $\left(2^{n(R+R_s)},n\right)$ codes, when code length $n$ tends to sufficiently large, the semantic distortion satisfies $\mathbb{E}d_s(\tilde{X},\hat{\tilde{X}})<D$.

On the contrary, if $R<R_s(D)$, then for any $\left(2^{n{(R+R_s)}},n\right)$ code, the semantic distortion meets $\mathbb{E}d_s(\tilde{X},\hat{\tilde{X}})>D$ with sufficiently large $n$.
\end{theorem}
\begin{proof}
We first prove the achievability part of the theorem and the converse will be left in the next part.

Given the desired distortion $D$ and the conditional probability distribution $p(\hat{x}|x)$, set the rate-distortion function as $R_s(\frac{D}{1+\epsilon})$ and let $p(\hat{x})=\sum_{x}p(x)p(\hat{x}|x)$. A $\left(2^{n(R_0+R_r)},n\right)$ semanticbook $\mathcal{S}=\left\{X^n(1),X^n(2),\cdots,X^n(2^{n(R_0+R_r)})\right\}$ can be generated randomly according to the distribution
\begin{equation}
p(x^n)=\prod_{k=1}^{n} p(x_k).
\end{equation}
This set consists of the syntactic sequences to represent the semantic source sequence. Furthermore, these sequences can be uniformly divided into $2^{nR_0}$ groups according to the synonymous mapping $f_x^n(\tilde{x}^n)=\prod_{k=1}^{n}\mathcal{X}_{\tilde{x}_{k}}$. Let $\mathcal{S}_r(i_r)\subset B_{\epsilon}^{(n)}(X^n)$ denote the $i_r$-th synonymous sequence set.

Similarly, a $\left(2^{n(R+R_s)},n\right)$ code $\mathcal{C}=\left\{\hat{X}^{n}(1),\cdots,\hat{X}^{n}(2^{n(R+R_s)})\right\}$ can be generated randomly according to the distribution
\begin{equation}
p(\hat{x}^n)=\prod_{k=1}^{n} p(\hat{x}_k)=\prod_{k=1}^n \sum_{x_k}p(x_k)p(\hat{x}_k|x_k).
\end{equation}
Correspondingly, these codewords can be uniformly divided into $2^{nR}$ groups according to the synonymous mapping
$f_{\hat{x}}^n(\hat{\tilde{x}}^n)=\prod_{k=1}^{n}\hat{\mathcal{X}}_{\hat{\tilde{x}}_{k}}$. Let $\mathcal{C}_s(i_s)\subset B_{\epsilon}^{(n)}(\hat{X}^n)$ denote the $i_s$-th synonymous codeword set. Thus the semanticbook $\mathcal{S}$ and the codebook $\mathcal{C}$ are produced and shared in the encoder and the decoder.

Similar to the idea in \cite{Book_ElGamal}, we also use jointly typical encoding for the semantic lossy source code. Given a semantic index $w_r$, base on the synonymous mapping $f_x^n$, we determine a syntactic source sequence $x^n$. Furthermore, find an index $w_s$ such that $(\tilde{x}^n,\hat{\tilde{X}}^n(w_s))\in \tilde{A}_{\epsilon}^{(n)}$. Equivalently, we can select one codeword $\hat{X}^n(w_s,l)$ in a synonymous set $\mathcal{C}_s(w_s)$ to present the semantic reconstruction sequence $\hat{\tilde{X}}^n(w_s)$ so as to satisfy $(x^n,\hat{X}^n(w_s,l))\in A_{\epsilon}^{(n)}$. If there is more than one semantic index, choose the smallest one. If no such semantic index exists, let $w_s=1$. After sending $(w_s,l)$ to the decoder, the decoder produces the reconstruction sequence $\hat{x}^n=\hat{X}^n(w_s,l)$. Then under the de-synonymous mapping $g_{\hat{x}}^n$, we obtain the estimated semantic index $\hat{w}_r$.

We now analyze the expected distortion by using semantic coding. Let $W_s$ denote the semantic index chosen by the encoder. We can bound the semantic distortion averaged over the random choice of the semanticbook $\mathcal{S}$ and the codebook $\mathcal{C}$. The encoding error events can be expressed as
\begin{equation}
\mathcal{E}=\left\{\left(X^n,\hat{X}^n(W_s,l)\right)\notin A_{\epsilon}^{(n)}, \hat{X}^n(W_s,l)\in \mathcal{C}_s(W_s)\right\}.
\end{equation}
This error event can be divided into two types of events, that is, $\mathcal{E}=\mathcal{E}_1\bigcup\mathcal{E}_2$, where
\begin{equation}
\mathcal{E}_1=\left\{X^n \notin A_{\epsilon}^{(n)}(X^n)\right\}
\end{equation}
and
\begin{equation}
\begin{aligned}
\mathcal{E}_2=&\left\{X^n\in \mathcal{S}_r(w_r), \hat{X}^n(w_s,l)\in \mathcal{C}_s(w_s),\right.\\
                        & \left(X^n, \hat{X}^n(w_s,l)\right)\notin A_{\epsilon}^{(n)}(X^n,\hat{X}^n),\\
                        &\left. \text{for all } w_s\in \left\{1,2,\cdots,2^{nR}\right\}, l\in \left\{1,2,\cdots,2^{nR_s}\right\}\right\}.
\end{aligned}
\end{equation}

So the error probability can be upper bounded by $P_e^{(n)}=\text{Pr}(\mathcal{E})\leq P(\mathcal{E}_1)+P(\mathcal{E}_2)$.
For the first term, by the weak law of large numbers, the probability tends to zero as $n\to\infty$. For the second term, recall that $\mathcal{S}_r(w_r)\subset B_{\epsilon}^{(n)}(X^n)$ and $\mathcal{C}_s(w_s)\subset B_{\epsilon}^{(n)}(\hat{X}^n)$, we derive the probability shown as follows
\begin{equation}
\begin{aligned}
P(\mathcal{E}_2)&=\sum_{x^n\in A_{\epsilon}^{(n)}} p(x^n) \\
&\cdot P\left(X^n\in \mathcal{S}_r(w_r), \hat{X}^n(w_s,l)\in \mathcal{C}_s(w_s), (X^n, \hat{X}^n(w_s,l))\notin A_{\epsilon}^{(n)}, \forall w_s\in \mathcal{I}_s|X^n=x^n\right)\\
&=\sum_{x^n\in A_{\epsilon}^{(n)}} p(x^n) \prod_{w_s=1}^{2^{nR}} P\left(X^n\in \mathcal{S}_r(w_r), \hat{X}^n(w_s,l)\in \mathcal{C}_s(w_s),(X^n, \hat{X}^n(w_s,l))\notin A_{\epsilon}^{(n)}\right)\\
&=\sum_{x^n\in A_{\epsilon}^{(n)}} p(x^n) \left[P\left(X^n\in \mathcal{S}_r(w_r), \hat{X}^n(1,l)\in \mathcal{C}_s(1),(X^n, \hat{X}^n(1,l))\notin A_{\epsilon}^{(n)}\right)\right]^{2^{nR}}.
\end{aligned}
\end{equation}

Since $x^n\in A_{\epsilon}^{(n)}$ and $\hat{X}^n(1,l)\sim \prod_{k=1}^{n}p(\hat{x}_k)$, by Theorem \ref{Distortion_JAEP_theorem}, for sufficiently large $n$, we have
\begin{equation}
\begin{aligned}
&P\left(X^n\in \mathcal{S}_r(w_r), \hat{X}^n(1,l)\in \mathcal{C}_s(1),(X^n, \hat{X}^n(1,l))\in A_{\epsilon}^{(n)}\right)\\
& \geq (1-\epsilon) 2^{-n(I_s(\tilde{X};\hat{\tilde{X}})+3\epsilon)}.
\end{aligned}
\end{equation}

Furthermore, since $(1-x)^m\leq e^{-mx}$ for $x\in [0,1]$ and $m\geq 0$, we have
\begin{equation}
\begin{aligned}
P(\mathcal{E}_2) & \leq \left(1-(1-\epsilon)2^{-n(I_s(\tilde{X};\hat{\tilde{X}})+3\epsilon)}\right)^{2^{nR}}\\
                           & \leq \exp\left(-2^{nR}(1-\epsilon)2^{-n(I_s(\tilde{X};\hat{\tilde{X}})+3\epsilon)}\right)\\
                           & =\exp\left(-(1-\epsilon)2^{n(R-I_s(\tilde{X};\hat{\tilde{X}})-3\epsilon)}\right),
\end{aligned}
\end{equation}
which tends to zero as $n\to\infty$ and $\epsilon\to0$ if $R>I_s(\tilde{X};\hat{\tilde{X}})$. Consequently, the error probability $P_e^{(n)}\to 0$.

Hence, we can drive the expectation of semantic distortion as follows
\begin{equation}
\begin{aligned}
\mathbb{E}\left[d_s(\tilde{X}^n,\hat{\tilde{X}}^n(W_s))\right]&=P(\mathcal{E})\mathbb{E}\left[d_s(\tilde{X}^n,\hat{\tilde{X}}^n(W_s))|\mathcal{E}\right] \\
                                                                      &+P(\mathcal{E}^c)\mathbb{E}\left[d_s(\tilde{X}^n,\hat{\tilde{X}}^n(W_s))|\mathcal{E}^c\right]  \\
                                                                      &\leq P_e^{(n)}d_{\max}\\
                                                                      &+ (1-P_e^{(n)})(1+\epsilon) \mathbb{E} (d_s(\tilde{X},\hat{\tilde{X}})),
\end{aligned}
\end{equation}
where $d_{\max}$ is the maximum semantic distortion.

By the assumption that $\mathbb{E} (d_s(\tilde{X},\hat{\tilde{X}}))\leq \frac{D}{1+\epsilon}$, we have $\mathbb{E}\left[d_s(\tilde{X}^n,\hat{\tilde{X}}^n(W_s))\right]\leq D$ for sufficiently large $n$ if $R>I_s(\tilde{X};\hat{\tilde{X}})+3\epsilon$.

Furthermore, by Theorem \ref{SynTSet_Theorem}, we have $1\leq 2^{nR_r}\leq 2^{n(H(X)-H_s(\tilde{X}))}$ and $1\leq 2^{nR_s}\leq 2^{n(H(\hat{X})-H_s(\hat{\tilde{X}}))}$. Hence, it follows that $I_s(\tilde{X};\hat{\tilde{X}}) \leq R'=R+R_r+R_s \leq I_s(\tilde{X};\hat{\tilde{X}})+H(X)-H_s(\tilde{X})+H(\hat{X})-H_s(\hat{\tilde{X}})=I(X;\hat{X})$. Thus we can conclude that if $R_r=R_s=0$ with sufficiently large $n$, the compression rate $R'\to R_s(D)$. On the other hand, if $R_r$ and $R_s$ gradually increase, then the compression rate $R'\to R(D)$. This complete the proof of achievability.
\end{proof}

Next, we prove the converse of the theorem. In order to prove the converse, we first illustrate the relationship between the sequential syntactic mutual information and the sequential semantic mutual information.
\begin{lemma}\label{lemma10}
Assume $X^n$ is the syntactic sequence of discrete memoryless source and the reconstruction semantic sequence is $\hat{\tilde{X}}^n$, we have
\begin{equation}
I(X^n; \hat{\tilde{X}}^n)\geq I_s(\tilde{X}^n;\hat{\tilde{X}}^n), \text{ for all } p(x^n).
\end{equation}
\end{lemma}
\begin{proof}
Due to the definitions of discrete memoryless source, we can write the sequential mutual information as
\begin{equation}
\begin{aligned}
&I(X^n; \hat{\tilde{X}}^n)-I_s(\tilde{X}^n;\hat{\tilde{X}}^n)\\
&=H(X^n)+H(\hat{\tilde{X}}^n)-H(X^n,\hat{\tilde{X}}^n)\\
&-\left[H_s(\tilde{X}^n)+H_s(\hat{\tilde{X}}^n)-H(X^n;\hat{X}^n)\right]\\
&= \sum_{k=1}^{n} \left[H(X_k)-H_s(\tilde{X}_k)+H(X_k,\hat{X}_k)-H(X_k,\hat{\tilde{X}}_k)\right]\geq 0.
\end{aligned}
\end{equation}
Due to $H(\hat{\tilde{X}}_n)=H_s(\hat{\tilde{X}}_n)$, $H(X_k)\geq H_s(\tilde{X}_k)$, and $H(X_k,\hat{X}_k)\geq H(X_k,\hat{\tilde{X}}_k)$ (Theorem \ref{theorem2}), we prove the lemma.
\end{proof}

\begin{proof}(Converse to Theorem \ref{Sem_RDT}, (Semantic Rate Distortion Coding Theorem)):

In order to prove $\lim_{n\to\infty} \mathbb{E} \left[d_s(\tilde{X}^n,\hat{\tilde{X}}^n)\right]\leq D$ for any sequence of $(2^{n(R+R_s)},n)$ codes, we must have $R\geq R_s(D)$. So we consider the following inequality,
\begin{equation}
\begin{aligned}
nR &\geq H(W_s) \geq I(W_s;X^n) \\
     &\geq I(\hat{\tilde{X}}^n(W_s);X^n)\\
     &\overset{(a)}{\geq} I_s(\tilde{X}^n;\hat{\tilde{X}}^n)\\
     &\overset{(b)}{=} \sum_{k=1}^{n} I_s(\tilde{X}_k;\hat{\tilde{X}}_k)\\
     &\overset{(c)}{\geq} \sum_{k=1}^{n} R_s\left(\mathbb{E}\left[d_s(\tilde{X}_k;\hat{\tilde{X}}_k)\right]\right)\\
     &\overset{(d)}{\geq} n R_s\left(\mathbb{E}\left[d_s(\tilde{X}^n;\hat{\tilde{X}}^n)\right]\right),
\end{aligned}
\end{equation}
where $(a)$ follows by Lemma \ref{lemma10}, (b) follows from the property of discrete memoryless source, $(c)$ follows by the definition of $R_s(D)$, and $(d)$ follows by the convexity of $R_s(D)$. So we conclude that $R_s(D)\leq R$ for sufficiently large $n$ and complete the proof of the converse.
\end{proof}

\begin {remark}
For the conventional lossy source coding, we can use quantization, linear prediction, and transform coding to approach the rate-distortion function. By now, many efficient methods based on deep learning are applied in lossy source compression, such as convolutional neural network, transformer based network and so on. Heuristically, these new coding methods sufficiently utilize the semantic information of source and demonstrate better performance than the conventional ones. However, there is not a mature theoretic framework to design and optimize these deep learning based source coding methods. Semantic rate distortion may reveal some insights for future lossy source coding. By integrating the synonymous set into the traditional lossy source coding or neural network model, we believe that the semantic lossy source coding will provide a new solution for source compression in speech, image, and video.
\end {remark}

\section{Semantic Information Measure of Continuous Message}
\label{section_IX}

In this section, we extend the semantic information measures, such as semantic entropy, semantic mutual information to the continuous message. First we give the definitions of semantic entropy and mutual information in the continuous case. Then we investigate the capacity of Gaussian channel in the semantic sense and obtain the semantic channel capacity formula of band-limited Gaussian channel. Finally, we derive the semantic rate-distortion function for the Gaussian source.

\subsection{Semantic Entropy and Semantic Mutual Information for Continuous Message}
In order to indicate the semantic entropy in the continuous case, we first define the synonymous mapping for the continuous variable as following.

\begin{definition}
Given a continuous random variable $U$ with a probability density distribution $p(u)$, $u\in \Omega$, the associated discrete semantic variable is $\tilde{U}$, the synonymous mapping between $\tilde{U}$ and $U$ is defined as
\begin{equation}
f: \tilde{\mathcal{U}}\to \Omega,
\end{equation}
where $\tilde{\mathcal{U}}=\left\{\tilde{u}_{i_s}\right\}_{i_s=1}^{\tilde{N}}$ is the semantic alphabet and $\Omega=\bigcup_{i_s=1}^{\tilde{N}}\Omega_{i_s}$ is the support set of the random variable with $\forall i_s\neq j_s, \Omega_{i_s}\bigcap \Omega_{j_s}=\varnothing$.
Specifically, for any $\tilde{u}_{i_s} \in \tilde{\mathcal{U}}$, it can be mapped into a subset $\Omega_{i_s}\subset \Omega$. Hence under the mapping $f$, the support set is partitioned into a series of synonymous intervals and $\left|\Omega_{i_s}\right|=L_{i_s}$ is named as the synonymous length of the $i_s$-th interval. 
\end{definition}

\begin{definition}
Given a continuous random variable $U$ with a probability density function $p(u)$, $u\in \Omega$, under a synonymous mapping $f$, for the associated semantic variable $\tilde{U}$, the semantic entropy is defined as
\begin{equation}\label{conti_semantic_entropy}
H_s(\tilde{U})=-\int_{\Omega} p(u)\log p(u)du-\mathbb{E}(\log L),
\end{equation}
where $\mathbb{E}(\log L)=\sum_{i_s=-\infty}^{\infty}\int_{\Omega_{i_s}}p(u)du \log L_{i_s}$ is the expectation of logarithm of synonymous interval length.
\end{definition}

We now illustrate the relationship between the continuous semantic entropy and the discrete counterpart. Let $L_{[-\infty:i_s]}=\sum_{j_s=-\infty}^{i_s}L_{j_s}$. Suppose the range of $U$ is divided into synonymous intervals $\Omega_{i_s}=\left[L_{[-\infty:i_s-1]},L_{[-\infty:i_s]}\right]$ of length $L_{i_s}$ under the synonymous mapping $f$. Furthermore, each interval is divided into bins of length $\Delta$ so that $L_{i_s}=J_{i_s}\Delta$. Assume the density function is continuous in bins, by the integration mean value theorem, there is a value $u_{J_{[-\infty:i_s-1]}+j}$ such that
\begin{equation}
p(u_{J_{[-\infty:i_s-1]}+j})\Delta=\int_{L_{[-\infty:i_s]}+j\Delta}^{L_{[-\infty:i_s]}+(j+1)\Delta} p(u)du.
\end{equation}
Considering the synonymous mapping $f$, we can write
\begin{equation}
\sum_{j=0}^{J_{i_s}-1} p(u_{J_{[-\infty:i_s-1]}+j})\Delta =p(\tilde{u}_{i_s})J_{i_s}\Delta.
\end{equation}

So the semantic entropy of the quantized version is
\begin{equation}
\begin{aligned}
H_s(\tilde{U}^{\Delta})&=-\sum_{i_s=-\infty}^{\infty} p(\tilde{u}_{i_s})J_{i_s}\Delta \log \left[p(\tilde{u}_{i_s})J_{i_s}\Delta\right]\\
&=-\sum_{i_s=-\infty}^{\infty}\sum_{j=0}^{J_{i_s}-1} p(u_{J_{[-\infty:i_s-1]}+j})\Delta \log \left[\sum_{j=0}^{J_{i_s}-1} p(u_{J_{[-\infty:i_s-1]}+j})\Delta\right]\\
&\underset{\Delta \to 0}{\mathop{\Rightarrow }} -\sum_{i_s=-\infty}^{\infty}\int_{L_{[-\infty:i_s-1]}}^{L_{[-\infty:i_s]}} p(u) \log \left[\int_{L_{[-\infty:i_s-1]}}^{L_{[-\infty:i_s]}} p(v)dv\right]du,
\end{aligned}
\end{equation}
where the last equality is from the Riemann integrability.

By using the integration mean value theorem, there is a value $\zeta$ in the interval $\left[L_{[-\infty:i_s-1]},L_{[-\infty:i_s]}\right]$ such that $\int_{L_{[-\infty:i_s-1]}}^{L_{[-\infty:i_s]}} p(v)dv=L_{i_s}p(\zeta) $. So the quantized semantic entropy can be further approximated as
\begin{equation}
\begin{aligned}
&-\sum_{i_s=-\infty}^{\infty}\int_{L_{[-\infty:i_s-1]}}^{L_{[-\infty:i_s]}} p(u) \log \left[\int_{L_{[-\infty:i_s-1]}}^{L_{[-\infty:i_s]}} p(v)dv\right]du \\
&= -\sum_{i_s=-\infty}^{\infty}\int_{L_{[-\infty:i_s-1]}}^{L_{[-\infty:i_s]}} p(u) \log \left[L_{i_s} p(\zeta)\right]du \\
&\approx -\sum_{i_s=-\infty}^{\infty}\int_{\Omega_{i_s}} p(u) \log \left[L_{i_s} p(u)\right]du \\
&=-\int_{-\infty}^{\infty} p(u)\log p(u) du -\sum_{i_s=-\infty}^{\infty}\left[ \int_{\Omega_{i_s}} p(u) du \right] \log L_{i_s} =H_s(\tilde{U}).
\end{aligned}
\end{equation}

Although Eq. (\ref{conti_semantic_entropy}) is an approximation form, it has a concise expression and clear physical meaning. So we use this formula to present the semantic entropy in the continuous case.

\begin{corollary}\label{SemEntropy_LB}
Given a continuous random variable $U$ with a probability density function $p(u)$ and the associated semantic variable $\tilde{U}$ under a synonymous mapping $f$, the semantic entropy is lower bounded by
\begin{equation}\label{LB_semantic_entropy}
H_s(\tilde{U})\geq -\int_{\Omega} p(u)\log p(u)du-\log S,
\end{equation}
where $S=\mathbb{E}( L)$ is the average length of synonymous interval. The equality holds when the optimal mapping $f$ is a proportional partition based on the probability of synonymous interval.
\end{corollary}
\begin{proof}
The semantic entropy can be written as
\begin{equation}
H_s(\tilde{U})=-\int_{\Omega} p(u)\log p(u) du -\sum_{i_s=1}^{\tilde{N}}\left[ \int_{\Omega_{i_s}} p(u) du \right] \log \frac{L_{i_s}}{|\Omega|}-\log |\Omega|.
\end{equation}

Let $\left[ \int_{\Omega_{i_s}} p(u) du \right]=p_{s,i_s}$ and $q_{s,i_s}=\frac{L_{i_s}}{|\Omega|}$. Since $D(p_s||q_s)\geq 0$, we have
\begin{equation}
-\sum_{i_s=1}^{\tilde{N}}\left[ \int_{\Omega_{i_s}} p(u) du \right] \log \frac{L_{i_s}}{|\Omega|}\geq -\sum_{i_s=1}^{\tilde{N}}\left[ \int_{\Omega_{i_s}} p(u) du \right] \log \left[ \int_{\Omega_{i_s}} p(u) du \right].
\end{equation}
The equality holds when the following condition is satisfied
\begin{equation}\label{condition1}
\int_{\Omega_{i_s}} p(u) du=\frac{L^{*}_{i_s}}{|\Omega|},  i_s=1,2,\cdots, \tilde{N}.
\end{equation}
This condition means that the synonymous length $L^{*}_{i_s}$ of interval $\Omega_{i_s}$ is proportional to the probability of synonymous interval $\int_{\Omega_{i_s}} p(u) du $ for the optimal mapping $f$.



Furthermore, by using Jensen's inequality, we can derive that
\begin{equation}
\begin{aligned}
H_s(\tilde{U})&\geq -\int_{\Omega} p(u)\log p(u) du -\sum_{i_s=1}^{\tilde{N}}\left[ \int_{\Omega_{i_s}} p(u) du \right] \log L^{*}_{i_s}\\
&\geq -\int_{\Omega} p(u)\log p(u) du -\log \sum_{i_s=1}^{\tilde{N}}\left[ \int_{\Omega_{i_s}} p(u) du \right] L^{*}_{i_s}\\
&= -\int_{\Omega} p(u)\log p(u) du -\log S,
\end{aligned}
\end{equation}
where $S$ is the average length of synonymous interval, which is defined as
\begin{equation}
S=\mathbb{E}( L)=\sum_{i_s=1}^{\tilde{N}}\left[ \int_{\Omega_{i_s}} p(u) du \right] L^{*}_{i_s}.
\end{equation}
So we complete the proof.
\end{proof}

In fact, under the synonymous mapping, if $S\to 0$, $H_s(\tilde{U})=-\int_{-\infty}^{\infty} p(u)\log p(u) du -\log S\to \infty$, the first term is the differential entropy of $U$, that is, $h(U)=-\int_{-\infty}^{\infty} p(u)\log p(u) du$. Specifically, if $S=1$, then $H_s(\tilde{U})=h(U)$, that is, the semantic entropy is equal to the differential entropy. On the other hand, if $S\to \infty$, then $H_s(\tilde{U})\to -\infty$. That means if the synonymous length goes sufficiently large we can obtain no extra semantic information.

\begin{remark}
The average synonymous length $S$ indicates the identification ability of information. If $S=1$, the semantic variable $\tilde{U}$ obtains the same identification result as the random variable $U$. On the other hand, if $S>1$, the former loses some identification ability and attains a smaller semantic entropy than the latter.
\end{remark}

We now give the definition of joint/conditional synonymous mapping as following.
\begin{definition}
Given a continuous random variable pair $(U,V)$ with a probability density distribution $p(u,v)$, $u\in \Omega_u, v\in \Omega_v$, the associated discrete semantic variable pair is $(\tilde{U},\tilde{V})$, the joint synonymous mapping is defined as
\begin{equation}
f_{uv}: \tilde{\mathcal{U}}\times \tilde{\mathcal{V}} \to \Omega_u \times \Omega_v,
\end{equation}
where $(\tilde{\mathcal{U}},\tilde{\mathcal{V}})=\left\{(\tilde{u}_{i_s},\tilde{v}_{j_s})\right\}_{i_s,j_s=1}^{\tilde{N_u},\tilde{N_v}}$ is the semantic alphabet and $\forall i_s,j_s, \left|\Omega_{(i_s,j_s)}\right|=L_{i_s}L_{j_s}$.
Here $L_{i_s}$ and $L_{j_s}$ are the synonymous lengths.
Similarly, the conditional synonymous mapping is defined as
\begin{equation}
f_{v|u}: \tilde{\mathcal{V}}|U  \to   \Omega_v|U,
\end{equation}
where for all $ j_s, \left|\Omega_{j_s}|U\right|=L_{j_s}$.
\end{definition}

Thus we give the definitions of semantic conditional entropy and semantic joint entropy as following.
\begin{definition}
Given a pair of semantic variables $(\tilde {U},\tilde {V})$ and the associated continuous random variable pairs $\left(U,V\right)$ with a joint density function $p(u,v)$, under a joint mapping $f_{uv}$, the semantic joint entropy $H_s(\tilde {U},\tilde {V})$ is defined as
\begin{equation}
H_s(\tilde{U},\tilde{V})=-\int_{\Omega_v}\int_{\Omega_u} p(u,v)\log p(u,v)dudv-\mathbb{E}\left[\log (L_u L_v)\right],
\end{equation}
where $L_u$ and $L_v$ are the synonymous lengths of random variables $U$ and $V$ respectively.

Correspondingly, under an conditional mapping $f_{v|u}$, the semantic conditional entropy $H_s(\tilde {V}|U)$ is defined as
\begin{equation}
H_s(\tilde{V}|U)=-\int_{\Omega_v}\int_{\Omega_u} p(u,v)\log p(v|u)dudv-\mathbb{E}(\log L_v),
\end{equation}
where $L_v$ is the synonymous length of random variable $V$.
\end{definition}

\begin{example}
Consider a random variable $U$ with the uniform distribution $p(u)=\frac{1}{b-a},u\in[a,b]$. Assume the synonymous mapping $f$ evenly partitions the interval $[a,b]$ into $\tilde{N}$ parts, so the synonymous length is $S=\frac{b-a}{\tilde{N}}$. Hence the semantic entropy of the associated variable $\tilde{U}$ is
\begin{equation}
H_s(\tilde{U})=-\int_{a}^{b}\frac{1}{b-a}\log\frac{1}{b-a}du-\log \frac{b-a}{\tilde{N}}= \log{\tilde{N}} \text{ sebits}.
\end{equation}

\end{example}

\begin{example}
Let $U$ denote the Gaussian random variable with the density function $p(u)=\frac{1}{\sqrt{2\pi\sigma^2}}e^{-\frac{u^2}{2\sigma^2}}$. Under a proportional-partition mapping $f$ with the synonymous length $S$, the semantic entropy is written as
\begin{equation}
\begin{aligned}
H_s(\tilde{U})&=\mathbb{E} [-\log p(U)] -\log S \\
&=\mathbb{E} \left(\log \sqrt{2\pi\sigma^2} + \left(\log e\right)\frac{U^2}{2\sigma^2}\right) -\log S \\
&=\frac{1}{2}\log {2\pi\sigma^2}+ \left(\log e\right)\frac{\mathbb {E} U^2}{2\sigma^2} -\log S \\
&= \frac{1}{2}\log \frac{2\pi e \sigma^2}{S^2} \text{ sebits}.
\end{aligned}
\end{equation}
\end{example}

We now give the definition of up/down semantic mutual information as follows.
\begin{definition}
Given a pair of semantic variables $(\tilde {U},\tilde {V})$ and the associated continuous random variable pair $\left(U,V\right)$ with a joint density function $p(u,v)$, under a joint mapping $f_{uv}$, the up semantic mutual information $I^s(\tilde {U};\tilde {V})$ is defined as
\begin{equation}\label{up_SMI_conti}
\begin{aligned}
I^s(\tilde{U};\tilde{V})&=H(U)+H(V)-H_s(\tilde{U},\tilde{V})\\
                                    &=-\int_{\Omega_v}\int_{\Omega_u} p(u,v)\log \frac{p(u)p(v)}{p(u,v)}dudv+\mathbb{E}\left[\log (L_u L_v)\right],
\end{aligned}
\end{equation}
where $L_u$ and $L_v$ is the synonymous length of random variable $U$ and $V$.
Similarly, the down semantic mutual information $I_s(\tilde {U};\tilde {V})$ is defined as
\begin{equation}\label{down_SMI_conti}
\begin{aligned}
I_s(\tilde{U};\tilde{V})&=H_s(\tilde{U})+H_s(\tilde{V})-H(U,V)\\
                                    &=-\int_{\Omega_v}\int_{\Omega_u} p(u,v)\log \frac{p(u)p(v)}{p(u,v)}dudv-\mathbb{E}\left[\log (L_u L_v)\right].
\end{aligned}
\end{equation}
\end{definition}

Clearly, in (\ref{up_SMI_conti}) and (\ref{down_SMI_conti}), the first term is the classic mutual information. The main difference between semantic and syntactic mutual information is the logarithmic production of synonymous lengths. If $\mathbb{E}\left[\log (L_u L_v)\right]\geq 0$, we have
\begin{equation}
I_s(\tilde{U};\tilde{V}) \leq I(U;V) \leq I^s(\tilde{U};\tilde{V}).
\end{equation}
Similar to the discrete case, the down semantic mutual information may be negative. Consider the practical condition, we can set $(I_s(\tilde{U};\tilde{V}) )^{+}$.

Analog to the typical set for the discrete case, asymptotic equipartition property also holds for the continuous case and we can introduce the typical set for this case. Conceptually, the volume of continuous typical set for the semantic variable can be approximated as $2^{nH_s(\tilde{U})}=\frac{2^{nH(U)}}{S^n}$. So the synonymous length $S$ can be interpreted as the reduced proportion in each side length. Similar interpretation can be applied for the semantic conditional/joint entropy and mutual information. Due to the page length limitation, we will not discuss the details of the continuous typical set in the semantic sense.

\begin{remark}
In the algorithm of signal detection and estimation, we often make a decision based on observations in an interval. This process can be modeled as a synonymous mapping. Thus, we can handle some problems in the radar signal detection, hypothesis testing, integrated sensing and communication from the viewpoint of semantic information processing. Therefore, semantic information theory may provide a theoretical explanation and establish the fundamental limits for these signal processing problems.
\end{remark}
\subsection{Semantic Capacity of Gaussian Channel}
Consider a Gaussian channel model, the received signal $y_k$ at time $k$ can be written as
\begin{equation}
y_k=x_k+z_k,
\end{equation}
where $z_k\sim \mathcal{N}(0,\sigma^2)$ is the noise sample drawn i.i.d. from a Gaussian distribution with variance $\sigma^2$. Under a joint synonymous mapping $f_{xy}$, the associated semantic variables are $\tilde{Y}$, $\tilde{X}$ and $\tilde{Z}$ respectively. Given the average power constraint $\mathbb{E}X^2\leq P$, by using Jensen's inequality, the semantic capacity of the Gaussian channel can be derived as
\begin{equation}
\begin{aligned}
C_s&=\max_{f_{xy}}\max_{\left\{p(x): \mathbb{E}X^2\leq P\right\}} I^s(\tilde{X};\tilde{Y})\\
      &=\max_{p(x);f_{xy}} H(X)+H(Y)-H_s(\tilde{X},\tilde{Y}) \\
      &=\max_{p(x)} H(X)+H(Y)-H(X,Y)+\max_{f_{xy}}\mathbb{E}\left[\log (L_xL_y)\right]   \\
      &=\frac{1}{2}\log\left(1+\frac{P}{\sigma^2}\right)+\log \left[\mathbb{E}(L_x) \mathbb{E}(L_y)\right]\\
      &=\frac{1}{2}\log\left(1+\frac{P}{\sigma^2}\right)+\log(S^2),
\end{aligned}
\end{equation}
where we assume $S_x=\mathbb{E}(L_x)=S_y=\mathbb{E}(L_y)=S$ and $S\geq 1$.

\begin{remark}
Like the classic information theory, the semantic capacity of Gaussian channel is achieved when $X\sim \mathcal{N}(0,P)$ and the synonymous mapping $f_{xy}$ is a proportional-partition mapping. Specifically speaking, if the transmitted variable $X$ (the received signal $Y$) is partitioned based on an equiprobability mapping, the semantic capacity may be achieved.
\end{remark}

Furthermore, we can also obtain a lower bound of the semantic capacity, that is
\begin{equation}
\begin{aligned}
C_s=\frac{1}{2}\log\left(\frac{P+\sigma^2}{\frac{\sigma^2}{S^4}}\right) & \geq \frac{1}{2}\log\left(\frac{P+\frac{\sigma^2}{S^4}}{\frac{\sigma^2}{S^4}}\right)\\
&=\frac{1}{2}\log \left(1+S^4\frac{P}{\sigma^2}\right)=\underline{C}_s.
\end{aligned}
\end{equation}

\begin{theorem}
Given a Gaussian channel with power constraint $P$, noise variance $\sigma^2$ and average synonymous length $S\geq 1$, the achievable semantic rate is
\begin{equation}
C_s=\frac{1}{2}\log\left(1+\frac{P}{\sigma^2}\right)+\log(S^2) \text{ sebits per transmission}.
\end{equation}
Similarly, the lower bound of semantic capacity is
\begin{equation}
\underline{C}_s=\frac{1}{2}\log\left(1+S^4\frac{P}{\sigma^2}\right) \text{ sebits per transmission}.
\end{equation}
\end{theorem}

By using random coding, synonymous mapping, and joint typicality decoding in the continuous case, we can prove this theorem. The details is omitted due to the limitation of page length.

\begin{remark}
We now give a geometric interpretation for this theorem as shown in Fig. \ref{Semantic_Sphere_packing}. Given signal power $P$ and noise variance $\sigma^2$, and a codeword of length $n$, for the classic channel coding, the transmitted codeword is normally scattered in a sphere of radius $\sqrt{nP}$ and the decoding region of each codeword is confined to a sphere of radius $\sqrt{n\sigma^2}$ (labeled by red circle) with high probability. Since the energy of received vectors is no more than $\sqrt{n(P+\sigma^2)}$, the volume of received vector sphere is $A_n(n(P+\sigma^2))^{\frac{n}{2}}$ with the radius $\sqrt{n(P+\sigma^2)}$. Therefore, in order to avoid the intersection of the decoding region, the maximum number of decoding spheres in this volume is limited as
\begin{equation}
\frac{A_n(n(P+\sigma^2))^{\frac{n}{2}}}{A_n(n\sigma^2)^{\frac{n}{2}}}=2^{\frac{n}{2}\log (1+\frac{P}{\sigma^2})}.
\end{equation}
In the syntactic sense, such a maximum number of codewords with no error probability decoding cannot be surpassed.

\end{remark}
\begin{figure}[htbp]
\setlength{\abovecaptionskip}{0.cm}
\setlength{\belowcaptionskip}{-0.cm}
  \centering{\includegraphics[scale=0.8]{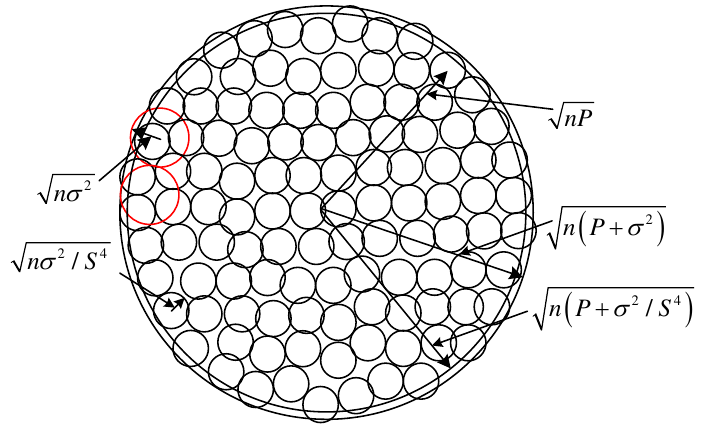}}
  \caption{Semantic Sphere packing for the Gaussian channel.}\label{Semantic_Sphere_packing}
\end{figure}

On the contrary, in the semantic sense, thanks to the synonymous mapping, the radius of a decoding sphere is further reduced to $\sqrt{n\sigma^2/S^4}$. Due to the radius decreasing, the semantic decoding sphere (labeled by black circle) has a smaller volume $A_n(n\sigma^2/S^4)^{\frac{n}{2}}$. Hence, the maximum number of semantic decoding spheres in received vector volume is limited as
\begin{equation}
\frac{A_n(n(P+\sigma^2))^{\frac{n}{2}}}{A_n(n\sigma^2/S^4)^{\frac{n}{2}}}=2^{\frac{n}{2}\log \left(S^4\left(1+\frac{P}{\sigma^2}\right)\right)}.
\end{equation}
Furthermore, if we consider a conservative estimate for the received vector volume, we can limit the energy of the received vector as $\sqrt{n(P+\sigma^2/S^4)}$. So the number of semantic decoding spheres is bounded by
\begin{equation}
\frac{A_n(n(P+\sigma^2/S^4))^{\frac{n}{2}}}{A_n(n\sigma^2/S^4)^{\frac{n}{2}}}=2^{\frac{n}{2}\log \left(1+S^4\frac{P}{\sigma^2}\right)}.
\end{equation}
Compare with the classic coding in Gaussian channel, when the synonymous length is equal to one, the semantic capacity is the same as the classic capacity. However, with the growth of synonymous length, we find that semantic coding based on the synonymous mapping can further reduce the uncertainty range of decoding region so that the volume of decoding sphere can be decreased. Therefore, semantic channel coding can pack larger number of codewords in the same volume of received vector sphere than the traditional method. In this sense, the capacity of Gaussian channel can be further improved by using semantic coding.

\subsection{Semantic Capacity of Band-limited Gaussian Channel}
We now investigate the semantic capacity of band-limited Gaussian channel. Given a Gaussian noise channel with limited bandwidth $B$ and two-sided power spectrum $N_0/2$, we transmit signals on this channel with a limited time interval $[0,T]$ and limited power $P$. So Shannon's channel capacity formula is written as
\begin{equation}
C=B\log\left(1+\frac{P}{N_0B}\right) \text{ bits per second.}
\end{equation}

\begin{theorem}
If a signal $S(t)$ in a time interval $[0,T]$ with the power constraint $P$ is transmitted on the Gaussian noise channel with limited bandwidth $B$, under a synonymous mapping $f$ with the average synonymous length $S\geq 1$, the semantic channel capacity can be written by
\begin{equation}
\begin{aligned}
C_s&=B\log\left[S^4\left(1+\frac{P}{N_0B}\right)\right] \\
      &=B\log \left(1+\frac{P}{N_0B}\right)+4B\log S  \text{ sebits per second.}
\end{aligned}
\end{equation}
Correspondingly, the lower bound of semantic capacity can be written by
\begin{equation}
\underline{C}_s=B\log \left(1+S^4\frac{P}{N_0B}\right)  \text{ sebits per second.}
\end{equation}
\end{theorem}
\begin{proof}
By using the Shannon-Nyquist sampling theorem, the signal $S(t)$ is decomposed into a series of i.i.d. samples. Hence the semantic capacity per sample is $\frac{1}{2}\log\left[S^4\left(1+\frac{P}{N_0B}\right)\right]$. Since there are $2B$ samples each second, the semantic capacity of the channel can be rewritten as $C_s=B\log\left[S^4\left(1+\frac{P}{N_0B}\right)\right]$. A similar method can also be applied to derive the lower bound of semantic capacity. So we complete the proof.
\end{proof}

Figure \ref{Channel_capacity_comp} depicts the comparison of semantic and classic capacity of band-limited Gaussian channel at various bit signal-to-noise ratios ($E_b/N_0$). For the semantic cases, we draw the semantic capacity and the lower bound for different synonymous lengths, such as $S=8$, $S=4$, and $S=2$. We can see that the semantic capacity of Gaussian channel is significantly larger than the classic capacity due to using the synonymous mapping and semantic coding. With the increasing of average synonymous length, the improvement of capacity will become more noticeable.

\begin{figure*}[htbp]
\setlength{\abovecaptionskip}{0.cm}
\setlength{\belowcaptionskip}{-0.cm}
  \centering{\includegraphics[scale=0.8]{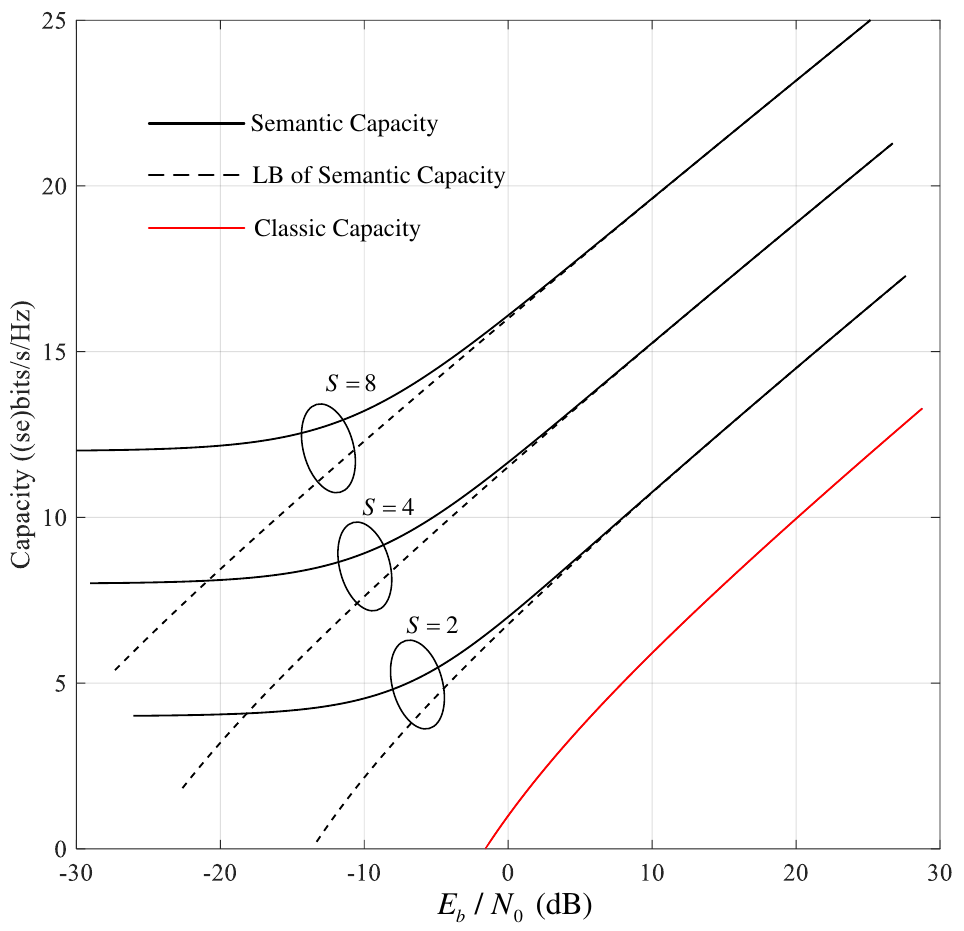}}
  \caption{Semantic and classic capacity of Gaussian channel.}\label{Channel_capacity_comp}
\end{figure*}

If the bandwidth tend to infinity, for the lower bound of semantic channel capacity, we obtain the limitation as following,
\begin{equation}
\lim_{B\to \infty} \underline{C}_s = \frac{S^4}{\ln2}\cdot\frac{ P}{ N_0} \approx 1.44S^4\frac{ P}{ N_0}.
\end{equation}
On the other hand, let $\eta=\frac{\underline{C}_s}{B}$ be the spectrum efficiency. When $\eta\to0$. we obtain the limitation of $E_b/N_0$ as
\begin{equation}
\lim_{\eta\to0}\frac{E_b}{N_0}=\frac{\ln 2}{S^4}\approx \frac{0.693}{S^4}.
\end{equation}

Next, we explore the minimum energy per sebit. Let $P=ER$ be the signal power and $\mu=\frac{R}{B}$ be the spectrum efficiency. Since $R\leq B\log(1+S^4\frac{ER}{N_0 B})$, we can derive the minimum energy needed for semantic communication at spectrum efficiency $\mu$, that is,
\begin{equation}
E(\mu)=\frac{1}{S^4\mu}(2^{\mu}-1).
\end{equation}
Here we assume the Gaussian noise power $N_0=1$. Figure \ref{Energyl_comp} shows the minimum energy needed for semantic and classic communication under various spectrum efficiency with the average synonymous lengths $S=2$ and $S=4$. Compared with the classic communication, we can observe that the minimum energy of semantic communication is dramatically reduced. So it follows that semantic communication may be an important method to implement the green communication.

\begin{figure*}[htbp]
\setlength{\abovecaptionskip}{0.cm}
\setlength{\belowcaptionskip}{-0.cm}
  \centering{\includegraphics[scale=0.8]{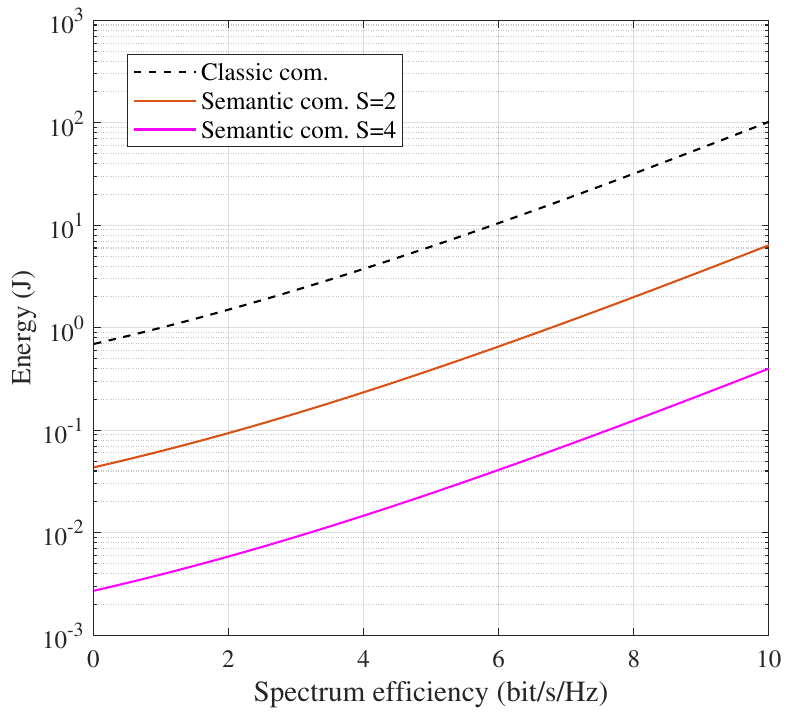}}
  \caption{Minimum energy versus spectrum efficiency of semantic and classic communication.}\label{Energyl_comp}
\end{figure*}

All these theoretic results show an extraordinary advantage of semantic coding over the classic coding. They reveal the tremendous potential of semantic channel coding in the communication application.

\subsection{Semantic Rate Distortion of Gaussian Source}
As a dual problem, we now consider the semantic rate-distortion of Gaussian source and have the following theorem.

\begin{theorem}
Given a Gaussian source $X\sim \mathcal{N}(0,P)$ and the reconstruction signal $\hat{X}$, under the synonymous mapping $f_x$ and $f_{\hat{x}}$, the associated semantic variables are $\tilde{X}$ and $\hat{\tilde{X}}$ respectively, with the mean squared error (MSE) distortion $\mathbb{E}[(\tilde{X}-\hat{\tilde{X}})^2]\leq D$, the signal model can be written as
\begin{equation}
\tilde{X}=\hat{\tilde{X}}+Z,
\end{equation}
where $Z$ is the noise sample drawn i.i.d. from a Gaussian distribution with variance $D$. So the semantic rate-distortion of the Gaussian source is
\begin{equation}
R_s(D)=\left\{
\begin{aligned}
&\frac{1}{2}\log \left( \frac{P}{S^4 D}\right), &0\leq D \leq \frac{P}{S^4},\\
&0,  &D > \frac{P}{S^4}.
\end{aligned}\right.
\end{equation}
where $S$ is the average synonymous length.
\end{theorem}

\begin{proof}
For the Gaussian source $X\sim \mathcal{N}(0,P)$, by using Corollary \ref{SemEntropy_LB}, we can write the down semantic mutual information as
\begin{equation}
\begin{aligned}
I_s(\tilde{X};\hat{\tilde{X}})&= H_s(\tilde{X})+H_s(\hat{\tilde{X}})-H(X,\hat{X}) \\
                                            &= h(X)-\mathbb{E}\log(L_x)+h(\hat{X})- \mathbb{E}\log(L_{\hat{x}})-H(X,\hat{X})   \\
                                            &\geq I(X;\hat{X})-\log\mathbb{E}(L_x)-\log\mathbb{E}(L_{\hat{x}})  \\
                                            &\geq\frac{1}{2}\log\left(\frac{P}{D}\right)-\log(S_xS_y)\\
                                            &= \frac{1}{2}\log\left(\frac{P}{D}\right)-\log(S^2)\\
                                            &= \frac{1}{2}\log\left(\frac{P}{S^4 D}\right),
\end{aligned}
\end{equation}
where we assume $S_x=\mathbb{E}(L_x)=S_{\hat{x}}=\mathbb{E}(L_{\hat{x}}) =S$ and $S\geq 1$. Like the classic information theory, this semantic rate-distortion function is achieved when $Z\sim \mathcal{N}(0,D)$ and $f_{x}$ ($f_{\hat{x}}$) is an equiprobability partition synonymous mapping. If $0\leq D \leq  P/S^4$, $R_s(D)\geq 0$, otherwise $R_s(D)=0$.
\end{proof}

We now give a geometric interpretation for this theorem as shown in Fig. \ref{Semantic_Sphere_covering}. For the classic lossy source code, we should use a group of encoding spheres of radius $\sqrt{nD}$ to cover the source volume of radius $\sqrt{nP}$. So the minimum number of the source codewords required is $2^{nR(D)}=\left(\frac{P}{D}\right)^{n/2}$. On the contrary, due to the synonymous mappings for the source and reconstruction sequence, the equivalent volume of source space can be reduced to a sphere of radius $\sqrt{nP/S^4}$ so that the minimum number of codewords is $2^{nR_s(D)}=\left(\frac{P/S^4}{D}\right)^{n/2}$.

\begin{figure}[htbp]
\setlength{\abovecaptionskip}{0.cm}
\setlength{\belowcaptionskip}{-0.cm}
  \centering{\includegraphics[scale=0.8]{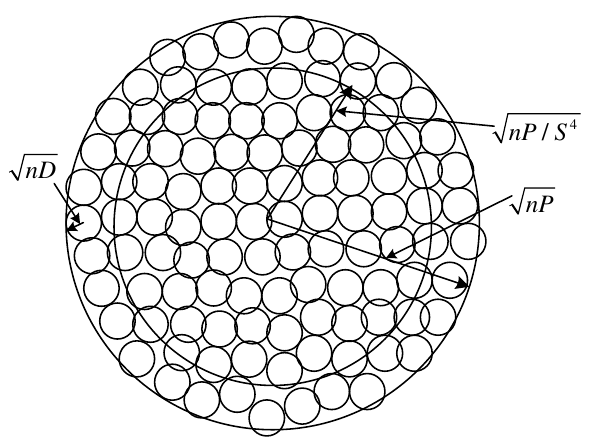}}
  \caption{Semantic Sphere covering for the Gaussian source.}\label{Semantic_Sphere_covering}
\end{figure}

Figure \ref{rate_distortion_comp} shows the semantic and syntactic rate distortion function of a Gaussian source. Here, the source signal power is $P=1$ and the synonymous lengths are set to $S=1.5$ and $S=2$ respectively. We observe that the semantic rate distortion functions dramatically decrease with the growth of synonymous length and become significantly lower than the classic counterparts. These results manifest that semantic lossy source coding has enormous potential for the source compression in the future.

\begin{figure*}[htbp]
\setlength{\abovecaptionskip}{0.cm}
\setlength{\belowcaptionskip}{-0.cm}
  \centering{\includegraphics[scale=0.8]{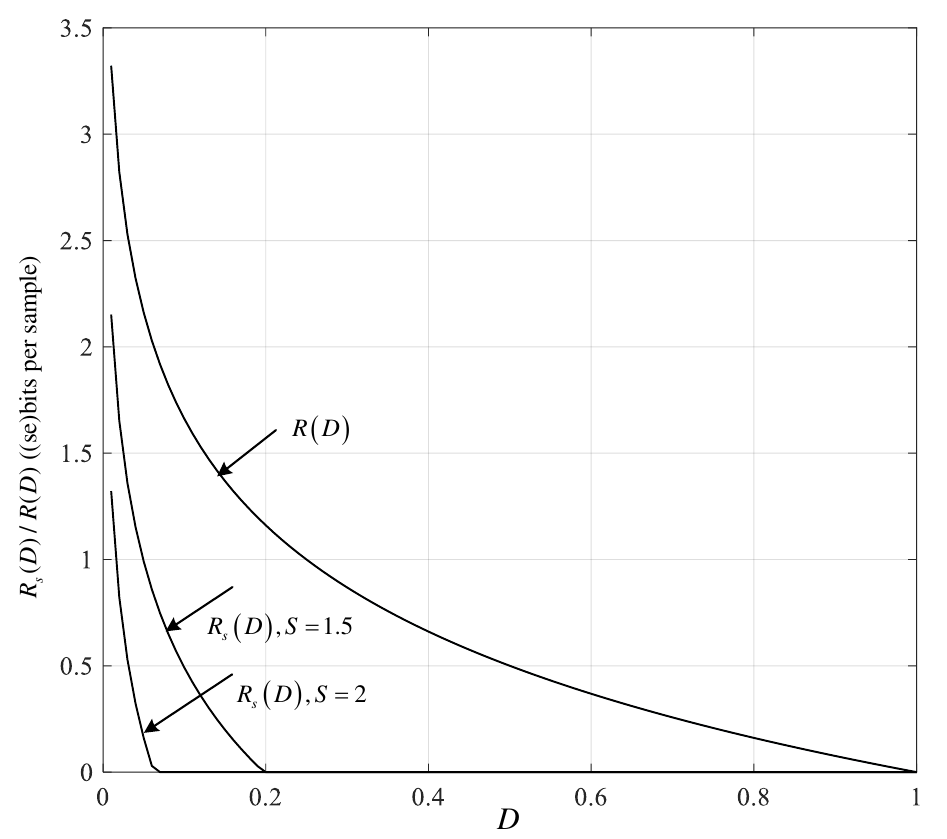}}
  \caption{Rate distortion function of Gaussian source.}\label{rate_distortion_comp}
\end{figure*}

\section{Semantic Joint Source Channel Coding}
\label{section_X}
In this section, we consider the semantic joint source channel coding. Similar to the classic information theory, we can tie together two basic methods of semantic communication: semantic source coding and semantic channel coding. Figure \ref{Semantic_JSCC} gives the system model of semantic joint source channel coding. Let $\tilde{U}$ denote the discrete memoryless semantic source with entropy $H_s(\tilde{U})$ and $d_s(\tilde{u},\hat{\tilde{u}})$ be a semantic measure with rate-distortion function $R_s(D)$. Furthermore, let $X$ denote the coded symbol after synonymous mapping $f$ and $\left\{\tilde{\mathcal{X}}(\tilde{\mathcal{U}}),\mathcal{X},\mathcal{Y},\tilde{\mathcal{Y}}(\hat{\tilde{\mathcal{U}}}), p(Y|X)\right\}$ be a discrete memoryless channel with semantic capacity $C_s$.

\begin{figure*}[htbp]
\setlength{\abovecaptionskip}{0.cm}
\setlength{\belowcaptionskip}{-0.cm}
  \centering{\includegraphics[scale=0.8]{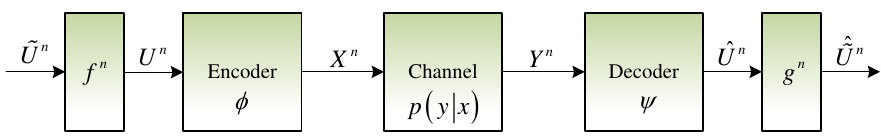}}
  \caption{Semantic joint source channel coding.}\label{Semantic_JSCC}
\end{figure*}

\begin{theorem}\label{Sem_SCCT_Lossy}
(Semantic Source Channel Coding Theorem in Lossy case):

Given a discrete memoryless source $U\sim p(u)$ and a discrete memoryless channel $p(y|x)$, the source is associated with a semantic variable $\tilde{U}$ under the synonymous mapping $f$ and the codeword $X^n$ is transmitted on the channel. In the case of lossy transmission, if the code rate satisfies $R_s(D)<R<C_s$, there exists a sequence of $\left(2^{nR},n\right)$ semantic source channel codes, when code length $n$ tends to sufficiently large, then semantic distortion satisfies $\mathbb{E}d_s(\tilde{X},\hat{\tilde{X}})<D$.

On the contrary, if $R_s(D)>C_s$, then for any $\left(2^{nR},n\right)$ code, with sufficiently large $n$, the semantic distortion meets $\mathbb{E}d_s(\tilde{X},\hat{\tilde{X}})>D$.
\end{theorem}

\begin{proof}
We use separate semantic lossy source coding and semantic channel coding to prove the achievability.

(1) Source coding: Give any $\epsilon>0$, there exists a sequence of semantic lossy source codes with rate $R\geq R_s(D)+\epsilon$ that satisfies $\mathbb{E}d_s(\tilde{u},\hat{\tilde{u}})<D$. The index of each code can be regarded as a semantic message to be sent to the channel.

(2) Channel coding: The source indices can be encoded and reliably transmitted over the channel if $R\leq C_s-\epsilon$.

In the channel decoder, when $n\to \infty$, the error probability $P_e^{(n)}\to 0$ so that the source decoder can reconstruct sequence and de-mapping the semantic index $\hat{i}_s$ and the average distortion meets $\mathbb{E}d_s(\tilde{X},\hat{\tilde{X}})\leq D$.

Next we prove the converse.
By the converse proof of semantic rate distortion coding theorem (Theorem \ref{Sem_RDT}), we have $R_s(D)\leq I_s(\tilde{U}^n,\hat{\tilde{U}}^n)$. Furthermore, by Corollary \ref{corollary3}, we have $I_s(\tilde{U}^n,\hat{\tilde{U}}^n)\leq I^s(\tilde{U}^n,\hat{\tilde{U}}^n)$. Then we can derive $I^s(\tilde{U}^n,\hat{\tilde{U}}^n)\leq C_s$. So we complete the proof of the theorem.
\end{proof}

\begin{theorem}
(Semantic Source Channel Coding Theorem in Lossless case):

Given a discrete memoryless source $U\sim p(u)$ and a discrete memoryless channel $p(y|x)$, the source is associated with a semantic variable $\tilde{U}$ under the synonymous mapping $f$ and the codeword $X^n$ is transmitted on the channel. In the case of lossless transmission, for each code rate $H_s(\tilde{U})<R<C_s$, there exists a sequence of $\left(2^{nR},n\right)$ semantic source channel codes, when code length $n$ tends to sufficiently large, the error probability tends to zero.

On the contrary, if $H_s(\tilde{U})>C_s$, then for any $\left(2^{nR},n\right)$ code, with sufficiently large $n$, the error probability cannot achieve arbitrarily low.
\end{theorem}

The proof of this theorem is similar to Theorem \ref{Sem_SCCT_Lossy} by setting $D=0$. In principle, the separate semantic source and channel coding is asymptotically optimal for lossless or lossy transmission. Therefore, the fundamental criteria for the semantic communication are summarized as following
\begin{equation}
\left\{
\begin{aligned}
& H_s({\tilde{U}})\leq R \leq C_s,   &\text{for lossless transmission,}\\
& R_s(D) \leq R \leq C_s,                &\text{for lossy transmission.}
\end{aligned}
\right.
\end{equation}
On the contrary, for classic communication, the code rate should be confined in $H(U)\leq R \leq C$ or $R(D)\leq R \leq C$. In these common intervals, both classic communication and semantic communication can work well. On the other hand, if $H(U)>C$ or $R(D)>C$, the classic communication cannot work yet the semantic communication still works well as long as $H_s(\tilde{U})\leq C_s$ or $R_s(D)\leq C_s$. Thus, we conclude that semantic source-channel coding can extend the range of code rates and provide new insight to improve the performance of the communication system.

\section{Conclusions}
\label{section_XI}
In this paper, we develop an information-theoretic framework of semantic communication. We start from the synonym, a fundamental property of semantic information, to build the semantic information measures including semantic entropy, up/down semantic mutual information, semantic channel capacity, and semantic rate distortion function. Then we extend the asymptotic equipartition property to the semantic sense and introduce the synonymous typical set to prove three significant coding theorems, that is, semantic source coding theorem, semantic channel coding theorem, and semantic rate distortion coding theorem. Additionally, we investigate the semantic information measures in the continuous case and derive the semantic capacity of Gaussian channel and semantic rate distortion of Gaussian source. All these works uncover the critical features of semantic communication and constitute the theoretic basis of semantic information theory.

For the theoretic analysis, the semantic information theory needs further development. In this paper, we only consider the semantic information measure and the fundamental limitation in the discrete or continuous memoryless case. In the future, we can further investigate the measure and limitation of semantic information in various memory source or channel cases, such as stationary and ergodic process (e.g. Markov process) or non-stationary non-ergodic process. Strong asymptotic equipartition property and strong typicality in the semantic sense should be further explored. On the other hand, the analysis of semantic capacity or semantic rate distortion with finite block length may also be an interesting research topic. In addition, in various multiuser communication scenarios, such as multiple access, broadcasting, relay etc., we can further analyze and derive the corresponding measure and performance limit of semantic information.

Guided by the classic information theory, in the past seventy years, the source coding and channel coding techniques have approached the theoretic limitation. On the contrary, the semantic information theory paves a new way for the coding techniques. From the viewpoint of semantic processing, with the help of synonymous mapping, the lossless source coding has much space to improve and the existing coding methods can be further modified and polished. The construction of semantic channel codes may be centered on the group Hamming distance and the optimization of decoding algorithms will be concentrated on the group decoding so that the information transmission techniques will usher in a new era that surpasses the classic limitation and approaches the semantic capacity. By the optimization of synonymous mapping, the classic lossy source coding techniques, such as vector quantization, prediction coding, and transform coding, will demonstrate new advantages to further improve the compression efficiency. Briefly, the performance bottleneck of classic communication will be broken and the traditional communication will naturally evolve to the semantic communication.

For the new coding techniques based on deep learning (DL), the semantic information theory will lift its mystery veil and provide a systematic design and optimization tool. The synonymous mapping will provide a reasonable explanation for the semantic information extracted by the deep neural network. The basic structures of mainstream DL models, such as convolutional neural networks, transformer model, variational auto-encoder and so on, may be analyzed and optimized based on the semantic information measures. Furthermore, the system architecture of semantic communication based on deep learning can be simplified or optimized guided by the semantic information theory.

In summary, the theoretic framework proposed in this paper may help understanding the essential features of semantic information and shed light on some ambiguity problems in semantic communication. We believe that the semantic information theory will uncover a new chapter of information theory and have a profound impact on many fields such as communication, signal detection and estimation, deep learning and machine learning, and integrated sensing and communication etc.

\appendix
\subsection{Proof of Lemma \ref{concave_sementropy}}\label{proof_concave_sementropy}
\begin{proof}
Suppose two probability distributions $p_1(u)$ and $p_2(u)$, for all $0 \leq \theta \leq 1$, we have $p_{\theta}(u)=\theta p_1(u)+(1-\theta) p_2(u)$. Using Jensen's inequality, we can write
\begin{equation}
\begin{aligned}
&\theta H_s(p_1(u))+(1-\theta) H_s(p_2(u))-H_s(p_{\theta}(u))\\
&=\theta \sum_{i_s}\sum_{i\in\mathcal{N}_{i_s}}p_1(u_i) \log \frac{\sum_{i\in\mathcal{N}_{i_s}}p_{\theta}(u_i)}{\sum_{i\in\mathcal{N}_{i_s}}p_1(u_i)}\\
&+(1-\theta)\sum_{i_s}\sum_{i\in\mathcal{N}_{i_s}}p_2(u_i) \log \frac{\sum_{i\in\mathcal{N}_{i_s}}p_{\theta}(u_i)}{\sum_{i\in\mathcal{N}_{i_s}}p_2(u_i)}\\
&\leq \theta \log \sum_{i_s} \sum_{i\in\mathcal{N}_{i_s}}p_{\theta}(u_i) +(1-\theta) \log \sum_{i_s} \sum_{i\in\mathcal{N}_{i_s}}p_{\theta}(u_i)\\
&=\theta \log 1 + (1-\theta) \log 1 =0.
\end{aligned}
\end{equation}
So we prove the concavity of semantic entropy.
\end{proof}

\subsection{Proof of Theorem \ref{convexity_crossentropy}}\label{proof_convexity_crossentropy}
\begin{proof}
To prove the first inequality, by using Jensen's inequality, we have
\begin{equation}
\begin{aligned}
&D_s(\theta p_{s,1} + (1-\theta) p_{s,2}\| \theta q_{s,1} + (1-\theta) q_{s,2}) \\
&- \theta D_s(p_{s,1}\|q_{s,1}) -(1-\theta) D_s(p_{s,2}\|q_{s,2})\\
&=\theta \sum_{i_s}\sum_{u_i\in \mathcal{U}_{i_s}}p_1(u_i) \\
&\cdot \log \frac{\sum_{u_i\in \mathcal{U}_{i_s}}\theta p_1(u_i) + (1-\theta) p_2(u_i)}{\sum_{u_i\in \mathcal{U}_{i_s}} \theta q_1(u_i) + (1-\theta) q_2(u_i)} \frac{\sum_{u_i\in \mathcal{U}_{i_s}} q_1(u_i)}{\sum_{u_i\in \mathcal{U}_{i_s}}p_1(u_i)}\\
&+(1-\theta) \sum_{i_s}\sum_{u_i\in \mathcal{U}_{i_s}}p_2(u_i)\\
&\cdot \log \frac{\sum_{u_i\in \mathcal{U}_{i_s}}\theta p_1(u_i) + (1-\theta) p_2(u_i)}{\sum_{u_i\in \mathcal{U}_{i_s}} \theta q_1(u_i) + (1-\theta) q_2(u_i)}\frac{\sum_{u_i\in \mathcal{U}_{i_s}} q_2(u_i)}{\sum_{u_i\in \mathcal{U}_{i_s}}p_2(u_i)}\\
&\leq \log \sum_{i_s}\sum_{u_i\in \mathcal{U}_{i_s}}\theta p_1(u_i) + (1-\theta) p_2(u_i)=\log 1=0.
\end{aligned}
\end{equation}
The other two inequalities can also be proved by using similar methods. So we prove the theorem.
\end{proof}

\subsection{Proof of Theorem \ref{SemMI_concave_convex}}\label{proof_SemMI_concabe_convex}
\begin{proof}
First, we prove $I^s(\tilde{U};\tilde{V})$ is a concave function of $p(u)$ for fixed $p(v|u)$. Since $I^s(\tilde{U};\tilde{V})=H(U)+H(V)-H_s(\tilde{U},\tilde{V})$, the entropies of $H(U)$ and $H(V)$ are concave functions of $p(u)$ for fixed $p(v|u)$. Furthermore, the semantic joint entropy $H_s(\tilde{U},\tilde{V})$ is also a concave function of $p(u)$. So we conclude that $I^s(\tilde{U};\tilde{V})$ is also a concave function of $p(u)$.

Second, we prove $I_s(\tilde{U};\tilde{V})$ is a convex function of $p(v|u)$ for fixed $p(u)$. Since $I_s(\tilde{U};\tilde{V})=D_s\left(p\left(u,v\right)\|p_s(u)p_s(v)\right)$, due to the convexity of semantic relative entropy, we conclude that $I_s(\tilde{U};\tilde{V})$ is a convex function of $p(v|u)$.

By using the similar methods, we can prove $I^s(\tilde{U};\tilde{V})$ is a convex function of $p(v|u)$ for fixed $p(u)$ and $I_s(\tilde{U};\tilde{V})$ is a concave function of $p(u)$ for fixed $p(v|u)$.
\end{proof}

\end{document}